\documentclass{lmcs}
\usepackage{amsmath}
\usepackage{amssymb}
\usepackage{bm}
\usepackage{booktabs}
\usepackage{color}
\usepackage{graphicx}

\usepackage[
   a4paper,
   pdftex,
   pdftitle={Term Orders for Optimistic Lambda-Superposition},
   pdfauthor={Alexander Bentkamp, Jasmin Blanchette, and Matthias Hetzenberger},
   pdfkeywords={},
   pdfborder={0 0 0},
   draft=false,
   bookmarksnumbered,
   bookmarks,
   bookmarksdepth=2,
   bookmarksopenlevel=2,
   bookmarksopen]{hyperref}

\allowdisplaybreaks

\DeclareFontFamily{OT1}{pzc}{}
\DeclareFontShape{OT1}{pzc}{m}{it}{<-> s * [1.10] pzcmi7t}{}
\DeclareMathAlphabet{\mathcalx}{OT1}{pzc}{m}{it}

\makeatletter
\g@addto@macro{\UrlBreaks}{\UrlOrds\do\=\do\_}
\makeatother

\def\negvthinspace{\kern-0.083333em}
\def\vthinspace{\kern+0.083333em}
\def\vvthinspace{\kern+0.0416667em}
\def\negvvthinspace{\kern-0.0416667em}

\newcommand\q{\noindent\hbox{}\quad}
\newcommand\Q{\q\llap{$\mid$~\vvthinspace}}

\newdimen\carpetH
\newdimen\carpetV
\def\carpet#1{\setbox0=\hbox{\ensuremath{#1}}%
  \kern+2\carpetH
    \raise+2\carpetV\copy0\kern-\wd0
    \raise+\carpetV\copy0\kern-\wd0
    \copy0\kern-\wd0
    \raise-\carpetV\copy0\kern-\wd0
    \raise-2\carpetV\copy0\kern-\wd0
  \kern-\carpetH
    \raise+2\carpetV\copy0\kern-\wd0
    \raise+\carpetV\copy0\kern-\wd0
    \copy0\kern-\wd0
    \raise-\carpetV\copy0\kern-\wd0
    \raise-2\carpetV\copy0\kern-\wd0
  \kern-\carpetH
    \raise+2\carpetV\copy0\kern-\wd0
    \raise+\carpetV\copy0\kern-\wd0
    \copy0\kern-\wd0
    \raise-\carpetV\copy0\kern-\wd0
    \raise-2\carpetV\copy0\kern-\wd0
  \kern-\carpetH
    \raise+2\carpetV\copy0\kern-\wd0
    \raise+\carpetV\copy0\kern-\wd0
    \copy0\kern-\wd0
    \raise-\carpetV\copy0\kern-\wd0
    \raise-2\carpetV\copy0\kern-\wd0
  \kern-\carpetH
    \raise+2\carpetV\copy0\kern-\wd0
    \raise+\carpetV\copy0\kern-\wd0
    \copy0\kern-\wd0
    \raise-\carpetV\copy0\kern-\wd0
    \raise-2\carpetV\copy0
  \kern2\carpetH
}

\newcommand\heavy[1]{\carpetH=.02ex\carpetV=.02ex\carpet{#1}}

\newcommand{\lang}{\begin{picture}(5,7)
\put(1.1,2.5){\rotatebox{45}{\line(1,0){6.0}}}
\put(1.1,2.5){\rotatebox{315}{\line(1,0){6.0}}}
\end{picture}}
\newcommand{\rang}{\begin{picture}(5,7)
\put(0,2.5){\rotatebox{135}{\line(1,0){6.0}}}
\put(0,2.5){\rotatebox{225}{\line(1,0){6.0}}}
\end{picture}}

\newcommand{\leftgreensubterm}{\lang\,}
\newcommand{\rightgreensubterm}{\rang}

\newcommand{\leftorangesubterm}{\lang\!\!\leftgreensubterm}
\newcommand{\rightorangesubterm}{\rightgreensubterm\!\!\rang}
\newcommand{\orangesubterm}[2]{#1\leftorangesubterm #2\rightorangesubterm}

\newcommand\lnf[1]{{#1}{\uparrow}_{\!\eta}}

\newcommand\cst[1]{\mathsf{#1}}

\newcommand\cal[1]{\mathcalx{#1}}
\newcommand\DB[1]{#1}
\newcommand\ENConly{\cal{E}}
\newcommand\ENC[1]{\ENConly(#1)}
\newcommand\NORM[1]{\bm{[\vthinspace}#1\bm{\vthinspace]}}
\newcommand\NORMonly{\NORM{\phantom{i}}}
\newcommand\POLY[1]{\cal{poly}(#1)}
\newcommand\LEFT{^{\>\star}}
\newcommand\RIGHT{^{\>\star\!\star}}

\newcommand\Nat{\mathbb{N}}
\newcommand\Natpos{\Nat_{>0}}
\newcommand\Ord{\mathbf{O}}
\newcommand\Ordpos{\Ord_{>0}}
\newcommand\Poly{\mathbf{P}}
\newcommand\Tm[2]{\cal T(#1,#2)}
\newcommand\GTm[1]{\Tm{#1}{\emptyset}}
\newcommand\HTm[4]{\cal T^\infty(#1,\allowbreak#2,\allowbreak#3,\allowbreak#4)}
\newcommand\HTmp[4]{\cal T^\infty_\mathrm{pre}(#1,\allowbreak#2,\allowbreak#3,\allowbreak#4)}
\newcommand\GHTmp[2]{\HTmp{#1}{#2}{\emptyset}{\emptyset}}
\newcommand\Ty[2]{\cal{Ty}(#1,\allowbreak#2)}
\newcommand\GTy[1]{\Ty{#1}{\emptyset}}

\newcommand\arityonly{\cal{arity}}
\newcommand\arity[1]{\arityonly(#1)}
\newcommand\tyarityonly{\cal{tyarity}}
\newcommand\tyarity[1]{\tyarityonly(#1)}

\newcommand\smalllam{\raisebox{\depth}{\scalebox{1}[-1.05]{\ensuremath{_\mathsf{y}}}}\negvthinspace}

\newcommand\Prec{<}
\newcommand\Succ{>}
\newcommand\Preceq{\leq}

\newcommand\Less{\prec}
\newcommand\Greater{\succ}

\newcommand\Greatereq{\succsim}
\newcommand\GreaterEQ{\succeq}
\newcommand\Lesssim{\precsim}
\newcommand\Greatersim{\succsim}

\newcommand\g{\mathsf{g}}
\newcommand\m{\mathsf{m}}
\newcommand\lex{{\smash{\mathsf{l}}\mathsf{ex}}}
\newcommand\fo{\mathsf{fo}}
\newcommand\kb{\mathsf{kbo}}
\newcommand\lkbg{\mathsf{g\smalllam kbo}}
\newcommand\lkbm{\mathsf{m\smalllam kbo}}
\newcommand\lkb{\mathsf{\smalllam kbo}}
\newcommand\lp{\mathsf{lpo}}
\newcommand\llpg{\mathsf{g\smalllam lpo}}
\newcommand\llpm{\mathsf{m\smalllam lpo}}
\newcommand\llp{\mathsf{\smalllam lpo}}
\newcommand\ltltonly{\mathrel{\Less\!\!\!\!\!\Less}}
\newcommand\gtgtonly{\mathrel{\Greater\!\!\!\!\!\Greater}}
\newcommand\ltltlex{\ltltonly^\lex}
\newcommand\gtgtlex{\gtgtonly^\lex}
\newcommand\gtgtsimlex{\mathrel{{\Greatersim}\llap{\color{white}\vrule height.2ex depth.6ex width1.7ex}\!\!\!\!\!{\Greatersim}}^\lex}
\newcommand\gtgtapproxlex{\mathrel{{\succapprox}\llap{\color{white}\vrule height.4ex depth.7ex width1.7ex}\!\!\!\!\!{\succapprox}}^\lex}
\newcommand\ty{\mathsf{ty}}
\newcommand\db{\mathsf{db}}

\begin{document}

\title{Term Orders for Optimistic Lambda-Superposition}

\author[A.~Bentkamp]{Alexander Bentkamp\lmcsorcid{0000-0002-7158-3595}}[a]
\author[J.~Blanchette]{Jasmin Blanchette\lmcsorcid{0000-0002-8367-0936}}[a]
\author[M.~Hetzenberger]{Matthias Hetzenberger\lmcsorcid{0000-0002-2052-8772}}[b]
\address{Ludwig-Maximilians-Universität München, Geschwister-Scholl-Platz 1,
80539 München, Germany}
\email{a.bentkamp@ifi.lmu.de,jasmin.blanchette@ifi.lmu.de}

\address{TU Wien Informatics,
Favoritenstraße 9--11,
1040 Vienna, Austria}
\email{matthias.hetzenberger@tuwien.ac.at}

\maketitle

\begin{abstract}
We introduce $\lambda$KBO and $\lambda$LPO, two variants of the Knuth--Bendix
order (KBO) and the lexicographic path order (LPO) designed for use with the
$\lambda$-superposition calculus. We establish the desired properties via
encodings into the familiar first-order KBO and LPO.
\end{abstract}

\section{Introduction}
\label{sec:introduction}

The $\lambda$-superposition calculus, by Bentkamp et al.\
\cite{bentkamp-et-al-2021-hosup}, is a highly competitive
\cite{sutcliffe-2021-casc} approach for proving higher-order problems
automatically. It works by saturation, performing inferences between available
clauses until the empty clause $\bot$ is derived. A clause consists of literals,
which are predicates (often equality $\approx$) applied to arguments.
Terms are equivalence classes modulo the $\alpha$-, $\beta$-, and
$\eta$-conversions of the $\lambda$-calculus. Thus, $\cst{f}$, $\lambda x.\>
\cst{f}\>x$, and $(\lambda y.\>y)\>\cst{f}$ are all considered syntactically
equal.

To break symmetries in the search space, $\lambda$-superposition uses an order
$\Greater$ on the terms. The stronger the order, the fewer clauses need to be
generated to saturate the clause set. Yet the \emph{derived higher-order orders}
used by the only implementation of $\lambda$-superposition
\cite[Sect.~3]{bentkamp-et-al-2021-hosup} are a crude encoding in terms of a
standard term order, whether the Knuth--Bendix order (KBO) or the lexicographic
path order (LPO) \cite{zantema-2003}. It is very weak in the presence of applied
variables; for example, it cannot orient the terms $y\>\cst{b}$ and
$y\>\cst{a}$, even with the precedence $\cst{b} \Succ \cst{a}$.

In this work, we introduce two stronger orders, called $\lambda$KBO and
$\lambda$LPO. As the names suggest, $\lambda$KBO and $\lambda$LPO are variants
of KBO and LPO, which are the most widely used orders with superposition
calculi. KBO compares terms by first comparing their syntactic weight, resorting
to a lexicographic comparison as a tiebreaker. LPO essentially performs a
lexicographic comparison while ensuring the subterm property (i.e., the property
that a term is larger than its proper subterms). We define three versions of
$\lambda$KBO and $\lambda$LPO, of increasing expressiveness:\ for ground
(i.e., closed) terms, monomorphic nonground (i.e., open) terms, and polymorphic
nonground terms.

The ground orders (Sect.~\ref{sec:the-ground-level}) form the first level of a
development by stepwise refinement. The monomorphic orders
(Sect.~\ref{sec:the-monomorphic-level}) add support for term variables; their
properties are justified in terms of the ground level. Similarly, the
polymorphic orders (Sect.~\ref{sec:the-polymorphic-level}) add support for type
variables on top of the monomorphic level. For $\lambda$KBO, weights are
computed as polynomials over indeterminates whose values depend on the variables
occurring in the terms. These polynomials can be compared symbolically.

Both orders are specified as a strict relation $\Greater$ and a nonstrict
relation $\Greatersim$, along the lines of Sternagel and Thiemann
\cite{sternagel-thiemann-2013}. The nonstrict orders make the comparison
$y\>\cst{b} \Greatersim y\>\cst{a}$ possible if $\cst{b} \Succ \cst{a}$. In this
example, a strict comparison would fail because $y$ could be instantiated by
$\lambda x.\>\cst{c}$, which ignores the argument and makes both terms equal.

A requirement imposed by the \emph{optimistic $\lambda$-superposition} calculus,
for which the two orders are specifically designed, is that
the order must ensure $u \Greater u\>\cst{diff}\langle\tau,\upsilon\rangle(s, t)$
for a dedicated Skolem symbol $\cst{diff}$,
for all ground terms $s,t,u$, and for all ground types $\tau$ and $\upsilon$.
This allows optimistic $\lambda$-superposition to provide special support for the
functional extensionality axiom
\[z\>(\cst{diff}\langle\alpha,\beta\rangle(\lambda\> z\>\DB{0}, \lambda\> y\>\DB{0})) \not\approx
  y\>(\cst{diff}\langle\alpha,\beta\rangle(\lambda\> z\>\DB{0}, \lambda\> y\>\DB{0})) \mathrel\lor (\lambda\> z\>\DB{0})
  \approx (\lambda\> y\>\DB{0})
\tag*{\text{(\textsc{Ext})}}
\]
Notice that the two arguments of the Skolem symbol $\cst{diff}$ are
specified in parentheses, as mandatory arguments or parameters.

\section{Preliminaries}
\label{sec:preliminaries}

We use the notation $\bar x_n$ or $\bar x$ for tuples or lists $x_1,\dots,x_n$
of length $\left|\bar x\right| = n \ge 0$. Applying a unary function~$f$ to
such a tuple applies it pointwise: $f(\bar x_n) = (f(x_1), \dots, f(x_n))$.

We write $\Nat$ for the set of natural numbers starting with 0
and $\Natpos$ for $\Nat \setminus \{0\}$.
We write $\Ord$ for the set of ordinals below $\epsilon_0$ and $\Ordpos$ for
$\Ord \setminus \{0\}$.

\subsection{Terms}
\label{ssec:terms}

We will need both untyped first-order and typed higher-order terms:
\begin{itemize}
\item Given an untyped first-order signature $\Sigma$, we write $\Tm{\Sigma}{X}$
  for the set of arity-respecting terms built using symbols from $\Sigma$ and
  the variables $X$---the ($\Sigma$-)\emph{terms}. A first-order term~$t$ is
  \emph{ground} if it contains no
  variables, or equivalently if $t \in \GTm{\Sigma}$.

\smallskip

\item For the higher order, the types and terms are those of polymorphic
  higher-order logic, as defined in Bentkamp et al.\
  \cite{bentkamp-et-al-2021-hosup}, but with a few specificities noted below.
\end{itemize}

A higher-order signature $(\Sigma_\ty, \Sigma)$ consists of a type signature
$\Sigma_\ty$ and a term signature $\Sigma$, which depends on
$\Sigma_\ty$. With each type constructor $\kappa \in \Sigma_\ty$ is associated
an arity---the number of arguments it takes. The set of types
$\Ty{\Sigma_\ty}{X_\ty}$ over $X_\ty$ is built using variables from $X_\ty$ and
type constructors applied to the expected number of arguments. The functional
type constructor $\to$ is distinguished. We abbreviate $\tau_1 \to \cdots \to
\tau_n \to \upsilon$ to $\bar\tau_n \to \upsilon$.

The first departure from Bentkamp et al.\ is that we find it convenient to
represent $\lambda$-terms using a locally nameless notation
\cite{chargueraud-2012} based on De Bruijn indices \cite{de-bruijn-1972}. This
notation is essentially isomorphic to a nominal notation, with
$\alpha$-equivalence built in. For example, $\lambda x.\,\lambda y.\, x$ will be
represented as $\lambda\, \lambda\> \DB{1}$, where the De Bruijn index $\DB{1}$
is a ``nameless dummy.'' We allow leaking De Bruijn indices---indices that
point beyond all $\lambda$-binders---but these will be ignored by substitutions.
The operator $t{\uparrow}^n$ shifts all leaking De Bruijn indices by $n$; if
omitted, $n = 1$.

The second departure from Bentkamp et al.\ is that we will use the $\eta$-long
$\beta$-normal form as representatives for $\beta\eta$-equivalence classes, or
``\emph{terms},'' where they used the $\eta$-short $\beta$-normal form. The main
advantage of $\eta$-long is that it makes it possible to obtain the desired
maximality result for the functional extensionality axiom. Moreover, since with
$\eta$-long terms of function type are always $\lambda$-abstractions, we will
find that this simplifies the arithmetic when defining the $\lambda$KBO. Given
a $\lambda$-term $t$, we will denote its $\eta$-long $\beta$-normal form as
$\lnf{t}$.

The main disadvantage of the $\eta$-long $\beta$-normal form arises with
polymorphism: Instantiating a type variable with a functional type can result in
an $\eta$-expansion, dramatically changing the term's shape. For example, if $z
: \alpha$, then $z\{\alpha \mapsto (\kappa\to\kappa)\} = \lambda\> z\> \DB{0}$.
This can be accounted for by the orders, but only at the cost of some weakening.

The third departure from Bentkamp et al.\ is that the symbols (also called
constants) may take parameters, passed in parentheses, in addition to their
regular curried arguments. These parameters do not count as subterms. This
mechanism is used for $\cst{diff}$. Parameters are supported by the optimistic
$\lambda$-superposition calculus.

We write $t : \tau$ to indicate that $t$ has type $\tau$. With each symbol
$\cst{f} \in \Sigma$ is associated a typing $\Pi\bar\alpha_m.\; \bar\tau_n
\Rightarrow \upsilon$, where $\bar\alpha_m$ is a tuple of distinct variables
that contains all type variables from $\bar\tau_n$ and $\upsilon$, $\bar\tau_n$
is the tuple of parameter types, and $\upsilon$ is the (possibly
functional) body type. Given $\cst{f}$, we let $\tyarity{\cst{f}} = m$ and
$\arity{\cst{f}} = n$. We specify a type instance by specifying a tuple
$\bar\alpha_m\sigma$ of types in angle brackets corresponding to the type
arguments: $\cst{f}\langle\bar\alpha_m\sigma\rangle : \bar\tau_n\sigma
\Rightarrow \upsilon\sigma$. Parameters are passed in parentheses.

The set of $\lambda$-preterms is built from the following expressions:
\begin{itemize}
\item a variable $x\langle\tau\rangle : \tau$ for $x\in X$ and a type $\tau$;
\item a symbol $\cst{f}\langle\bar\upsilon_m\rangle(\bar u_n) : \tau$
for a constant $\cst{f}\in\Sigma$ with type declaration $\Pi\bar{\alpha}_m.\>\bar{\tau}_n\Rightarrow\tau$, types $\bar{\upsilon}_m$,
and $\lambda$-preterms $\bar u : \bar{\tau}_n$ such that all De Bruijn indices in $\bar u$ are bound;
\item a De Bruijn index $\DB{n}\langle\tau\rangle : \tau$ for a natural number $n\geq 0$ and a type $\tau$, where $\tau$
  represents the type of the bound variable;
\item a $\lambda$-expression $\lambda\langle\tau\rangle\> t : \tau\to\upsilon$
  for a type $\tau$ and a $\lambda$-preterm $t : \upsilon$
  such that all De Bruijn indices bound by the new $\lambda\langle\tau\rangle$
  have type $\tau$;
\item an application $s\>t : \upsilon$ for
  $\lambda$-preterms $s : \tau\to\upsilon$ and $t : \tau$.
\end{itemize}
The type arguments $\langle\bar\tau\rangle$ carry enough information to enable
typing of any $\lambda$-preterm without any context. We often leave them
implicit, when they are irrelevant or can be inferred.
In $\cst{f}\langle\bar\upsilon_m\rangle(\bar u_n) : \tau$,
we call $\bar u_n$ the parameters.
We omit $()$ when a symbol has no parameters.
As a syntactic convenience, symbols
corresponding to infix operators are applied infix.
Notice that it is possible for a term to contain multiple occurrences of the same free De Bruijn index
with different types. In contrast, the types of bound De Bruijn indices always match.

A $\lambda$-term is a $\lambda$-preterm without free De Bruijn indices.

The size $|\phantom{j}|$ of a $\lambda$-preterm is defined recursively by the
following equations:
\begin{align*}
  |x| & = 1
& |\cst{f}(\bar u)| & = 1 + \sum\nolimits_i |u_i|
& |\DB{n}| & = 1
& |\lambda\> t| & = 1 + |t|
& |s\> t| & = |s| + |t|
\end{align*}

The set $\HTmp{\Sigma_\ty}{\Sigma}{X_\ty}{X}$ of \emph{preterms} consists of the
$\beta\eta$-equivalence classes of $\lambda$-preterms.
The set $\HTm{\Sigma_\ty}{\Sigma}{X_\ty}{X}$ of ``\emph{terms}'' consists of the
$\beta\eta$-equivalence classes of $\lambda$-terms.
Preterms have the following
four mutually exclusive forms, where $t, \bar{t}$ are terms:
\begin{itemize}
\item a fully applied variable $x\langle\tau\rangle\> \bar{t}$;
\item a fully applied symbol $\cst{f}\langle\bar\tau\rangle(\bar u)\> \bar{t}$;
\item a fully applied De Bruijn index $\DB{n}\langle\tau\rangle\> \bar{t}$;
\item a $\lambda$-abstraction $\lambda\langle\tau\rangle\> t$.
\end{itemize}
``Fully applied'' means that the preterm as a whole has nonfunctional type.
The above view is reminiscent of first-order terms: The variable and symbol
cases are essentially as for first-order terms, De Bruijn indices are regarded
as symbols, and even the $\lambda$-abstraction $\lambda\> t$ can be thought of as
a unary function application $\lambda(t)$. This will be the key to adapting the
first-order KBO and LPO to higher-order preterms.

A type $\tau$ is monomorphic
if $\tau \in \GTy{\Sigma_\ty}$---i.e., if it contains no variables; otherwise,
it is polymorphic.
A preterm is monomorphic if all its type arguments are monomorphic; otherwise, it
is polymorphic. A preterm $t$ is \emph{ground} if it is closed and monomorphic, or
equivalently if $t \in \GHTmp{\Sigma_\ty}{\Sigma}$.

Substitutions are defined as mappings from a set of type variables to types and
from term variables to terms of the same type.
A \emph{monomorphizing} type substitution maps all
type variables to ground types and leaves term variables unchanged. A
\emph{grounding} substitution maps all variables to ground types and terms.

We will say that a preterm is \emph{steady} if 
its type is neither of function type nor a type variable.

Unless otherwise specified, all preterms will be presented in
$\eta$-long normal form.

\subsection{Orders}
\label{ssec:orders}

\begin{defi}
\label{def:reflive-closure}
Given a binary relation $\Greater$, we write $\GreaterEQ$ for its reflexive
closure.
\end{defi}

\begin{defi}
\label{def:lex-extension}
Given a binary relation $\Greater$, we write $\gtgtlex$ for its left-to-right
lexicographic extension, defined as follows on same-length tuples:
$() \gtgtlex ()$ does not hold, and for $n \geq 1$,
$(y_1,\dots,y_n)\allowbreak \gtgtlex (x_1,\dots,x_n)$ holds if and only if
$y_1 \Greater x_1$ or else $y_1 = x_1$ and $(y_2,\dots,y_n) \gtgtlex (x_2,\dots,x_n)$.
\end{defi}

\begin{defi}
\label{def:strict-lex-extension}
Given binary relations $\Greater$ and $\Greatersim$, we write $\gtgtsimlex$ for
their left-to-right strict lexicographic extension, defined as follows on
same-length tuples: $() \gtgtsimlex ()$ does not hold, and for $n \geq 1$,
$(y_1,\dots,y_n) \gtgtsimlex (x_1,\dots,x_n)$ holds if and only if
$y_1 \Greater x_1$ or else $y_1 \Greatersim x_1$ and $(y_2,\dots,y_n) \gtgtsimlex (x_2,\dots,x_n)$.
\end{defi}

\begin{defi}
\label{def:nonstrict-lex-extension}
Given binary relations $\Greater$ and $\Greatersim$, we write $\gtgtapproxlex$ for
their left-to-right nonstrict lexicographic extension, defined as follows on
same-length tuples: $() \gtgtapproxlex ()$ holds, and for $n \geq 1$,
$(y_1,\dots,y_n) \gtgtapproxlex (x_1,\dots,x_n)$ holds if and only if
$y_1 \Greater x_1$ or else $y_1 \Greatersim x_1$ and $(y_2,\dots,y_n) \gtgtapproxlex (x_2,\dots,x_n)$.
\end{defi}

\begin{defi}
\label{def:precedence}
A \emph{precedence} $\Succ$ on a set $A$ is a well-founded total order $\Succ$
on $A$.
\end{defi}

The first-order Knuth--Bendix order will constitute a useful stepping stone.
Like the original \cite{knuth-bendix-1970}, the version we use is untyped.
Unlike the original, but like the transfinite KBO \cite{ludwig-waldmann-2007},
it uses ordinal weights instead of natural numbers and supports argument
coefficients.

\begin{defi}
\label{def:weight-kb}
Let $\cal w : \Sigma \to \Ordpos$ and
$\cal k : \Sigma \times \Natpos \to \Ordpos$.
Define the weight function $\cal W : \GTm{\Sigma} \to \Ord$ recursively by
\begin{align*}
  \cal W(x)
& = 0
& \cal W(\cst{f}(\bar s_m))
& = \cal w(\cst{f}) + \sum\nolimits_{i=1}^m
  \cal k(\cst{f},i) \cal W(s_i)
\end{align*}
\end{defi}

\begin{defi}
\label{def:kb}
Let $\cal w, \cal k, \cal W$ be as in Definition~\ref{def:weight-kb}.
Let $\Succ$ be an order (typically, a precedence) on an untyped signature $\Sigma$.
The \emph{strict first-order KBO} $\Greater_\kb$ induced by $\cal w, \cal k, \Succ$ on
nonground $\Sigma$-terms is defined inductively so that $t \Greater_\kb s$ if every
variable occurring in $s$ occurs at least as many times in $t$ as in $s$ and if
any of these conditions is met:
\begin{enumerate}
\item \label{itm:kb-wt}
  $\cal W(t) > \cal W(s)$;

\item \label{itm:kb-sym}
  $\cal W(t) = \cal W(s)$,
  $t = \cst{g}(\bar t)$, $s = \cst{f}(\bar s)$,
  and $\cst{g} \Succ \cst{f}$;

\item \label{itm:kb-sym-args}
  $\cal W(t) = \cal W(s)$,
  $t = \cst{g}(\bar t)$, $s = \cst{g}(\bar s)$,
  and $\bar t \gtgtlex_\kb \bar s$.
\end{enumerate}
\end{defi}

\begin{defi}
\label{def:lp}
Let $\Succ$ be an order (typically, a precedence) on an untyped signature $\Sigma$.
The \emph{strict first-order LPO} $\Greater_\lp$ induced by $\Succ$ on
nonground $\Sigma$-terms is defined inductively so that $t \Greater_\lp s$ if any of
these conditions is met, where $t = \cst{g}(\bar t_k)$:
\begin{enumerate}
\item \label{itm:lp-sub}
  $t_i \GreaterEQ_\lp s$ for some $i \in \{1,\dots,k\}$;

\item \label{itm:lp-syms}
  $s = \cst{f}(\bar s)$, $\cst{g} \Succ \cst{f}$, and
  $\cal{chkargs}(t, \bar s)$;

\item \label{itm:lp-args}
  $s = \cst{g}(\bar s)$, $\bar t \gtgtlex_\lp \bar s$, and
  $\cal{chkargs}(t, \bar s)$
\end{enumerate}
where $\cal{chkargs}(t, \bar s_k)$ if and only if $t \Greater_\lp s_i$ for every $i \in
\{1,\dots,k\}$.
\end{defi}

\section{The Ground Level}
\label{sec:the-ground-level}

We start by defining the $\lambda$KBO and $\lambda$LPO on ground preterms. We
connect them to the first-order KBO and LPO via an encoding so that we can lift
various properties, such as totality, well-foundedness, and compatibility with a
wide class of contexts. For $\lambda$KBO, in addition to $\cal w$, we will use
the parameter $\cal w_\lambda$ for the weight of a $\lambda$ and $\cal w_\db$
for the weight of a De Bruijn index.

For the rest of this paper, we fix a higher-order signature $(\Sigma_\ty,
\Sigma)$ and two infinite sets of variables $X_\ty, X$.

\subsection{\texorpdfstring{$\bm{\lambda}$KBO}{Lambda KBO}}
\label{ssec:lambda-kbo-ground}

\begin{defi}
\label{def:weight-lkbg}
Let $\cal w : \Sigma \to \Ordpos$, $\cal w_\lambda, \cal w_\db \in \Ordpos$,
and $\cal k : \Sigma \times \Natpos \to \Ordpos$.
Define the ground weight function $\cal W_\g : \GHTmp{\Sigma_\ty}{\Sigma} \to \Ordpos$
recursively by
\begin{align*}
  \cal W_\g(\cst{f}(\bar u)\> \bar t_n)
& = \cal w(\cst{f}) + \sum\nolimits_{i=1}^n \cal k(\cst{f},i) \cal W_\g(t_i)
&
  \cal W_\g(\DB{m}\> \bar t_n)
& = \cal w_\db + \sum\nolimits_{i=1}^n \cal W_\g(t_i)
\\
  \cal W_\g(\lambda\> t)
& = \cal w_\lambda + \cal W_\g(t)
\end{align*}

\end{defi}

\begin{defi}
\label{def:lkbg}
Let $\cal w_\ty : \Sigma_\ty \to \Ordpos$.
Let $\Succ^\ty$ be a precedence on $\Sigma_\ty$.
Let $\Greater_\ty$ be the strict first-order KBO induced by $\cal w_\ty$ and
$\Succ^\ty$ on $\GTm{\Sigma_\ty}$.
Let $\cal w, \cal w_\lambda, \cal w_\db, \cal k, \cal W_\g$ be as in
Definition~\ref{def:weight-lkbg}.
Let $\Succ$ be a precedence on $\Sigma$.

The \emph{strict ground $\lambda$KBO} $\Greater_\lkbg$ induced by $\cal w_\ty, \cal
w,\allowbreak \cal w_\lambda,\allowbreak \cal w_\db,\allowbreak \cal
k,\allowbreak \Succ^\ty, \Succ$ on $\GHTmp{\Sigma_\ty}{\Sigma}$ is defined inductively
so that $t \Greater_\lkbg s$ if any of these conditions is met:

\begin{enumerate}
\item \label{itm:lkbg-wt}
  $\cal W_\g(t) > \cal W_\g(s)$;

\item \label{itm:lkbg-lam}
  $\cal W_\g(t) = \cal W_\g(s)$,
  $t$ is of the form $\lambda\langle\upsilon\rangle\> t'$, and
  any of these conditions is met:
  \begin{enumerate}
  \item \label{itm:lkbg-lam-typ}
    $s$ is of the form $\lambda\langle\tau\rangle\> s'$ and
    $\upsilon \Greater_\ty \tau$, or

  \item \label{itm:lkbg-lam-body}
    $s$ is of the form $\lambda\langle\upsilon\rangle\> s'$
    and $t' \Greater_\lkbg s'$, or

  \item \label{itm:lkbg-lam-other}
    $s$ is of the form $\DB{m}\>\bar s$ or
    $\cst{f}(\bar u)\>\bar s$;
  \end{enumerate}

\item \label{itm:lkbg-db}
  $\cal W_\g(t) = \cal W_\g(s)$, $t$ is of the form
  $\DB{n}\> \bar t$, and
  any of these conditions is met:
  \begin{enumerate}
  \item \label{itm:lkbg-db-diff}
    $s$ is of the form $\DB{m}\>\bar s$ and $n > m$, or

  \item \label{itm:lkbg-db-args}
    $s$ is of the form $\DB{n}\>\bar s$ and
    $\bar t \gtgtlex_\lkbg \bar s$, or

  \item \label{itm:lkbg-db-other}
    $s$ is of the form $\cst{f}(\bar u)\>\bar s$;
  \end{enumerate}

\item \label{itm:lkbg-sym}
  $\cal W_\g(t) = \cal W_\g(s)$, $t$ is of the form
  $\cst{g}\langle\bar\upsilon\rangle(\bar w)\>\bar t$, and
  any of these conditions is met:
  \begin{enumerate}
  \item \label{itm:lkbg-sym-diff}
    $s$ is of the form $\cst{f}(\bar u)\>\bar s$ and
    $\cst{g} \Succ \cst{f}$, or

  \item \label{itm:lkbg-sym-typ}
    $s$ is of the form $\cst{g}\langle\bar\tau\rangle(\bar u)\>\bar s$ and
    $\bar\upsilon \gtgtlex_\ty \bar\tau$, or

  \item \label{itm:lkbg-sym-args}
    $s$ is of the form $\cst{g}\langle\bar\upsilon\rangle(\bar u)\>\bar s$ and
    $(\bar w, \bar t) \gtgtlex_\lkbg (\bar u, \bar s)$.
  \end{enumerate}
\end{enumerate}
\end{defi}

In rule~\ref{itm:lkbg-db-args}, we assume that leaking De Bruijn indices in
$t$ and $s$ refer to the same variable and hence have the same type. This
invariant is preserved by the recursive application in
rule~\ref{itm:lkbg-lam-body}. A more defensive approach would be to compare the
types first and then the argument tuples as a tie breaker, as in rules
\ref{itm:lkbg-sym-typ}~and~\ref{itm:lkbg-sym-args}.

\begin{defi}
Given a higher-order signature $(\Sigma_\ty, \Sigma)$, let
\begin{align*}
\Sigma_\fo
\vthinspace=\vthinspace
&
\{\cst{f}_{\bar u}^{\bar\tau} \mid \cst{f} \in \Sigma{,}\;
  \bar\tau \in (\GTy{\Sigma_\ty})^{\tyarity{\cst{f}}}{,}\;
  \bar u \in (\GHTmp{\Sigma_\ty}{\Sigma})^{\arity{\cst{f}}}\}
\\[-\jot]
& \mathrel\uplus \{\cst{db}_k^i \mid i{,}\; k \in \Nat\}
\mathrel\uplus \{\cst{lam}^\tau \mid \tau \in \GTy{\Sigma_\ty}\}
\end{align*}
be an untyped first-order signature.
\end{defi}

\begin{defi}
\label{def:encode}
The translation $\ENConly$ defined by the following equations encodes
$\GHTmp{\Sigma_\ty}{\Sigma}$ into $\GTm{\Sigma_\fo}$:
\begin{align*}
\ENC{\cst{f}\langle\bar\tau\rangle(\bar u)\> \bar t}
  & = \cst{f}_{\bar u}^{\bar\tau}(\ENC{\bar t})
&
\ENC{\DB{m}\> \bar t_n}
  & = \cst{db}_n^m(\ENC{\bar t_n})
&
\ENC{\lambda\langle\tau\rangle\> t}
  & = {\cst{lam}^\tau}(\ENC{t})
\end{align*}
\end{defi}

\begin{lem}
\label{lem:encode-injective}
The translation $\ENConly$ is injective on ground terms.
\end{lem}

\begin{proof}
By straightforward induction on $\ENConly$'s definition. 
Since we claim injectivity only for terms, not all preterms,
there is no need the type of a De Bruijn index.
The type of a De Bruijn
index is given by the corresponding enclosing $\cst{lam}^\tau$.
Similarly, the type of a parameter is
given by the function it is passed to.
\end{proof}

It will be useful to apply $\Greater_\kb$ to encoded terms. Let the symbol weights
$\cal w_\fo$ and coefficients $\cal k_{\,\,\fo}$ be derived from $\cal w$ as follows:
\begin{align*}
\cal w_\fo(\cst{f}_{\bar u}^{\bar\tau})
  & = \cal w(\cst{f})
& \cal w_\fo(\cst{db}^i_k)
  & = \cal w_\db
& \cal w_\fo(\cst{lam}^\tau)
  & = \cal w_\lambda
\\
\cal k_{\,\,\fo}(\cst{f}_{\bar u}^{\bar\tau}, i)
  & = \cal k(\cst{f}, i)
& \cal k_{\,\,\fo}(\cst{db}^j_k, i)
  & = 1
& \cal k_{\,\,\fo}(\cst{lam}^\tau, i)
  & = 1
\end{align*}
Next, let $\Succ^\kb$ be the
precedence on $\Sigma_\fo$ that sorts the elements as follows,
from smallest to largest:
\begin{enumerate}
\item Start with the symbols $\cst{f}_{\bar u}^{\bar\tau}$ in
  $\Prec$-increasing order of their symbols $\cst{f}$, using $\ltltlex_\ty$ on
  their superscripts as first tiebreaker and $\ltltlex_\lkbg$ on the subscripts
  as second tiebreaker.

\smallskip

\item Continue with the De Bruijn indices:
  $\cst{\cst{db}^0_0}, \cst{\cst{db}^0_1}, \dots$,
  followed by $\cst{\cst{db}^1_0}, \cst{\cst{db}^1_1}, \dots$,
  $\cst{\cst{db}^2_0},\allowbreak \cst{\cst{db}^2_1}, \dots,$ and so on.

\smallskip

\item Conclude with the symbols $\cst{lam}^\tau$ in $\Less_\ty$-increasing order
  of their superscripts.
\end{enumerate}
This definition ensures that symbols from $\Sigma$ are smallest and
$\cst{lam}^\tau$ symbols are largest.

Let $\Greater_\kb$ denote the first-order KBO instance induced by $\cal w_\fo,\cal
k_{\,\,\fo}, \Succ^\kb$, and let $\cal W_\kb$ denote its weight function. The
translation $\ENConly$ is faithful in the following sense:

\begin{lem}
\label{lem:weight-encode}
$\cal W_\kb(\ENC{u}) = \cal W_\g(u)$ for every $u \in \GHTmp{\Sigma_\ty}{\Sigma}$.
\end{lem}

\begin{proof}
By induction on the definition of $\cal W_\g$.
\end{proof}

\begin{lem}
\label{lem:lkbg-encode-faithful}
Given $s, t \in \GHTmp{\Sigma_\ty}{\Sigma}$, we have
$t \Greater_\lkbg s$ if and only if $\ENC{t} \Greater_\kb \ENC{s}$.
\end{lem}

\begin{proof}
By Lemma~\ref{lem:weight-encode}, any preterm $u$ has the same weight according to
$\Greater_\lkbg$ as $\ENC{u}$ according to $\Greater_\kb$. The rules for establishing $t
\Greater_\lkbg s$ and $\ENC{t} \Greater_\kb \ENC{s}$ correspond according to the following
table:

\begin{center}
\begin{tabular}{@{}l@{\quad}l@{}}
\toprule
  $\Greater_\kb$
& $\Greater_\lkbg$
\\
\midrule
  Rule \ref{itm:kb-wt}
& Rule \ref{itm:lkbg-wt}
\\
  Rule \ref{itm:kb-sym}
& Rule \ref{itm:lkbg-lam-typ},
  \ref{itm:lkbg-lam-other},
  \ref{itm:lkbg-db-diff},
  \ref{itm:lkbg-db-other},
  \ref{itm:lkbg-sym-diff},
  \ref{itm:lkbg-sym-typ}, or
  \ref{itm:lkbg-sym-args}
\\
  Rule \ref{itm:kb-sym-args}
& Rule
  \ref{itm:lkbg-lam-body},
  \ref{itm:lkbg-db-args}, or
  \ref{itm:lkbg-sym-args}
\\
\bottomrule
\end{tabular}
\end{center}
The equivalence can then be established by two proofs by induction on the
definition on $\Greater_\kb$ and $\Greater_\lkbg$, one for each direction of the
equivalence.
\end{proof}

\begin{thm}
\label{thm:lkbg-strict-partial-order}
The relation $\Greater_\lkbg$ is a strict partial order.
\end{thm}

\begin{proof}
This amounts to proving irreflexivity, antisymmetry, and transitivity.
The strategy is always the same and is illustrated for irreflexivity below.

\medskip

\noindent
\textsc{Irreflexivity}:\enskip
We must show $t \not\Greater_\lkbg t$. By Lemma~\ref{lem:lkbg-encode-faithful}, this
amounts to showing $\ENC{t} \not\Greater_\kb \ENC{t}$, which is obvious since
$\Greater_\kb$ is irreflexive.
\end{proof}

\begin{lem}
\label{lem:succ-kb-fo-precedence}
The relation $\Succ^\kb$ is a precedence.
\end{lem}

\begin{proof}
It is easy to see that the relation is total. For well-foundedness, suppose
there exists an infinite descending chain $\cst{g}_0 \Succ^\kb \cst{g}_1
\Succ^\kb \cdots$.

We say that a symbol $\cst{g}$ is \emph{bad} if there exists an infinite
chain $\cst{g} \Succ^\kb \cdots$.
Let us define the size $\|\phantom{j}\|$ of a symbol as follows:
\begin{align*}
  \|\cst{f}_{\bar u}^{\bar\sigma}\| & = 1 + \sum\nolimits_i \|u_i\|
& \|\cst{db}^i_k\| & = 1
& \|\cst{lam}^\tau\| & = 1
\end{align*}
We can assume without loss of generality that the chain
$\cst{g}_0 \Succ^\kb \cst{g}_1
\Succ^\kb \cdots$ is minimal in the following sense: $\cst{g}_0$ has
minimal size among bad symbols, and each $\cst{g}_{i+1}$ has minimal size
among bad symbols $\cst{g}$ such that $\cst{g}_i \Succ^\kb \cst{g}$.

The chain must have infinitely many steps of type 1, 2,
or 3. Since all steps of the same type are grouped together, there must exist an
index $k$ from which all steps are of the same type. We distinguish three cases,
corresponding to the three types.

\medskip

\noindent
\textsc{Case 1:}\enskip
The chain $\cst{g}_k \Succ^\kb \cst{g}_{k+1} \Succ^\kb \cdots$, where
each symbol $\cst{g}_i$ is of the form $\cst{f}_{\bar u}^{\bar\tau}$, is also an
infinite descending chain with respect to the lexicographic order induced by
$\Prec_\g$, $\ltltlex_\ty$ (for a fixed length $n$ given by $\tyarityonly$), and
$\ltltlex_\lkbg$ (for a fixed length given by $\arityonly$). Both $\Prec_\g$ and
$\ltltlex_\ty$ are well founded, so there must exist an index $l$ from which
the symbol $\cst{f}$ and its superscript $\bar \tau$ are fixed, and only the
subscripts $\bar u$ change. This means that we have an infinite chain of the
form $(\bar u_n)_l \gtgtlex_\lkbg (\bar u_n)_{l+1} \gtgtlex_\lkbg \cdots$.
By Lemma~\ref{lem:lkbg-encode-faithful}, there would also exist a chain
$\ENC{(\bar u_n)_l} \gtgtlex_\kb \ENC{(\bar u_n)_{l+1}} \gtgtlex_\kb \cdots$.

Since the bounded lexicographic order is well founded, this means that there
exists an infinite chain of the form $\ENC{v_l} \Greater_\kb \ENC{v_{l+1}} \Greater_\kb \cdots$.
Recall that the standard KBO is well founded if the underlying precedence is
well founded. If it is not, the standard well-foundedness argument tells us that
there must exist an infinite chain of distinct head symbols
$\cst{h}_0 \Succ^\kb \cst{h}_1 \Succ^\kb \cdots$. Clearly,
$\cst{h}_0$ is both bad and smaller than $\cst{g}_0$, contradicting the
minimality of $\cst{g}_0$.

\medskip

\noindent
\textsc{Case 2:}\enskip
The chain $\cst{f}_k \Succ^\kb \cst{f}_{k+1} \Succ^\kb \cdots$
corresponds to an infinite descending chain with respect to the lexicographic
order on pairs of natural numbers. Since that order is well founded, the chain
is impossible.

\medskip

\noindent
\textsc{Case 3:}\enskip
From $\cst{lam}^\tau_0 \Succ^\kb \cst{lam}^\tau_1 \Succ^\kb \cdots$, we extract
an infinite chain $\tau_0 \Greater_\ty \tau_1 \Greater_\ty \cdots$, contradicting the
well-foundedness of the first-order KBO.
\end{proof}

\begin{thm}
\label{thm:lkbg-total}
The relation $\Greater_\lkbg$ is total on ground terms.
\end{thm}

\begin{proof}
Assume $t \not= s$. We must show that $t \Greater_\lkbg s$ or $t \Less_\lkbg s$.
Note that by Lemma \ref{lem:encode-injective}, $\ENC{t} \not= \ENC{s}$.
Hence, by totality of $\Greater_\kb$, either $\ENC{t} \Greater_\kb \ENC{s}$ or
$\ENC{t} \Less_\kb \ENC{s}$. We obtain the desired result by applying
Lemma~\ref{lem:lkbg-encode-faithful} twice.
\end{proof}

\begin{thm}
\label{thm:lkbg-well-founded}
The relation $\Greater_\lkbg$ is well founded.
\end{thm}

\begin{proof}
This follows again straightforwardly by Lemma~\ref{lem:lkbg-encode-faithful}.
If there existed an infinite chain $t_0 \Greater_\lkbg t_1 \Greater_\lkbg \cdots$, there
would also exist an infinite chain $\ENC{t_0} \Greater_\kb \ENC{t_1} \Greater_\kb \cdots$,
contradicting the well-foundedness of $\Greater_\kb$.
\end{proof}

The $\lambda$-superposition calculus relies on notions of green and orange
subterms:\ the core inference rules use green subterms, whereas optional
simplification rules use orange subterms. Since all green subterms are orange
subterms, we focus on the latter.

\begin{defi}
\label{def:orange-subterms}
\emph{Orange subterms} are defined inductively on ground preterms as follows:
\begin{enumerate}
\item Every preterm is an orange subterm of itself.
\item Every orange subterm of an argument $s_i$ in $\cst{f}(\bar t)\> \bar s$ is
  an orange subterm of $\cst{f}(\bar t)\> \bar s$.
\item Every orange subterm of an argument~$s_i$ in $\DB{m}\> \bar s$ is an
  orange subterm of $\DB{m}\> \bar s$.
\item Every orange subterm of $u$ is an orange subterm of $\lambda\> u$.
\end{enumerate}
The context $u[\phantom{i}]$ surrounding an orange subterm $s$ of $u[s]$ is
called an \emph{orange context}. The notation $\orangesubterm{u}{s}$ indicates
that $s$ is an orange subterm in $u[s]$, and $\orangesubterm{u}{\phantom{i}}$
indicates that $u[\phantom{i}]$ is an orange context. The \emph{depth} of an
orange context is the number of $\lambda$s in $u[\phantom{i}]$ that have the
hole in their scope.
\end{defi}

\begin{defi}
\label{def:orange-order-properties}
A relation $\Greater$ is \emph{compatible with orange contexts} if $t \mathbin{\Greater} s$ implies
$\orangesubterm{u}{t{\uparrow}^k}\allowbreak \Greater \orangesubterm{u}{s{\uparrow}^k}$
for every orange context $\orangesubterm{u}{\phantom{i}}$, where $k$ is its
depth.
The relation $\Greater$ enjoys the \emph{orange subterm property} if
$\orangesubterm{u}{s{\uparrow}^k} \GreaterEQ s$
for every orange context $\orangesubterm{u}{\phantom{i}}$, where $k$ is its
depth.
\end{defi}

\begin{thm}
\label{thm:lkbg-compat-orange-contexts}
The relation $\Greater_\lkbg$ is compatible with orange contexts.
\end{thm}

\begin{proof}
Let $\orangesubterm{u}{\phantom{i}}$ be an orange context of depth $k$.
Assume $t \Greater_\lkbg s$. Note that by Lemma~\ref{lem:lkbg-encode-faithful}, $\ENC{t}
\Greater_\kb \ENC{s}$.
Moreover, by inspection of the rules of $\Greater_\kb$, we find that
$\ENC{t{\uparrow}^k} \Greater_\kb \ENC{s{\uparrow}^k}$.
This works because we give all De Bruijn indices the same weight, and the
precedence of indices remains stable under shifting.

Now, observe that orange subterms are mapped to first-order subterms by
$\ENConly$. In particular, there exists a first-order context $v[\phantom{i}]$
such that $\ENC{\orangesubterm{u}{t{\uparrow}^k}} = v[\ENC{t{\uparrow}^k}]$ and
$\ENC{\orangesubterm{u}{s{\uparrow}^k}} = v[\ENC{s{\uparrow}^k}]$. By
compatibility of $\Greater_\kb$ with contexts, we have $v[\ENC{t{\uparrow}^k}]
\allowbreak\Greater_\kb v[\ENC{s{\uparrow}^k}]$. Thus, by
Lemma~\ref{lem:lkbg-encode-faithful}, we get $\orangesubterm{u}{t{\uparrow}^k}
\Greater_\lkbg \orangesubterm{u}{s{\uparrow}^k}$, as desired.
\end{proof}

\begin{thm}
\label{thm:lkbg-orange-subterm-property}
The relation $\Greater_\lkbg$ has the orange subterm property.
\end{thm}

\begin{proof}
The key idea is as in the proof of Theorem~\ref{thm:lkbg-compat-orange-contexts}.
For any orange context $\orangesubterm{u}{\phantom{i}}$ of depth~$k$,
there exists a first-order context $v[\phantom{i}]$ such that
$\ENC{\orangesubterm{u}{s{\uparrow}^k}} = v[\ENC{s{\uparrow}^k}]$.
By the subterm property of $\Greater_\kb$, we have
$v[\ENC{s{\uparrow}^k}] \ge_\kb \ENC{s{\uparrow}^k}$.
By inspection of the rules of $\Greater_\kb$, we also have
$\ENC{s{\uparrow}^k} \ge_\kb \ENC{s}$.
By transitivity and Lemma~\ref{lem:lkbg-encode-faithful},
we get $\orangesubterm{u}{s{\uparrow}^k} \Greater_\lkbg s$, as desired.
\end{proof}

The last property is necessary for $\lambda$-superposition. It is easy to prove.

\begin{thm}
\label{thm:lkbg-top-bot-smallest}
Assume $\cal W_\g(\heavy{\top}) = \cal W_\g(\heavy{\bot}) = 1$ and
$\heavy{\top} \Prec \heavy{\bot} \Prec \cst{f}$ for every $\cst{f} \in
\Sigma \setminus \{\heavy{\top},\heavy{\bot}\}$.
Then $\heavy{\top} \Less_\lkbg \heavy{\bot} \Less_\lkbg t$ for every $t \in
\GHTmp{\Sigma_\ty}{\Sigma} \setminus \{\heavy{\top},\heavy{\bot}\}$.
\end{thm}

\begin{proof}
This follows straightforwardly from the definition of $\Greater_\lkbg$.
\end{proof}

The $\lambda$-superposition calculus also specifies a requirement on applied
quantifiers $\heavy{\forall}$ and $\heavy{\exists}$ occurring in clauses, after
clausification. However, this requirement is not met by our order. To circumvent
the issue, we can preprocess the quantifiers, replacing
$\heavy{\forall}\>(\lambda\> t)$ by $(\lambda\> t) \approx (\lambda\> \heavy{\top})$ and
$\heavy{\exists}\>(\lambda\> t)$ by $(\lambda\> t) \not\approx (\lambda\> \heavy{\bot})$.

\begin{thm}
  \label{thm:lkbg-diff}
  Assume $\cal w(\cst{diff}) \le \cal w_\db$ and $\cal k(\cst{diff}, i) = 1$ for every $i$. 
  For all ground types $\tau, \upsilon$ and ground preterms $s,t,u : \tau \to \upsilon$,
  we have $u \Greater_\lkbg u\>\cst{diff}\langle\tau,\upsilon\rangle(s, t)$.
\end{thm}
\begin{proof}
Since $u$ is of type $\tau \to \upsilon$, in its $\eta$-long normal form, it has the form $\lambda\> u'$ for some $u'$.
Since $\cst{diff}\langle\tau,\upsilon\rangle(s, t)$ is a symbol, we can obtain
the $\eta$-long $\beta$-normal form of $u\>\cst{diff}\langle\tau,\upsilon\rangle(s, t)$
by replacing free De Bruijn indices of $u'$ by
$\cst{diff}\langle\tau,\upsilon\rangle(s, t)$.
Since $\cal w(\cst{diff}) \le \cal w_\db$ and $\cal k(\cst{diff}, i) = 1$ for every $i$,
it follows that
$\cal W_\g(u\>\cst{diff}\langle\tau,\upsilon\rangle(s, t)) \leq \cal W_\g(u') <  W_\g(u)$.
Therefore, 
$u \Greater_\lkbg u\>\cst{diff}\langle\tau,\upsilon\rangle(s, t)$
by rule~\ref{itm:lkbg-wt}.
\end{proof}

\subsection{\texorpdfstring{$\bm{\lambda}$LPO}{Lambda LPO}}
\label{ssec:lambda-lpo-ground}

\begin{defi}
\label{def:llpg}
Let $\Succ^\ty$ be a precedence on $\Sigma_\ty$.
Let $\Greater_\ty$ be the strict first-order LPO induced by $\Succ^\ty$ on
$\GTm{\Sigma_\ty}$.
Let $\Succ$ be a precedence on $\Sigma$.
Let $\cst{ws} \in \Sigma$ be a distinguished element called the
\emph{watershed}.

The \emph{strict ground $\lambda$LPO} $\Greater_\llpg$ induced by $\Succ^\ty,\Succ$
on $\GHTmp{\Sigma_\ty}{\Sigma}$ is defined inductively so that $t \Greater_\llpg s$ if any of
these conditions is met:
\begin{enumerate}
\item \label{itm:llpg-sym}
  $t$ is of the form $\cst{g}\langle\bar\upsilon\rangle(\bar w)\> \bar t_k$ and
  any of these conditions is met:

  \begin{enumerate}
  \item \label{itm:llpg-sym-sub}
    $t_i \GreaterEQ_\llpg s$ for some $i \in \{1,\dots,k\}$, or
  \item \label{itm:llpg-sym-diff}
    $s = \cst{f}(\bar u)\> \bar s$, $\cst{g} \Succ \cst{f}$, and
    $\cal{chkargs}(t, \bar s)$, or
  \item \label{itm:llpg-sym-types}
    $s = \cst{g}\langle\bar\tau\rangle(\bar u)\> \bar s$, $\bar\upsilon \gtgtlex_\ty \bar\tau$,
    and $\cal{chkargs}(t, \bar s)$, or
  \item \label{itm:llpg-sym-args}
    $s = \cst{g}\langle\bar\upsilon\rangle(\bar u)\> \bar s$,
    $(\bar w, \bar t) \gtgtlex_\llpg (\bar u, \bar s)$,
    and $\cal{chkargs}(t, \bar s)$, or
  \item \label{itm:llpg-sym-other}
    $\cst{g} \Succ \cst{ws}$,
    $s$ is of the form $\DB{m}\>\bar s$ and $\cal{chkargs}(t, \bar s)$
    or of the form $\lambda\> s'$ and $\cal{chkargs}(t,\allowbreak [s'])$;
  \end{enumerate}

\item \label{itm:llpg-db}
  $t$ is of the form $\DB{n}\> \bar t_k$ and any of these
  conditions is met:

  \begin{enumerate}
  \item \label{itm:llpg-db-sub}
    $t_i \GreaterEQ_\llpg s$ for some $i \in \{1,\dots,k\}$, or
  \item \label{itm:llpg-db-diff}
    $s = \DB{m}\> \bar s$, $n > m$, and $\cal{chkargs}(t, \bar s)$, or
  \item \label{itm:llpg-db-args}
    $s = \DB{n}\> \bar s$, $\bar t \gtgtlex_\llpg \bar s$, and
    $\cal{chkargs}(t, \bar s)$, or
  \item \label{itm:llpg-db-other}
    $s$ is of the form $\lambda\> s'$ and $\cal{chkargs}(t, [s'])$ or
    of the form or $\cst{f}(\bar u)\> \bar s$,
    where $\cst{f} \Preceq \cst{ws}$,
    and $\cal{chkargs}(t, \bar s)$;
  \end{enumerate}

\item \label{itm:llpg-lam}
  $t$ is of the form $\lambda\langle\upsilon\rangle\> t'$ and any of these
  conditions is met:

  \begin{enumerate}
  \item \label{itm:llpg-lam-sub}
    $t' \GreaterEQ_\llpg s$, or
  \item \label{itm:llpg-lam-types}
    $s = \lambda\langle\tau\rangle\> s'$, $\upsilon \Greater_\ty \tau$, and
    $\cal{chkargs}(t, [s'])$, or
  \item \label{itm:llpg-lam-bodies}
    $s = \lambda\langle\upsilon\rangle\> s'$ and $t' \Greater_\llpg s'$, or
  \item \label{itm:llpg-lam-other}
    $s$ is of the form $\cst{f}(\bar u)\> \bar s$,
    where $\cst{f} \Preceq \cst{ws}$,
    and $\cal{chkargs}(t, \bar s)$
  \end{enumerate}
\end{enumerate}
where $\cal{chkargs}(t, \bar s_k)$ if and only if $t \Greater_\llpg s_i$ for every $i \in
\{1,\dots,k\}$. The notation $[\phantom{i}]$ is used to represent lists---here,
the singleton list.
\end{defi}

Let $\Sigma_\fo$ be a first-order signature as defined in
Sect.~\ref{ssec:lambda-kbo-ground}. Let $\Succ^\lp$ be the precedence on
$\Sigma_\fo$ that orders the elements as follows,
from smallest to largest:
\begin{enumerate}
\item Start with the symbols $\cst{f}_{\bar u}^{\bar\tau}$ such that
  $\cst{f} \Preceq \cst{ws}$ in
  $\Prec$-increasing order of their symbols $\cst{f}$, using $\ltltlex_\ty$ on
  their superscripts as first tiebreaker and $\ltltlex_\llpg$ on the subscripts
  as second tiebreaker.

\smallskip

\item Continue with the symbols $\cst{lam}^\tau$ in $\Less_\ty$-increasing order
  of their superscripts.

\smallskip

\item Continue with the De Bruijn indices:
  $\cst{\cst{db}^0_0}, \cst{\cst{db}^0_1}, \dots$,
  followed by $\cst{\cst{db}^1_0}, \cst{\cst{db}^1_1}, \dots$,
  $\cst{\cst{db}^2_0},\allowbreak \cst{\cst{db}^2_1}, \dots,$ and so on.

\smallskip

\item Conclude with the symbols $\cst{f}_{\bar u}^{\bar\tau}$ such that
  $\cst{f} \Succ \cst{ws}$ in
  $\Prec$-increasing order of their symbols $\cst{f}$, using $\ltltlex_\ty$ on
  their superscripts as first tiebreaker and $\ltltlex_\llpg$ on the subscripts
  as second tiebreaker.
\end{enumerate}
This definition ensures that symbols below the watershed are smallest and
symbols above the watershed are largest. When considering polymorphism, we will
see that it is advantageous to put symbols above the watershed. However, the
special symbol $\cst{diff}$ belongs below the watershed.

\begin{lem}
\label{lem:succ-lp-fo-precedence}
The relation $\Succ^\lp$ is a precedence
\end{lem}

\begin{proof}
The proof is analogous to that of Lemma~\ref{lem:succ-kb-fo-precedence}.
\end{proof}

Let $\Greater_\lp$ denote the first-order LPO instance induced by the precedence
$\Succ^\lp$.

\begin{lem}
\label{lem:llpg-encode-faithful}
Given $s, t \in \GHTmp{\Sigma_\ty}{\Sigma}$, we have
$t \Greater_\llpg s$ if and only if $\ENC{t} \Greater_\lp \ENC{s}$.
\end{lem}

\begin{proof}
The rules for establishing $t \Greater_\llpg s$ and $\ENC{t} \Greater_\lp \ENC{s}$ correspond
according to the following table:

\begin{center}
\begin{tabular}{@{}l@{\quad}l@{}}
\toprule
  $\Greater_\lp$
& $\Greater_\llpg$
\\
\midrule
  Rule \ref{itm:lp-sub}
& Rule \ref{itm:llpg-sym-sub},
    \ref{itm:llpg-db-sub}, or
    \ref{itm:llpg-lam-sub}
\\
  Rule \ref{itm:lp-syms}
& Rule
    \ref{itm:llpg-sym-diff},
    \ref{itm:llpg-sym-types},
    \ref{itm:llpg-sym-args},
    \ref{itm:llpg-sym-other},
    \ref{itm:llpg-db-diff},
    \ref{itm:llpg-db-other},
    \ref{itm:llpg-lam-types}, or
    \ref{itm:llpg-lam-other}
\\
  Rule \ref{itm:lp-args}
& Rule \ref{itm:llpg-sym-args},
    \ref{itm:llpg-db-args}, or
    \ref{itm:llpg-lam-bodies}
\\
\bottomrule
\end{tabular}
\end{center}
The equivalence can then be established by two proofs by induction on the
definition on $\Greater_\lp$ and $\Greater_\llpg$, one for each direction of the
equivalence.
The only nontrivial case is that of rule~\ref{itm:llpg-lam-bodies} of
$\Greater_\llpg$, because it lacks the $\cal{chkargs}$ condition of the corresponding
rule \ref{itm:lp-args} of $\Greater_\lp$.
Given $t \Greater_\llpg s$ by rule~\ref{itm:llpg-lam-bodies},
to obtain $\ENC{t} \Greater_\lp \ENC{s}$ by rule~\ref{itm:lp-args},
we must show $\cal{chkargs}(\ENC{t}, (\ENC{s'}))$.
We apply transitivity to combine $\ENC{t}
\Greater_\lp \ENC{t'}$, which follows from the subterm property, and the induction
hypothesis $\ENC{t'} \Greater_\lp \ENC{s'}$.
\end{proof}

Using Lemma~\ref{lem:llpg-encode-faithful}, we can prove the following theorems
using the same strategy as for Theorems
\ref{thm:lkbg-strict-partial-order}--\ref{thm:lkbg-top-bot-smallest}:

\begin{thm}
\label{thm:llpg-strict-partial-order}
The relation $\Greater_\llpg$ is a strict partial order.
\end{thm}

\begin{thm}
\label{thm:llpg-total}
The relation $\Greater_\llpg$ is total on ground preterms.
\end{thm}

\begin{thm}
\label{thm:llpg-well-founded}
The relation $\Greater_\llpg$ is well founded.
\end{thm}

\begin{thm}
\label{thm:llpg-compat-orange-contexts}
The relation $\Greater_\llpg$ is compatible with orange contexts.
\end{thm}

\begin{thm}
\label{thm:llpg-orange-subterm-property}
The relation $\Greater_\llpg$ has the orange subterm property.
\end{thm}

\begin{thm}
\label{thm:llpg-top-bot-smallest}
Assume $\heavy{\top} \Prec \heavy{\bot} \Prec \cst{f}$ for every $\cst{f} \in
\Sigma \setminus \{\heavy{\top},\heavy{\bot}\}$
and $\heavy{\bot} \Preceq \cst{ws}$.
Then $\heavy{\top} \Less_\llpg \heavy{\bot} \Less_\llpg t$ for every $t \in
\GHTmp{\Sigma_\ty}{\Sigma} \setminus \{\heavy{\top},\heavy{\bot}\}$.
\end{thm}

\begin{thm}
  \label{thm:llpg-diff}
  Let $\cst{diff} \Preceq \cst{ws}$.
  For all ground types $\tau, \upsilon$ and ground preterms $s,t,u : \tau \to \upsilon$,
  we have $u \Greater_\llpg u\>\cst{diff}\langle\tau,\upsilon\rangle(s, t)$.
\end{thm}
\begin{proof}
Since $u$ is of type $\tau \to \upsilon$, in its $\eta$-long normal form, it has the form $\lambda\> u'$ for some $u'$.
Since $\cst{diff}\langle\tau,\upsilon\rangle(s, t)$ is a symbol, we can obtain
the $\eta$-long $\beta$-normal form of $u\>\cst{diff}\langle\tau,\upsilon\rangle(s, t)$
by replacing the free De Bruijn indices of $u'$ by
$\cst{diff}\langle\tau,\upsilon\rangle(s, t)$.

So,
in order to show that $u \Greater_\llpg u\>\cst{diff}\langle\tau,\upsilon\rangle(s, t)$,
we apply rule~\ref{itm:llpg-lam-sub},
and it remains to show that  $u' \GreaterEQ_\llpg u\>\cst{diff}\langle\tau,\upsilon\rangle(s, t)$.
We follow the structure of $u'$
and $u\>\cst{diff}\langle\tau,\upsilon\rangle(s, t)$ as follows.
Whenever heads coincide,
we apply rules~\ref{itm:llpg-sym-args},
\ref{itm:llpg-db-args}, or \ref{itm:llpg-lam-bodies}
to decompose both sides.
When heads do not coincide, we note that by our observation above,
this can only happen when one side is a De Bruijn index
and the other side is 
$\cst{diff}\langle\tau,\upsilon\rangle(s, t)$.
So we can then apply rule~\ref{itm:llpg-db-other}
because $\cst{diff} \Preceq \cst{ws}$.
For any $\cal{chkargs}$ conditions arising in this procedure,
we apply rules \ref{itm:llpg-sym-sub},
\ref{itm:llpg-db-sub}, or \ref{itm:llpg-lam-sub},
and use the same procedure for the resulting proof obligations.
\end{proof}

The above proof crucially depends on $\cst{diff}$'s placement below the
watershed. If we allowed $\cst{diff} \Succ \cst{ws}$, the comparison $\DB{0}
\Greater_\llp \cst{diff}(\ldots)$ would fail. Theorem~\ref{thm:llpg-diff} is the
watershed's reason for being.

\section{The Monomorphic Level}
\label{sec:the-monomorphic-level}

Next, we generalize the definition of $\lambda$KBO to monomorphic nonground
preterms: preterms containing no type variables. The result coincides with the ground
$\lambda$KBO on ground preterms while supporting term variables. Variables give
rise to polynomial constraints, which must be solved when comparing terms.

The key idea, already present in the $\lambda$-free KBO by Becker et al.\
\cite{becker-et-al-2017}, is to use polynomials to symbolically represent the
weight of a nonground term. The weight of $y\> \cst{a}$, where $\cst{a} :
\kappa$, will be represented symbolically as $1 + \mathbf{w}_y +
\mathbf{k}_{y,1} (\cal w(\cst{a}) - \cal w_\db)$, where $1 + \mathbf{w}_y$
stands for the weight of whatever term will instantiate $y$ without its leading
$\lambda$s and $\mathbf{k}_{y,1}$ for the number of copies of the first curried
argument, here $\cst{a}$, that the term will make. If an argument coefficient
other than $1$ is used, that number of copies will be inflated by the
coefficient. The ${-}\> \cal w_\db$ monomial accounts for the loss of a De
Bruijn index occurring in $y$ when passing the argument $\cst{a}$ and
$\beta$-reducing.

A subtle difference between the indeterminate $\mathbf{k}_{y,1}$ and the
argument coefficient $\cal k(\cst f, i)$ is that $\mathbf{k}_{y,1}$ can take a
value of $0$; for example, $\lambda\> \cst{b}$ makes zero copies of its
argument. Becker et al.\ excluded this scenario so that they could get
compatibility with arguments, but this property is not needed by
$\lambda$-superposition.

Another subtlety concerns higher-order functions. The arithmetic above works
because the argument $\cst{a}$ is a simple symbol. If it were a
$\lambda$-abstraction, it could appear applied inside $y$ and trigger further
$\beta$-reductions, complicating matters. In such cases, we simply give up and
use a single indeterminate $\mathbf{w}_{y\>\bar{t}}$ to represent both the
applied variable and its arguments of functional types. We do the same with
arguments of variable type, since type variables can be instantiated with
functional types.

Some precision can be gained by normalizing the subscript of
$\mathbf{w}_{y\>\bar{t}}$. For example, if $\cst{a}$ and $\cst{b}$ have the same
weight, then $\mathbf{w}_{y\>(\lambda\>\cst{a}\>\DB{0})}$ and
$\mathbf{w}_{y\>(\lambda\>\cst{b}\>\DB{0})}$ will always evaluate to the same
result and can be identified. Our simple analysis merges all symbols and De
Bruijn indices with the same weight using a normalization function $\NORMonly$.

\subsection{\texorpdfstring{$\bm{\lambda}$KBO}{Lambda KBO}}
\label{ssec:lambda-kbo-monomorphic}

\begin{defi}
\label{def:polynomials-monomorphic}
Let $(\Sigma_\ty,\Sigma)$ be a higher-order signature.
We denote by $\Poly$ the set of
$\Ord$-valued polynomials of the following distinct indeterminates,
where $y \in X$, $\bar t \in (\HTmp{\Sigma_\ty}{\Sigma}{\emptyset}{X})^*$, and
$i \in \Natpos$:
\begin{itemize}
\item $\mathbf{w}_{y\>\bar t}$, ranging over $\Ord$, represents the weight, minus
  1, of the variable $y$ applied to the arguments $\bar t$ but without
  any leading $\lambda$s corresponding to extra arguments;

\item $\mathbf{k}_{y\>\bar t,i}$, ranging over $\Ord$, represents the coefficient
  to apply on the $i$th extra argument of $y$ already applied to $\bar t$.
\end{itemize}
An \emph{assignment} is a mapping from indeterminates to values in the
indeterminates' specified ranges. Given a polynomial $\cal w$ and an assignment
$A$, $\cal w{\big|}_A \in \Ord$ denotes $\cal w$'s value under $A$, obtained by
replacing each indeterminate $\mathbf{x}$ in $A$'s domain by $A(\mathbf{x})$.
Overloading notation, we write $\cal w{\big|}_\sigma$ for
the application of the polynomial substitution~$\sigma$ to $\cal w$;
for example, if $\sigma = \{\mathbf{w}_{y} \mapsto \mathbf{w}_{z}\}$, then
$\mathbf{w}_{y}{\big|}_\sigma = \mathbf{w}_{z}$.
Given polynomials $w, w'$, we write $w' \ge w$ if we have $w'{\big|}_A \ge
w{\big|}_A$ for every assignment $A$, and similarly for $>$, $\le$, $<$, and $=$.
\end{defi}

\begin{defi}
\label{def:normalize-lkm}
Let $\Sigma' = \Sigma \uplus \{k_\tau \mid k \in \Ordpos \text{ and } \tau \in \Ty{\Sigma_\ty}{\emptyset}\}$.
Define the normalization function $\NORMonly :
\HTmp{\Sigma_\ty}{\Sigma}{\emptyset}{X} \to \HTmp{\Sigma_\ty}{\Sigma'}{\emptyset}{X}$
recursively by
\begin{align*}
\NORM{y\> \bar t}
  & = y\> \NORM{\bar t}
\\
\NORM{\cst{f}(\bar u)\> \bar t}
  & =
\begin{cases}
  k_\tau\> \NORM{\bar t} & \text{if $\cal k(\cst{f}, i) = 1$ for every $i$, with $\cal w(\cst{f}) = k$ and $\cst{f}(\bar u)\> \bar t : \tau$}
\\
  \cst{f}(\bar u)\> \NORM{\bar t} & \text{otherwise}
\end{cases}
\\
\NORM{\DB{m}\langle\tau\rangle\> \bar t}
  & = (\cal w_\db)_\tau\> \NORM{\bar t}
\\
\NORM{\lambda\> t}
  & = \lambda\> \NORM{t}
\end{align*}
\end{defi}

\begin{defi}
\label{def:weight-lkbm}
Let $\cal w : \Sigma \to \Ordpos$, $\cal w_\lambda, \cal w_\db \in \Ordpos$,
and $\cal k : \Sigma \times \Natpos \to \Ordpos$.
Given a list of preterms $\bar t$, let $\bar t\RIGHT$ denote the longest suffix
consisting of steady preterms, and let $\bar t\LEFT$
denote the complementary prefix.
Define the monomorphic weight function $\cal W_\m :
\HTmp{\Sigma_\ty}{\Sigma}{\emptyset}{X} \to \Poly$ recursively by
\begin{align*}
  \cal W_\m(y\>\bar{t})
& = 1 + \mathbf{w}_{y\>\NORM{\bar{t}\LEFT}}
  + \sum\nolimits_{i=1}^{|\bar{t}\RIGHT|}
    \mathbf{k}_{y\>\NORM{\bar{t}\LEFT},i} (\cal W_\m(\bar{t}\RIGHT_i) - \cal w_\db)
\\
  \cal W_\m(\cst{f}(\bar u)\> \bar t_n)
& = \cal w(\cst{f}) + \sum\nolimits_{i=1}^n \cal k(\cst{f}, i)
    \cal W_\m(t_i)
\\
  \cal W_\m(\DB{m}\> \bar t_n)
& = \cal w_\db + \sum\nolimits_{i=1}^n \cal W_\m(t_i)
\\
  \cal W_\m(\lambda\> t)
& = \cal w_\lambda + \cal W_\m(t)
\end{align*}
\end{defi}

In the first equation, $\cal W_\m(\bar{t}\RIGHT_i)$ gives the argument's weight,
whereas $\cal w_\db$ is the weight of the
De Bruijn index that gets replaced by the argument.

\begin{rem}
It is possible to generalize the theory above to let $\bar t \RIGHT$ consist of
all steady preterms, regardless of their location. The interpretation of
$\mathbf{w}_{y\>\bar t}$ and $\mathbf{k}_{y\>\bar t,i}$ must then be changed to
shuffle the $\lambda$s, pulling those corresponding to the arguments $\bar t$ to
the front. For example, if $y : \kappa \to (\kappa \to \kappa) \to \kappa$,
the indeterminate $\mathbf{w}_{y\>t}$ represents the weight of the term
$(\lambda\> \lambda\>
y\>\DB{0}\>\DB{1})\> t = \lambda\> y\>\DB{0}\>(t{\uparrow})$ (but without its
leading $\lambda$).

Another possible generalization would be to normalize complex preterms, producing
for example $\cst{2}\langle\tau\rangle$ instead of
$\cst{1}\langle\sigma\to\tau\rangle\> \cst{1}\langle\sigma\rangle$.
\end{rem}

\begin{defi}
\label{def:lkbm}
Let $\cal w_\ty, \cal w, \cal w_\lambda, \cal w_\db, \cal k,\allowbreak \cal W_\m$
be as in Definition~\ref{def:weight-lkbm}.
Let $\Succ^\ty$ be a precedence on $\Sigma_\ty$.
Let $\Greater_\ty$ be the strict first-order KBO on $\Tm{\Sigma_\ty}{\emptyset}$
induced by $\cal w_\ty$ and $\Succ^\ty$.
Let $\Succ$ be a precedence on $\Sigma$.

\begin{sloppypar}
The \emph{strict monomorphic $\lambda$KBO} $\Greater_\lkbm$
and the \emph{nonstrict monomorphic $\lambda$KBO} $\Greatersim_\lkbm$ induced by $\cal w_\ty, \cal
w, \cal w_\lambda, \cal w_\db, \cal k,\allowbreak \Succ^\ty, \Succ$ on
$\HTmp{\Sigma_\ty}{\Sigma}{\emptyset}{X}$
are defined by mutual induction.
The strict
relation is defined so that 
$t \Greater_\lkbm s$ if any
of these conditions is met:
\end{sloppypar}

\begin{enumerate}
\item \label{itm:lkbm-wt}
  $\cal W_\m(t) > \cal W_\m(s)$;

\item \label{itm:lkbm-lam}
  $\cal W_\m(t) \ge \cal W_\m(s)$,
  $t$ is of the form $\lambda\langle\upsilon\rangle\> t'$, and
  any of these conditions is met:
  \begin{enumerate}
  \item \label{itm:lkbm-lam-typ}
    $s$ is of the form $\lambda\langle\tau\rangle\> s'$ and
    $\upsilon \Greater_\ty \tau$, or

  \item \label{itm:lkbm-lam-body}
    $s$ is of the form $\lambda\langle\upsilon\rangle\> s'$
    and $t' \Greater_\lkbm s'$, or

  \item \label{itm:lkbm-lam-other}
    $s$ is of the form $\DB{m}\>\bar s$ or
    $\cst{f}(\bar u)\>\bar s$;
  \end{enumerate}

\item \label{itm:lkbm-db}
  $\cal W_\m(t) \ge \cal W_\m(s)$, $t$ is of the form
  $\DB{n}\> \bar t$, and
  any of these conditions is met:
  \begin{enumerate}
  \item \label{itm:lkbm-db-diff}
    $s$ is of the form $\DB{m}\>\bar s$ and $n > m$, or

  \item \label{itm:lkbm-db-args}
    $s$ is of the form $\DB{n}\>\bar s$ and
    $\bar t \gtgtsimlex_\lkbm \bar s$, or

  \item \label{itm:lkbm-db-other}
    $s$ is of the form $\cst{f}(\bar u)\>\bar s$;
  \end{enumerate}

\item \label{itm:lkbm-sym}
  $\cal W_\m(t) \ge \cal W_\m(s)$, $t$ is of the form
  $\cst{g}\langle\bar\upsilon\rangle(\bar w)\>\bar t$, and
  any of these conditions is met:
  \begin{enumerate}
  \item \label{itm:lkbm-sym-diff}
    $s$ is of the form $\cst{f}(\bar u)\>\bar s$ and
    $\cst{g} \Succ \cst{f}$, or

  \item \label{itm:lkbm-sym-typ}
    $s$ is of the form $\cst{g}\langle\bar\tau\rangle(\bar u)\>\bar s$ and
    $\bar\upsilon \gtgtlex_\ty \bar\tau$, or

  \item \label{itm:lkbm-sym-args}
    $s$ is of the form $\cst{g}\langle\bar\upsilon\rangle(\bar u)\>\bar s$ and
    $(\bar w, \bar t) \gtgtsimlex_\lkbm (\bar u, \bar s)$.
  \end{enumerate}
\end{enumerate}

The nonstrict
relation is defined so that 
$t \Greatersim_\lkbm s$ if any
of these conditions is met:
\begin{enumerate}
\item \label{itm:nslkbm-wt}
  $\cal W_\m(t) > \cal W_\m(s)$;

\item \label{itm:nslkbm-var}
  $t$ is of the form $y\> \bar t$, $s$ is of the form $y\> \bar s$, and for
  every $i$, $t_i$ is steady and $t_i \Greatersim_\lkbm s_i$;

\item \label{itm:nslkbm-lam}
  $\cal W_\m(t) \ge \cal W_\m(s)$,
  $t$ is of the form $\lambda\langle\upsilon\rangle\> t'$, and
  any of these conditions is met:
  \begin{enumerate}
  \item \label{itm:nslkbm-lam-typ}
    $s$ is of the form $\lambda\langle\tau\rangle\> s'$ and
    $\upsilon \Greater_\ty \tau$, or

  \item \label{itm:nslkbm-lam-body}
    $s$ is of the form $\lambda\langle\upsilon\rangle\> s'$
    and $t' \Greatersim_\lkbm s'$, or

  \item \label{itm:nslkbm-lam-other}
    $s$ is of the form $\DB{m}\>\bar s$ or
    $\cst{f}(\bar u)\>\bar s$;
  \end{enumerate}

\item \label{itm:nslkbm-db}
  $\cal W_\m(t) \ge \cal W_\m(s)$, $t$ is of the form
  $\DB{n}\> \bar t$, and
  any of these conditions is met:
  \begin{enumerate}
  \item \label{itm:nslkbm-db-diff}
    $s$ is of the form $\DB{m}\>\bar s$ and $n > m$, or

  \item \label{itm:nslkbm-db-args}
    $s$ is of the form $\DB{n}\>\bar s$ and
    $\bar t \gtgtapproxlex_\lkbm \bar s$, or

  \item \label{itm:nslkbm-db-other}
    $s$ is of the form $\cst{f}(\bar u)\>\bar s$;
  \end{enumerate}

\item \label{itm:nslkbm-sym}
  $\cal W_\m(t) \ge \cal W_\m(s)$, $t$ is of the form
  $\cst{g}\langle\bar\upsilon\rangle(\bar w)\>\bar t$, and
  any of these conditions is met:
  \begin{enumerate}
  \item \label{itm:nslkbm-sym-diff}
    $s$ is of the form $\cst{f}(\bar u)\>\bar s$ and
    $\cst{g} \Succ \cst{f}$, or

  \item \label{itm:nslkbm-sym-typ}
    $s$ is of the form $\cst{g}\langle\bar\tau\rangle(\bar u)\>\bar s$ and
    $\bar\upsilon \gtgtlex_\ty \bar\tau$, or

  \item \label{itm:nslkbm-sym-args}
    $s$ is of the form $\cst{g}\langle\bar\upsilon\rangle(\bar u)\>\bar s$ and
    $(\bar w, \bar t) \gtgtapproxlex_\lkbm (\bar u, \bar s)$.
  \end{enumerate}
\end{enumerate}
\end{defi}

Rules \ref{itm:lkbm-lam}~to~\ref{itm:lkbm-sym} for $\Greater_\lkbm$ and rules
\ref{itm:nslkbm-lam}~to~\ref{itm:nslkbm-sym} for $\Greatersim_\lkbm$ use $\ge$ instead
of $=$ to compare weights because polynomials cannot always be compared
precisely. For example, if $\cal w(\cst a) = 1$, where $\cst a : \kappa$, we can
know that $\cal W_\m(x) = 1 + \mathbf{w}_x \ge 1 = \cal W_\m(\cst a)$ even though
neither $\cal W_\m(x) > \cal W_\m(\cst a)$ nor $\cal W_\m(x) = \cal W_\m(\cst a)$.

To determine whether one preterm is larger than another, we must solve an
inequality, which can be recast into $w \ge 0$ or $w > 0$. The strict case
arises in rule~\ref{itm:lkbm-wt} of the definitions of $\Greater_\lkbm$ and
$\Greatersim_\lkbm$. The two cases can be unified by writing $w > 0$ as $w - 1 \ge
0$.
Solving systems of integer polynomial inequalities is in general undecidable.
Here, however, we have a single polynomial $w$, in which indeterminates range
only over nonnegative values. This is the key to solving the problem efficiently
in practice. If infinite ordinals (e.g. $\omega$) are used as the weight or
coefficient associated with any symbols, we must also let the
$\mathbf{w}_{y\>\bar t}$ and $\mathbf{k}_{y\>\bar t,i}$ indeterminates range over
these.

Specifically, we propose the following procedure to check an inequality of the
above form: Put $w$ in standard form. If all monomial coefficients are
nonnegative, report that the inequality holds. Otherwise, report that it might
not hold.
This simple procedure can lose solutions. For example,
$(\mathbf{w}_y - 3)\mathbf{w}_y + 3 \ge 0$ holds, yet its standard form
$\mathbf{w}_y^2 - 3\mathbf{w}_y + 3$ contains a negative coefficient, which is
enough to lead the procedure astray.

Below we will connect the monomorphic $\lambda$KBO with its ground counterpart
to lift its properties.

\begin{lem}\label{lem:nslkbm-reflexive}
  For all preterms $t$,  we have $t \Greatersim_\lkbm t$.
\end{lem}
\begin{proof}
  By structural induction on $t$, using rules~\ref{itm:nslkbm-lam-body}, \ref{itm:nslkbm-db-args}, and \ref{itm:nslkbm-sym-args}.
\end{proof}

\begin{lem}\label{lem:lkbm-implies-nslkbm}
  If $t \Greater_\lkbm s$, then $t \Greatersim_\lkbm s$.
\end{lem}
\begin{proof}
  By induction on the derivation of $t \Greater_\lkbm s$.
  For each rule, there is a clearly corresponding rule of $\Greatersim_\lkbm$ to apply.
  The induction hypothesis is only required for applying rule~\ref{itm:nslkbm-lam-body}.
  A crucial observation is that $\gtgtsimlex_\lkbm$ implies $\gtgtapproxlex_\lkbm$ by definition.
\end{proof}

\begin{defi}
\label{def:assignment-of-grounding-subst}
Let $\theta$ be a substitution
that maps to terms containing only nonfunctional variables.
We define an assignment $\POLY{\theta}$ that
maps each indeterminate~$\mathbf{x}$ to a value according to the semantics given
by Definition~\ref{def:polynomials-monomorphic} after applying $\theta$ onto the
considered preterms.
We define
$\POLY{\theta}(\mathbf{w}_{y\>\bar t}) = \cal W_\m((y\>\bar t)\theta!) - 1$,
where, given a preterm $t$, $t!$ denotes the same preterm without any leading
$\lambda$s---e.g., $(\lambda\>\lambda\>\cst{f}\>\DB{0})! = \cst{f}\>\DB{0}$.
As for
$\POLY{\theta}(\mathbf{k}_{y\>\bar t,i})$,
it is defined as the number of De Bruijn indices in $(y\>\bar t)\theta$ referring to its
$i$th argument, multiplied by all argument coefficients above it.
Here, it is crucial that $\theta$ maps to terms containing only nonfunctional variables
because the argument coefficient to assign to a variable is not always clear.

\end{defi}

\begin{lem}
\label{lem:weight-subst}
Given a substitution $\theta$ that maps to terms containing only nonfunctional variables, we have
$\cal W_\m(t){\big|}_{\POLY{\theta}} = \cal W_\m(t\theta)$.
\end{lem}

\begin{proof}
The proof is by induction on the definition of $\cal W_\m$. Let $A = \POLY{\theta}$.

\medskip

\noindent
\textsc{Case $t = y\>\bar{t}$:}\enskip
We have
\begin{align*}
      & \cal W_\m(y\>\bar{t}){\big|}_A
\\
{=}\; & 1 + \mathbf{w}_{y\>\bar{t}\LEFT}{\big|}_A
  + \sum\nolimits_{i=1}^{\smash{|\bar{t}\RIGHT|}}
    \mathbf{k}_{y\>\bar{t}\LEFT,i}{\big|}_A
    (\cal W_\m(\bar{t}\RIGHT_i){\big|}_A - \cal w_\db)
\\[-\jot]
& \quad\text{by definition of $\cal W_\m$ and $\NORMonly$}
\\
{=}\; & 1 + \mathbf{w}_{y\>\bar{t}\LEFT}{\big|}_A
  + \sum\nolimits_{i=1}^{|\bar{t}\RIGHT|}
    \mathbf{k}_{y\>\bar{t}\LEFT,i}{\big|}_A
    (\cal W_\m(\bar{t}\RIGHT_i\theta) - \cal w_\db)
\\[-\jot]
& \quad\text{by the induction hypothesis}
\\
{=}\; & \cal W_\m((y\>\bar{t})\theta)
\\[-\jot]
& \quad\text{by the semantics of $\mathbf{w}$ and $\mathbf{k}$}
\end{align*}
In the last step, the arithmetic works because all preterms in $\bar{t}\RIGHT$ are
steady. This means that they will not trigger any $\beta$-reductions when they replace
a De Bruijn index.

\medskip

\noindent
\textsc{Case $t = \cst{f}(\bar u)\> \bar t_n$:}\enskip
We have
\begin{align*}
      & \cal W_\m(\cst{f}(\bar u)\> \bar t_n){\big|}_A
\\
{=}\; & \cal w(\cst{f}) + \sum\nolimits_{i=1}^n
        \cal k(\cst{f}, i) \cal W_\m(t_i){\big|}_A
  && \text{by definition of $\cal W_\m$}
\\
{=}\; & \cal w(\cst{f}) + \sum\nolimits_{i=1}^n
        \cal k(\cst{f}, i) \cal W_\m(t_i\theta)
  && \text{by the induction hypothesis}
\\
{=}\; & \cal W_\m(\cst{f}(\bar u\theta)\> (\bar t_n\theta))
  && \text{by definition of $\cal W_\m$}
\\
{=}\; & \cal W_\m((\cst{f}(\bar u)\> \bar t_n)\theta)
  && \text{by definition of substitution}
\end{align*}

\medskip

\noindent
\textsc{Case $t = \DB{m}\> \bar t_n$:}\enskip
This case is similar to the previous one.
We have
\begin{align*}
      & \cal W_\m(\DB{m}\> \bar t_n){\big|}_{A}
\\
{=}\; & \cal w_\db + \sum\nolimits_{i=1}^n \cal W_\m(t_i){\big|}_A
  && \text{by definition of $\cal W_\m$}
\\
{=}\; & \cal w_\db + \sum\nolimits_{i=1}^n \cal W_\m(t_i\theta)
  && \text{by the induction hypothesis}
\\
{=}\; & \cal W_\m(\DB{m}\> (\bar t_n\theta))
  && \text{by definition of $\cal W_\m$}
\\
{=}\; & \cal W_\m((\DB{m}\> \bar t_n)\theta)
  && \text{by definition of substitution}
\end{align*}

\medskip

\noindent
\textsc{Case $t = \lambda\> t$:}\enskip
We have
\begin{align*}
      & \cal W_\m(\lambda\> t){\big|}_{A}
\\
{=}\; & \cal w_\lambda + \cal W_\m(t){\big|}_A
  && \text{by definition of $\cal W_\m$}
\\
{=}\; & \cal w_\lambda + \cal W_\m(t\theta)
  && \text{by the induction hypothesis}
\\
{=}\; & \cal W_\m(\lambda\> (t\theta))
  && \text{by definition of $\cal W_\m$}
\\
{=}\; & \cal W_\m((\lambda\> t)\theta)
  && \text{by definition of substitution}
\qedhere
\end{align*}
\end{proof}

The nonground relation $\Greater_\lkbm$ underapproximates the ground relation $\Greater_\lkbg$
in the following sense:

\begin{sloppypar}
\begin{thm}
\label{thm:lkbm-grounding-subst-stable}
If $t \Greater_\lkbm s$, then $t\theta \Greater_\lkbg s\theta$
for all grounding substitutions $\theta$.
If $t \Greatersim_\lkbm s$, then $t\theta \GreaterEQ_\lkbg s\theta$
for all grounding substitutions $\theta$.
\end{thm}
\end{sloppypar}

\begin{proof}
We prove both claims by
mutual induction on the shape of the derivation of $t \Greater_\lkbm s$
and $t \Greatersim_\lkbm s$.
Let $A = \POLY{\theta}$.

First,
we make the following observation:
For all tuples of preterms $\bar t$ and $\bar s$ covered by the induction hypothesis,
\begin{itemize}
\item
$\bar t \gtgtsimlex_\lkbm \bar s$ implies $\bar t\theta \gtgtlex_\lkbg \bar s\theta$, and
\medskip
\item
$\bar t \gtgtapproxlex_\lkbm \bar s$ implies 
$\bar t\theta \gtgtlex_\lkbg \bar s\theta$ or
$\bar t\theta = \bar s\theta$.
\end{itemize}
This follows from the induction hypothesis and the definitions of the lexicographic extensions (Definitions~\ref{def:lex-extension}, \ref{def:strict-lex-extension}, and \ref{def:nonstrict-lex-extension})
by induction on the length of the tuples.

With this observation, we prove the two claims of this theorem as follows.
For the first claim, we
make a case distinction on the rule deriving $t \Greater_\lkbm s$:
\medskip

\noindent
\textsc{Rule \ref{itm:lkbm-wt}:}\enskip
From $\cal W_\m(t) > \cal W_\m(s)$, we have $\cal W_\m(t){\big|}_A > \cal W_\m(s){\big|}_A$,
and by Lemma~\ref{lem:weight-subst}, we get
$\cal W_\m(t\theta) > \cal W_\m(s\theta)$.
By definition of $W_\g$ and $W_\m$, they coincide on ground preterms,
and thus $\cal W_\g(t\theta) > \cal W_\g(s\theta)$.
So, rule \ref{itm:lkbg-wt} of $\Greater_\lkbg$ applies.

\medskip

\noindent
\textsc{Rules \ref{itm:lkbm-lam}, \ref{itm:lkbm-db}, \ref{itm:lkbm-sym}:}\enskip
We have $\cal W_\m(t) \ge \cal W_\m(s)$. If $\cal W_\m(t){\big|}_A > \cal W_\m(s){\big|}_A$,
rule~\ref{itm:lkbg-wt} applies as above. Otherwise, $\cal W_\m(t){\big|}_A = \cal
W(s){\big|}_A$, and the corresponding rule \ref{itm:lkbg-lam}, \ref{itm:lkbg-db},
or \ref{itm:lkbg-sym} applies. The only mismatches between the two definitions are
the use of $\Greater_\lkbm$ versus $\Greater_\lkbg$ and
$\gtgtsimlex_\lkbm$ versus $\gtgtlex_\lkbg$.
These are repaired by the induction
hypothesis and our observation above.

For the second claim, we
make a case distinction on the rule deriving $t \Greater_\lkbm s$:

\medskip

\noindent
\textsc{Rule \ref{itm:nslkbm-wt}:}\enskip
As for $\Greater_\lkbm$ above.

\medskip

\noindent
\textsc{Rule \ref{itm:nslkbm-var}:}\enskip
We will focus on the case where the argument lists $\bar t$ and $\bar s$ have
length 1 and the corresponding De Bruijn index in $y\theta$ occurs exactly once.
The same line of reasoning can be repeated for further arguments or further De
Bruijn index occurrences by appealing to the transitivity of $\Greatereq_\lkbg$.

Since $t_1$ is of nonfunctional type, $(y\> t_1)\theta$ and $(y\>
s_1)\theta$ must be of the forms $t' = \orangesubterm{u}{
t_1\theta{\uparrow}^k}$ and $s' = \orangesubterm{u}{
s_1\theta{\uparrow}^k}$, respectively, where $k$ is the context's depth.
We have $t_1 \Greatersim_\llpm s_1$ by the rule's condition.
By the induction
hypothesis, $t_1\theta \allowbreak\GreaterEQ_\lkbg s_1\theta$.
Thus, either $t' = s'$
or, by Theorem~\ref{thm:lkbg-compat-orange-contexts}, $t' \Greater_\lkbg s'$.

\medskip

\noindent
\textsc{Rules \ref{itm:nslkbm-lam}, \ref{itm:nslkbm-db}, \ref{itm:nslkbm-sym}:}\enskip
We have $\cal W_\m(t) \ge \cal W_\m(s)$. If $\cal W_\m(t){\big|}_A > \cal W_\m(s){\big|}_A$,
rule~\ref{itm:lkbg-wt} applies as above. Otherwise, $\cal W_\m(t){\big|}_A =
\cal W_\m(s){\big|}_A$. If $t\theta = s\theta$, there is nothing to prove.
Otherwise, the corresponding rule \ref{itm:lkbg-lam},
\ref{itm:lkbg-db}, or \ref{itm:lkbg-sym} applies. The rest of the proof is as
for $\Greater_\lkbm$ above.
\end{proof}

The converse of Theorem~\ref{thm:lkbm-grounding-subst-stable} does not hold.
However, it does hold on ground preterms:

\begin{thm}
\label{thm:lkbm-coincide-ground}
The relation $\Greater_\lkbm$ coincides with $\Greater_\lkbg$ on ground preterms.
\end{thm}

\begin{proof}
One direction of the equivalence follows by
Theorem~\ref{thm:lkbm-grounding-subst-stable}. It remains to show that $t \Greater_\lkbg
s$ implies $t \Greater_\lkbm s$. The proof is by induction on the definition of
$\Greater_\lkbg$. It is easy to see that to every case in the definition of $\Greater_\lkbg$
corresponds a case in the definition of $\Greater_\lkbm$. As for the weights, $\cal
W_\g$ and $\cal W_\m$ coincide. In particular, for a ground preterm, the polynomial
returned by $\cal W_\m$ contains no indeterminates.
To account for the mismatch between $\gtgtlex_\lkbg$ and $\gtgtsimlex_\lkbm$,
we apply Lemma~\ref{lem:nslkbm-reflexive}.
\end{proof}

\begin{thm}
\label{thm:nslkbm-coincide-ground}
The relation $\Greatersim_\lkbm$ coincides with $\GreaterEQ_\lkbg$ on ground preterms.
\end{thm}

\begin{proof}
One direction of the equivalence follows by
Theorem~\ref{thm:lkbm-grounding-subst-stable}. It remains to show that $t
\Greater_\lkbg s$ implies $t \Greatersim_\lkbm s$. If $t = s$, Lemma~\ref{lem:nslkbm-reflexive} applies. Otherwise, we appeal to
Theorem~\ref{thm:lkbm-coincide-ground} to obtain $t \Greater_\lkbm s$.
By Lemma~\ref{lem:lkbm-implies-nslkbm}, this implies $t
\Greatersim_\lkbm s$.
\end{proof}

\begin{lem}\label{lem:nslkbm-weights}
  If $t \Greatersim_\lkbm s$, then $\cal W_\m(t) \ge \cal W_\m(s)$.
\end{lem}
\begin{proof}
  We proceed by structural induction on $t$.

  For all rules except rule~\ref{itm:nslkbm-var}, the claim is obvious.
  If $t \Greatersim_\lkbm s$ was derived by rule~\ref{itm:nslkbm-var},
  then $t$ is of the form $y\> \bar t$ and $s$ is of the form $y\> \bar s$,
  and for every $i$, $t_i$ is steady and $t_i \Greatersim_\lkbm s_i$.
  By the induction hypothesis, $\cal W_\m(t_i) \ge \cal W_\m(s_i)$ for every $i$.
  Let $\upsilon_i$ be the type of $t_i$, which is also the type of $s_i$.
  Let $n$ be the length of $\bar t$, which is also the length of $\bar s$.
  Then,
  \begin{align*}
    \cal W_\m(t) &= 1 + \mathbf{w}_{y}
      + \sum\nolimits_{i=1}^{n}
        \mathbf{k}_{y,i} (\cal W_\m({t}_i) - \cal w_\db) \\
    &\ge 1 + \mathbf{w}_{y}
      + \sum\nolimits_{i=1}^{n}
        \mathbf{k}_{y,i} (\cal W_\m({s}_i) - \cal w_\db) \\
    &= \cal W_\m(s)
  \end{align*}
\end{proof}

\begin{thm}\label{thm:lkbm-transitive}
  If $u \Greatersim_\lkbm t$ and $t \Greatersim_\lkbm s$, then $u \Greatersim_\lkbm s$.
  If in addition $u \Greater_\lkbm t$ or $t \Greater_\lkbm s$, then even $u \Greater_\lkbm s$.
\end{thm}
\begin{proof}
We proceed by well-founded induction on the multiset $\{|u|, |t|, |s|\}$.

If $u \Greatersim_\lkbm t$ or $t \Greatersim_\lkbm s$ was derived by rule~\ref{itm:nslkbm-wt},
then rule~\ref{itm:lkbm-wt} yields $u \Greater_\lkbm s$ by transitivity of weight comparison
and Lemma~\ref{lem:nslkbm-weights}.
So, for the remainder of this proof, we may assume that $\cal W_\m(u) = \cal W_\m(t) = \cal W_\m(s)$.

If $u \Greatersim_\lkbm t$ was derived by rule~\ref{itm:nslkbm-var},
then $t \Greatersim_\lkbm s$ must have been derived by rule~\ref{itm:nslkbm-var}, too.
Then rule~\ref{itm:nslkbm-var} also yields $u \Greatersim_\lkbm s$
by the induction hypothesis.
Since all three preterms $u,s,t$ are headed by variables, neither
$u \Greater_\lkbm t$ nor $t \Greater_\lkbm s$ hold, and thus we need not prove
$u \Greater_\lkbm s$.

If $u \Greatersim_\lkbm t$ was derived by rule~\ref{itm:nslkbm-lam-typ} or \ref{itm:nslkbm-lam-body},
then $t \Greatersim_\lkbm s$ must be derived by rule~\ref{itm:nslkbm-lam}.
If $t \Greatersim_\lkbm s$ was derived by rule~\ref{itm:nslkbm-lam-typ} or \ref{itm:nslkbm-lam-body}
as well, then rule~\ref{itm:nslkbm-lam-typ} or \ref{itm:nslkbm-lam-body}
also yield $u \Greatersim_\lkbm s$ by the induction hypothesis and
by transitivity of 
the strict first-order KBO $\Greater_\ty$ on $\Tm{\Sigma_\ty}{\emptyset}$.
If in addition $u \Greater_\lkbm t$ or $t \Greater_\lkbm s$,
then this must be by rule~\ref{itm:lkbm-lam-typ} or \ref{itm:lkbm-lam-body}.
Then rule~\ref{itm:lkbm-lam-typ} or \ref{itm:lkbm-lam-body}
yield $u \Greater_\lkbm s$ by the induction hypothesis and by transitivity of $\Greater_\ty$.
If $t \Greatersim_\lkbm s$ was derived by rule~\ref{itm:nslkbm-lam-other},
then rule~\ref{itm:lkbm-lam-other} yields  $u \Greater_\lkbm s$.

If $u \Greatersim_\lkbm t$ was derived by rule~\ref{itm:nslkbm-lam-other},
then $t \Greatersim_\lkbm s$ must be derived by rule~\ref{itm:nslkbm-db} or \ref{itm:nslkbm-sym}.
Then \ref{itm:lkbm-lam-other} yields $u \Greater_\lkbm s$.

If $u \Greatersim_\lkbm t$ was derived by rule~\ref{itm:nslkbm-db-diff} or~\ref{itm:nslkbm-db-args},
then $t \Greatersim_\lkbm s$ must be derived by rule~\ref{itm:nslkbm-db}.
If $t \Greatersim_\lkbm s$ was derived by rule~\ref{itm:nslkbm-db-diff} or~\ref{itm:nslkbm-db-args}
as well, then rule~\ref{itm:nslkbm-db-diff} or~\ref{itm:nslkbm-db-args}
also yield $u \Greatersim_\lkbm s$
by transitivity of $>$ on natural numbers
and by the induction hypothesis, which implies transitivity of $\gtgtapproxlex_\lkbm$ on the relevant preterms.
If in addition $u \Greater_\lkbm t$ or $t \Greater_\lkbm s$,
then this must be by rule~\ref{itm:lkbm-db-diff} or~\ref{itm:lkbm-db-args}.
Then rule~\ref{itm:lkbm-db-diff} or~\ref{itm:lkbm-db-args}
yield $u \Greater_\lkbm s$ by
by transitivity of $>$ on natural numbers and the induction hypothesis.
If $t \Greatersim_\lkbm s$ was derived by rule~\ref{itm:nslkbm-db-other},
then rule~\ref{itm:lkbm-db-other} yields $u \Greater_\lkbm s$.

If $u \Greatersim_\lkbm t$ was derived by rule~\ref{itm:nslkbm-db-other},
then $t \Greatersim_\lkbm s$ must be derived by rule~\ref{itm:nslkbm-sym}.
Then rule~\ref{itm:lkbm-db-other} yields $u \Greater_\lkbm s$.

If $u \Greatersim_\lkbm t$ was derived by rule~\ref{itm:nslkbm-sym},
then $t \Greatersim_\lkbm s$ must be derived by rule~\ref{itm:nslkbm-sym}, too.
Then rule~\ref{itm:nslkbm-sym} also yields $u \Greatersim_\lkbm s$ by transitivity
of the precedence $\Succ$,
by transitivity of $\Greater_\ty$ and its lexicographic extension,
and by the induction hypothesis, which implies transitivity of $\gtgtapproxlex_\lkbm$ on the relevant preterms.
If in addition $u \Greater_\lkbm t$ or $t \Greater_\lkbm s$,
then this must be by rule~\ref{itm:lkbm-sym}.
Then rule~\ref{itm:lkbm-sym}
yields $u \Greater_\lkbm s$ by transitivity
of the precedence $\Succ$,
by transitivity of $\Greater_\ty$ and its lexicographic extension,
and by the induction hypothesis.
\end{proof}

\begin{thm}\label{thm:lkbm-variable-guarantee}
Let $t \Greater_\lkbm s$.
Let $\theta$ be a substitution such that all variables in
$t\theta$ and $s\theta$ are nonfunctional.
Let $s\theta$ contain a nonfunctional variable $x$
outside of parameters.
Then $t\theta$ must also contain $x$ outside of parameters.
\end{thm}
\begin{proof}
Since $t \Greater_\lkbm s$, we have $\cal W_\m(t) \geq \cal W_\m(s)$.
By Lemma~\ref{lem:weight-subst}, we have $\cal W_\m(t\theta) \geq \cal W_\m(s\theta)$.
By definition of $\cal W_\m$, 
since all variables in $s\theta$ are nonfunctional,
$W_\m(s\theta)$ must contain $\mathbf{w}_{x}$ with a nonzero coefficient.
Since $\cal W_\m(t\theta) \geq \cal W_\m(s\theta)$,
$t\theta$ must also contain $\mathbf{w}_{x}$ with a nonzero coefficient.
Therefore, $x$ must occur outside of parameters in $t\theta$.
\end{proof}

\subsection{\texorpdfstring{$\bm{\lambda}$LPO}{Lambda LPO}}
\label{ssec:lambda-lpo-monomorphic}

\begin{defi}
\label{def:llpm}
Let $\Succ^\ty$ be a precedence on $\Sigma_\ty$.
Let $\Greater_\ty$ be the strict first-order LPO on $\Tm{\Sigma_\ty}{\emptyset}$
induced by $\Succ^\ty$.
Let $\Succ$ be a precedence on $\Sigma$.
Let $\cst{ws} \in \Sigma$ be the watershed.

The \emph{strict monomorphic $\lambda$LPO} $\Greater_\llpm$ and the \emph{nonstrict
monomorphic $\lambda$LPO} $\Greatersim_\llpm$ induced by $\Succ^\ty, \Succ$ on
$\HTmp{\Sigma_\ty}{\Sigma}{\emptyset}{X}$ are defined by mutual induction. The strict
relation is defined so that $t \Greater_\llpm s$ if any of these conditions is met:

\begin{enumerate}
\item \label{itm:llpm-sym}
  $t$ is of the form $\cst{g}\langle\bar\upsilon\rangle(\bar w)\> \bar t_k$ and
  any of these conditions is met:

  \begin{enumerate}
  \item \label{itm:llpm-sym-sub}
    $t_i \Greatersim_\llpm s$ for some $i \in \{1,\dots,k\}$, or
  \item \label{itm:llpm-sym-diff}
    $s = \cst{f}(\bar u)\> \bar s$, $\cst{g} \Succ \cst{f}$, and
    $\cal{chkargs}(t, \bar s)$, or
  \item \label{itm:llpm-sym-types}
    $s = \cst{g}\langle\bar\tau\rangle(\bar u)\> \bar s$, $\bar\upsilon \gtgtlex_\ty \bar\tau$,
    and $\cal{chkargs}(t, \bar s)$, or
  \item \label{itm:llpm-sym-args}
    $s = \cst{g}\langle\bar\upsilon\rangle(\bar u)\> \bar s$,
    $(\bar w, \bar t) \gtgtsimlex_\llpm (\bar u, \bar s)$,
    and $\cal{chkargs}(t, \bar s)$, or
  \item \label{itm:llpm-sym-other}
    $\cst{g} \Succ \cst{ws}$ and
    $s$ is either of the form $\DB{m}\>\bar s$ and $\cal{chkargs}(t, \bar s)$
    or of the form $\lambda\> s'$ and $\cal{chkargs}(t, [s'])$;
  \end{enumerate}

\item \label{itm:llpm-db}
  $t$ is of the form $\DB{n}\> \bar t_k$ and any of these conditions is met:

  \begin{enumerate}
  \item \label{itm:llpm-db-sub}
    $t_i \Greatersim_\llpm s$ for some $i \in \{1,\dots,k\}$, or
  \item \label{itm:llpm-db-diff}
    $s = \DB{m}\> \bar s$, $n > m$, and $\cal{chkargs}(t, \bar s)$, or
  \item \label{itm:llpm-db-args}
    $s = \DB{n}\> \bar s$, $\bar t \gtgtsimlex_\llpm \bar s$, and
    $\cal{chkargs}(t, \bar s)$, or
  \item \label{itm:llpm-db-other}
    $s$ is of the form $\lambda\> s'$ and $\cal{chkargs}(t, [s'])$ or
    of the form or $\cst{f}(\bar u)\> \bar s$,
    where $\cst{f} \Preceq \cst{ws}$,
    and $\cal{chkargs}(t, \bar s)$;
  \end{enumerate}

\item \label{itm:llpm-lam}
  $t$ is of the form $\lambda\langle\upsilon\rangle\> t'$ and any of these
  conditions is met:

  \begin{enumerate}
  \item \label{itm:llpm-lam-sub}
    $t' \Greatersim_\llpm s$, or
  \item \label{itm:llpm-lam-types}
    $s = \lambda\langle\tau\rangle\> s'$, $\upsilon \Greater_\ty \tau$, and
    $\cal{chkargs}(t, [s'])$, or
  \item \label{itm:llpm-lam-bodies}
    $s = \lambda\langle\upsilon\rangle\> s'$ and $t' \Greater_\llpm s'$, or
  \item \label{itm:llpm-lam-other}
    $s$ is of the form $\cst{f}(\bar u)\> \bar s$,
    where $\cst{f} \Preceq \cst{ws}$,
    and $\cal{chkargs}(t, \bar s)$
  \end{enumerate}
\end{enumerate}
where $\cal{chkargs}(t, \bar s_k)$ if and only if $t \Greater_\llpm s_i$ for every $i \in
\{1,\dots,k\}$.
The nonstrict relation is defined so that $t \Greatersim_\llpm s$ if any of these
conditions is met:

\begin{enumerate}
\item \label{itm:nsllpm-var}
  $t$ is of the form $y\> \bar t$, $s$ is of the form $y\> \bar s$, and for
  every $i$, $t_i$ is steady and $t_i \Greatersim_\llpm s_i$;

\item \label{itm:nsllpm-sym}
  $t$ is of the form $\cst{g}\langle\bar\upsilon\rangle(\bar w)\> \bar t_k$ and
  any of these conditions is met:

  \begin{enumerate}
  \item \label{itm:nsllpm-sym-sub}
    $t_i \Greatersim_\llpm s$ for some $i \in \{1,\dots,k\}$, or
  \item \label{itm:nsllpm-sym-diff}
    $s = \cst{f}(\bar u)\> \bar s$, $\cst{g} \Succ \cst{f}$, and
    $\cal{chkargs}(t, \bar s)$, or
  \item \label{itm:nsllpm-sym-types}
    $s = \cst{g}\langle\bar\tau\rangle(\bar u)\> \bar s$, $\bar\upsilon \gtgtlex_\ty \bar\tau$,
    and $\cal{chkargs}(t, \bar s)$, or
  \item \label{itm:nsllpm-sym-args}
    $s = \cst{g}\langle\bar\upsilon\rangle(\bar u)\> \bar s$, $(\bar w, \bar t) \gtgtapproxlex_\llpm (\bar u, \bar s)$,
    and $\cal{chkargs}(t, \bar s)$, or
  \item \label{itm:nsllpm-sym-other}
    $\cst{g} \Succ \cst{ws}$ and
    $s$ is either of the form $\DB{m}\>\bar s$ and $\cal{chkargs}(t, \bar s)$
    or of the form $\lambda\> s'$ and $\cal{chkargs}(t, [s'])$;
  \end{enumerate}

\item \label{itm:nsllpm-db}
  $t$ is of the form $\DB{n}\> \bar t_k$ and any of these conditions is met:

  \begin{enumerate}
  \item \label{itm:nsllpm-db-sub}
    $t_i \Greatersim_\llpm s$ for some $i \in \{1,\dots,k\}$, or
  \item \label{itm:nsllpm-db-diff}
    $s = \DB{m}\> \bar s$, $n > m$, and $\cal{chkargs}(t, \bar s)$, or
  \item \label{itm:nsllpm-db-args}
    $s = \DB{n}\> \bar s$, $\bar t \gtgtapproxlex_\llpm \bar s$, and
    $\cal{chkargs}(t, \bar s)$, or
  \item \label{itm:nsllpm-db-other}
    $s$ is of the form $\lambda\> s'$ and $\cal{chkargs}(t, [s'])$ or
    of the form or $\cst{f}(\bar u)\> \bar s$,
    where $\cst{f} \Preceq \cst{ws}$,
    and $\cal{chkargs}(t, \bar s)$;
  \end{enumerate}

\item \label{itm:nsllpm-lam}
  $t$ is of the form $\lambda\langle\upsilon\rangle\> t'$ and any of these
  conditions is met:

  \begin{enumerate}
  \item \label{itm:nsllpm-lam-sub}
    $t' \Greatersim_\llpm s$, or
  \item \label{itm:nsllpm-lam-types}
    $s = \lambda\langle\tau\rangle\> s'$, $\upsilon \Greater_\ty \tau$, and
    $\cal{chkargs}(t, [s'])$, or
  \item \label{itm:nsllpm-lam-bodies}
    $s = \lambda\langle\upsilon\rangle\> s'$ and $t' \Greatersim_\llpm s'$, or
  \item \label{itm:nsllpm-lam-other}
    $s$ is of the form $\cst{f}(\bar u)\> \bar s$,
    where $\cst{f} \Preceq \cst{ws}$,
    and $\cal{chkargs}(t, \bar s)$
  \end{enumerate}
\end{enumerate}
where $\cal{chkargs}(t, \bar s_k)$ is defined as above.
\end{defi}

The only syntactic differences between the definitions of $\Greater_\llpg$ and $\Greater_\llpm$
are 
that $\Greater_\llpm$ uses $\Greatersim_\llpm$ instead of $\GreaterEQ_\llpg$
and $\gtgtsimlex_\llpm$ instead of $\gtgtlex_\llpg$.
Moreover, rule~\ref{itm:nsllpm-var} of $\Greatersim_\llpm$ is analogous to
rule~\ref{itm:nslkbm-var} in the definition of $\Greatersim_\lkbm$. As for rules
\ref{itm:nsllpm-sym}--\ref{itm:nsllpm-lam} of $\Greatersim_\llpm$, they are nearly
identical to the rules defining the strict orders $\Greater_\llpg$ and $\Greater_\llpm$.

Analogous theorems to those about $\Greater_\lkbm$ and $\Greatersim_\lkbm$ also hold about
$\Greater_\llpm$ and $\Greatersim_\llpm$.

\begin{lem}
\label{lem:nsllpm-reflexive}
$s \Greatersim_\llpm s$ for every monomorphic preterm $s$.
\end{lem}

\begin{proof}
By straightforward induction on $s$.
\end{proof}

\begin{lem}\label{lem:llpm-implies-nsllpm}
  If $t \Greater_\llpm s$, then $t \Greatersim_\llpm s$.
\end{lem}
\begin{proof}
  By induction on the derivation of $t \Greater_\llpm s$.
  For each rule, there is a clearly corresponding rule of $\Greatersim_\llpm$ to apply.
  The induction hypothesis is only required for applying rule~\ref{itm:nsllpm-lam-bodies}.
  A crucial observation is that $\gtgtsimlex_\llpm$ implies $\gtgtapproxlex_\llpm$ by definition.
\end{proof}

\begin{thm}
\label{thm:llpm-grounding-subst-stable}
If $t \Greater_\llpm s$, then $t\theta \Greater_\llpg s\theta$ for any grounding substitution
$\theta$.
If $t \Greatersim_\llpm s$, then $t\theta \GreaterEQ_\llpg s\theta$ for any grounding
substitution $\theta$.
\end{thm}

\begin{proof}
The proof of the two claims is by induction on the shape of the derivation of $t
\Greater_\llpm s$ and $t \Greatersim_\llpm s$. 
As in the proof of Theorem~\ref{thm:lkbm-grounding-subst-stable},
we observe that
\begin{itemize}
  \item
  $\bar t \gtgtsimlex_\llpm \bar s$ implies $\bar t\theta \gtgtlex_\llpg \bar s\theta$, and
  \medskip
  \item
  $\bar t \gtgtapproxlex_\llpm \bar s$ implies 
  $\bar t\theta \gtgtlex_\llpg \bar s\theta$ or
  $\bar t\theta = \bar s\theta$.
\end{itemize}

For the first claim,
we make a case distinction on the rule deriving $t \Greater_\llpm s$.
In each case, the corresponding rule of $\Greater_\llpg$ applies.
The only mismatches between the two definitions are
the use of $\Greater_\lkbm$ versus $\Greater_\lkbg$ and
$\gtgtsimlex_\lkbm$ versus $\gtgtlex_\lkbg$.
These are repaired by the induction
hypothesis and our observation above.

For the second claim,
we make a case distinction on the rule deriving $t \Greatersim_\llpm s$:

\medskip

\noindent
\textsc{Rule \ref{itm:nsllpm-var}:}\enskip
Analogous to the case for rule~\ref{itm:nslkbm-var} of $\Greatersim_\lkbm$ in the proof of
Theorem~\ref{thm:lkbm-grounding-subst-stable}.

\medskip

\noindent
\textsc{Rules \ref{itm:nsllpm-sym}, \ref{itm:nsllpm-db}, \ref{itm:nsllpm-lam}:}\enskip
If $t\theta = s\theta$, there is nothing to prove. Otherwise, the corresponding
rule \ref{itm:llpg-sym}, \ref{itm:llpg-db}, or \ref{itm:llpg-lam} applies. The
rest of the proof is as for $\Greater_\llpm$.
\end{proof}

\begin{thm}
\label{thm:llpm-coincide-ground}
The relation $\Greater_\llpm$ coincides with $\Greater_\llpg$ on ground preterms.
\end{thm}

\begin{proof}
One direction of the equivalence follows by
Theorem~\ref{thm:llpm-grounding-subst-stable}. It remains to show that $t \Greater_\llpg
s$ implies $t \Greater_\llpm s$. The proof is by induction on the definition of
$\Greater_\llpg$. It is easy to see that to every case in the definition of $\Greater_\llpg$
corresponds a case in the definition of $\Greater_\llpm$.
To account for the mismatches
between $\GreaterEQ_\lkbg$ and $\Greatersim_\lkbm$ and
and
between $\gtgtlex_\llpg$ and $\gtgtsimlex_\llpm$,
we apply Lemmas~\ref{lem:nsllpm-reflexive} and \ref{lem:llpm-implies-nsllpm}.
\end{proof}

\begin{thm}
\label{thm:nsllpm-coincide-ground}
The relation $\Greatersim_\llpm$ coincides with $\GreaterEQ_\llpg$ on ground preterms.
\end{thm}

\begin{proof}
One direction of the equivalence follows by
Theorem~\ref{thm:llpm-grounding-subst-stable}. It remains to show that $t
\Greater_\llpg s$ implies $t \Greatersim_\llpm s$. If $t = s$, Lemma~\ref{lem:nsllpm-reflexive} applies. Otherwise, we appeal to
Theorem~\ref{thm:llpm-coincide-ground} to obtain $t \Greater_\llpm s$.
By Lemma~\ref{lem:llpm-implies-nsllpm}, this implies $t
\Greatersim_\llpm s$.
\end{proof}

\begin{thm}\label{thm:llpm-transitive}
  If $u \Greatersim_\llpm t$ and $t \Greatersim_\llpm s$, then $u \Greatersim_\llpm s$.
  If in addition $u \Greater_\llpm t$ or $t \Greater_\llpm s$, then even $u \Greater_\llpm s$.
\end{thm}
\begin{proof}
We proceed by well-founded induction on the multiset $\{|u|, |t|, |s|\}$.

If $u \Greatersim_\llpm t$ was derived by rule~\ref{itm:nsllpm-var},
then $t \Greatersim_\llpm s$ must have been derived by rule~\ref{itm:nsllpm-var}, too.
Then rule~\ref{itm:nsllpm-var} also yields $u \Greatersim_\llpm s$
by the induction hypothesis.

If $u \Greatersim_\llpm t$ was derived by rule~\ref{itm:nsllpm-sym-sub},
\ref{itm:nsllpm-db-sub}, or \ref{itm:nsllpm-lam-sub},
then $u \Greater_\llpm s$ by rule~\ref{itm:llpm-sym-sub},
\ref{itm:llpm-db-sub}, or \ref{itm:llpm-lam-sub} and the induction hypothesis.

If $u \Greatersim_\llpm t$ was derived by rule~\ref{itm:nsllpm-sym-diff},
\ref{itm:nsllpm-sym-types}, or \ref{itm:nsllpm-sym-args},
then $t \Greatersim_\llpm s$ must have been derived by rule~\ref{itm:nsllpm-sym}, too.
If $t \Greatersim_\llpm s$ was derived by rule~\ref{itm:nsllpm-sym-sub},
then the $\cal{chkargs}$-condition and the induction hypothesis
yield $u \Greater_\llpm s$.
If $t \Greatersim_\llpm s$ was derived by rule~\ref{itm:nsllpm-sym-diff},
\ref{itm:nsllpm-sym-types}, or \ref{itm:nsllpm-sym-args},
then rule \ref{itm:nsllpm-sym-diff},
\ref{itm:nsllpm-sym-types}, or \ref{itm:nsllpm-sym-args}
also yield $u \Greatersim_\llpm s$
by transitivity
of the precedence $\Succ$,
by transitivity of $\Greater_\ty$ and its lexicographic extension,
and by the induction hypothesis, which implies 
transitivity of $\gtgtapproxlex_\llpm$ on the relevant preterms
and the required $\cal{chkargs}$-condition.
If moreover $u \Greater_\llpm t$ or $t \Greater_\llpm s$,
we can similarly derive $u \Greater_\llpm s$ by
rule \ref{itm:llpm-sym-diff},
\ref{itm:llpm-sym-types}, or \ref{itm:llpm-sym-args}.
If $t \Greatersim_\llpm s$ was derived by rule~\ref{itm:nsllpm-sym-other},
then $u \Greater_\llpm s$ by
rule \ref{itm:llpm-sym-other},
using transitivity of the precedence $\Succ$ and the induction hypothesis.

If $u \Greatersim_\llpm t$ was derived by rule~\ref{itm:nsllpm-sym-other},
then $t \Greatersim_\llpm s$ must have been derived by rule~\ref{itm:nsllpm-db}
or \ref{itm:nsllpm-lam}.
If $t \Greatersim_\llpm s$ was derived by rule~\ref{itm:nsllpm-db-sub}
or \ref{itm:nsllpm-lam-sub},
then the $\cal{chkargs}$-condition and the induction hypothesis
yield $u \Greater_\llpm s$.
If $t \Greatersim_\llpm s$ was derived by rule~\ref{itm:nsllpm-db-diff}
\ref{itm:nsllpm-db-args}, \ref{itm:nsllpm-db-other},
\ref{itm:nsllpm-lam-types}, \ref{itm:nsllpm-lam-bodies}, or \ref{itm:nsllpm-lam-other},
then $s$ is of the form $m\>\bar u$, $\lambda\>u'$ or $\cst{f}(\bar u)\> \bar v$, where $\cst{f} \Preceq \cst{ws}$.
If it is of the form $m\>\bar u$ or $\lambda\>u'$,
then rule~\ref{itm:llpm-sym-other} yields $u \Greater_\llpm s$,
where the $\cal{chkargs}$-condition is satisfied by the induction hypothesis.
If it is of the form $\cst{f}(\bar u)\> \bar v$ with $\cst{f} \Preceq \cst{ws}$,
rule~\ref{itm:llpm-sym-diff} yields $u \Greater_\llpm s$,
using transitivity of the precedence $\Succ$ and the induction hypothesis.

If $u \Greatersim_\llpm t$ was derived by rule~\ref{itm:nsllpm-db-diff}
or \ref{itm:nsllpm-db-args},
then $t \Greatersim_\llpm s$ must have been derived by rule~\ref{itm:nsllpm-db}.
If $t \Greatersim_\llpm s$ was derived by rule~\ref{itm:nsllpm-db-sub},
then the $\cal{chkargs}$-condition and the induction hypothesis
yield $u \Greater_\llpm s$.
If $t \Greatersim_\llpm s$ was derived by rule~\ref{itm:nsllpm-db-diff} or
\ref{itm:nsllpm-db-args},
then rule~\ref{itm:nsllpm-db-diff} or \ref{itm:nsllpm-db-args} also yield $u \Greatersim_\llpm s$,
using transitivity of $>$ on natural numbers and the induction hypothesis,
which implies transitivity of $\gtgtapproxlex_\llpm$ on the relevant preterms
and the required $\cal{chkargs}$-condition.
If moreover $u \Greater_\llpm t$ or $t \Greater_\llpm s$,
we can similarly derive $u \Greater_\llpm s$ by
rule \ref{itm:llpm-db-diff} or \ref{itm:llpm-db-args}.
If $t \Greatersim_\llpm s$ was derived by rule~\ref{itm:nsllpm-db-other},
then rule~\ref{itm:llpm-db-other} yields $u \Greater_\llpm s$,
using the induction hypothesis to discharge the $\cal{chkargs}$-condition.

If $u \Greatersim_\llpm t$ was derived by rule~\ref{itm:nsllpm-db-other},
then $t \Greatersim_\llpm s$ must have been derived by
rule~\ref{itm:nsllpm-sym} (but not rule~\ref{itm:nsllpm-sym-other}) or \ref{itm:nsllpm-lam}.
If $t \Greatersim_\llpm s$ was derived by rule~\ref{itm:nsllpm-sym-sub}
or \ref{itm:nsllpm-lam-sub},
then the $\cal{chkargs}$-condition and the induction hypothesis
yield $u \Greater_\llpm s$.
If $t \Greatersim_\llpm s$ was derived by rule~\ref{itm:nsllpm-sym-diff},
\ref{itm:nsllpm-sym-types}, \ref{itm:nsllpm-sym-args},
\ref{itm:nsllpm-lam-types}, \ref{itm:nsllpm-lam-bodies}, or \ref{itm:nsllpm-lam-other},
then rule~\ref{itm:llpm-db-other}
yields $u \Greater_\llpm s$,
using transitivity of the precedence $\Succ$ and the induction hypothesis.

If $u \Greatersim_\llpm t$ was derived by rule~\ref{itm:nsllpm-lam-types}
or \ref{itm:nsllpm-lam-bodies},
then $t \Greatersim_\llpm s$ must have been derived by rule~\ref{itm:nsllpm-lam}.
If $t \Greatersim_\llpm s$ was derived by rule~\ref{itm:nsllpm-lam-sub},
then the $\cal{chkargs}$-condition and the induction hypothesis
yield $u \Greater_\llpm s$.
If $t \Greatersim_\llpm s$ was derived by rule~\ref{itm:nsllpm-lam-types} or
\ref{itm:nsllpm-lam-bodies}
then rule~\ref{itm:nsllpm-lam-types} or
\ref{itm:nsllpm-lam-bodies} yield $u \Greatersim_\llpm s$,
using transitivity of $\Greater_\ty$ and
the induction hypothesis.
If moreover $u \Greater_\llpm t$ or $t \Greater_\llpm s$,
we can similarly derive $u \Greater_\llpm s$ by
rule \ref{itm:llpm-lam-types} or \ref{itm:llpm-lam-bodies}.
If $t \Greatersim_\llpm s$ was derived by rule~\ref{itm:nsllpm-lam-other},
then rule~\ref{itm:llpm-lam-other} yields $u \Greater_\llpm s$,
using the induction hypothesis to discharge the $\cal{chkargs}$-condition.

If $u \Greatersim_\llpm t$ was derived by rule~\ref{itm:nsllpm-lam-other},
then $t \Greatersim_\llpm s$ must have been derived by
rule~\ref{itm:nsllpm-sym} (but not rule~\ref{itm:nsllpm-sym-other}).
If $t \Greatersim_\llpm s$ was derived by rule~\ref{itm:nsllpm-sym-sub},
then the $\cal{chkargs}$-condition and the induction hypothesis
yield $u \Greater_\llpm s$.
If $t \Greatersim_\llpm s$ was derived by rule~\ref{itm:nsllpm-sym-diff},
\ref{itm:nsllpm-sym-types}, or \ref{itm:nsllpm-sym-args},
then rule~\ref{itm:llpm-lam-other}
yields $u \Greater_\llpm s$,
using transitivity of the precedence $\Succ$ and the induction hypothesis.
\end{proof}

\begin{thm}\label{thm:llpm-variable-guarantee}
Let $t \Greater_\llpm s$ or $t \Greatersim_\llpm s$.
Let $\theta$ be a substitution such that all variables in
$t\theta$ and $s\theta$ are nonfunctional.
Let $s\theta$ contain a nonfunctional variable $x$
outside of parameters.
Then $t\theta$ must also contain $x$ outside of parameters.
\end{thm}
\begin{proof}
The proof of the two claims is by induction on the shape of the derivation of $t
\Greater_\llpm s$ or $t \Greatersim_\llpm s$.
In most cases, the claims follow directly form the induction hypothesis.
For case~\ref{itm:nsllpm-var}, we have $t = y\> \bar t$ and $s = y\> \bar s$
and for every $i$, $t_i$ is steady and $t_i \Greatersim_\llpm s_i$.
Our assumption is that $s\theta$ contains a nonfunctional variable $x$
outside of parameters.
The $x$ could originate from $y\theta$ or from $s_i\theta$ for some $i$.
If it originates from $y\theta$, then $x$ must also occur outside of parameters
in $t\theta = y\theta\> \bar t\theta$
because $t_i$ is steady for all $i$.
If it originates from $s_i\theta$, then $x$ must also occur in $t_i\theta$ outside of parameters
by the induction hypothesis because $t_i \Greatersim_\llpm s_i$.
Since $t_i$ is steady, $x$ must also occur in $t\theta$ outside of parameters.
\end{proof}

\section{The Polymorphic Level}
\label{sec:the-polymorphic-level}

In a third and final step, we generalize the definition of $\lambda$KBO and
$\lambda$LPO to polymorphic nonground preterms. The resulting orders coincide with
the monomorphic nonground $\lambda$KBO and $\lambda$LPO on monomorphic preterms
while supporting type variables.

Type variables, in conjunction with the $\eta$-long $\beta$-normal form, lead to
substantial complications. Instantiating a type variable with a functional type
causes $\eta$-expansion to take place, transforming for instance
$y\langle\alpha\rangle$ into $\lambda\> y\langle\beta\to\gamma\rangle\>\DB{0}$
or even $\lambda\> \lambda\> y\langle\beta\to(\beta\to\beta)\to\gamma\rangle\>\DB{1}\>(\lambda\> \DB{1}\> \DB{0})$.
This affects the weight calculation of $\lambda$KBO, since each $\eta$-expansion
increases the weight by $\cal w_\lambda + \cal w_\db$. Our solution is to add a
term to the polynomial to account for possible $\eta$-expansion. This also
affects the shape comparison of $\lambda$KBO and $\lambda$LPO, since the shape
of any preterm whose type is a type variable~$\alpha$ can change radically as a
result of instantiating~$\alpha$.

If we used the $\eta$-short $\beta$-normal form instead, we would be out of the
frying pan into the fire. Applying a $\lambda$-abstraction to an argument makes
not only the $\lambda$ but also De Bruijn indices disappear. Applying an
$\eta$-reduced functional term $t$ to an argument, however, makes neither a
$\lambda$ nor De Bruijn indices disappear, resulting in a weight discrepancy of
at least $\cal w_\lambda$ compared with a $\lambda$-abstraction of the same type
and weight. Moreover, the $\eta$-short normal form makes it more difficult, if
not impossible, to achieve another of our goals, namely, the order requirement
for the $\cst{diff}$ symbol of the optimistic $\lambda$-superposition calculus.

\subsection{\texorpdfstring{$\bm{\lambda}$KBO}{Lambda KBO}}
\label{ssec:lambda-kbo-polymorphic}

\begin{defi}
\label{def:polynomials-polymorphic}
Let $(\Sigma_\ty,\Sigma)$ be a higher-order signature.
We denote by $\Poly$ the set of $\Ord$-valued polynomials of the indeterminates
$\mathbf{w}_{y\>\bar t}$, $\mathbf{k}_{y\>\bar t,i}$, and $\mathbf{h}_\alpha$. The
first two are the same as in Definition~\ref{def:polynomials-monomorphic},
except that preterms are now polymorphic. The last one is as follows,
where $\alpha \in X_\ty$:
\begin{itemize}
\item $\mathbf{h}_\alpha$, ranging over $\Nat$, represents the number of
  $\eta$-expansions incurred as a result of instantiating $\alpha$
  for one preterm of type $\alpha$ (excluding any subterms).
\end{itemize}
Auxiliary concepts are defined as in
Definition~\ref{def:polynomials-monomorphic}.
\end{defi}

For example, if $\cst{c} : \alpha$ and $\alpha\theta = (\kappa \to \kappa)
\to \kappa$, then $\cst{c}\theta = \lambda\> \cst{c}\> (\lambda\> \DB{1}\> \DB{0})$.
In this case, instantiation caused two $\eta$-expansions, including one to
a De Bruijn index.

\begin{defi}
\label{def:normalize-lk-polymorphic}
Let $\Sigma' = \Sigma \uplus \{k \mid k \in \Ordpos\}$ with
$k : \Pi\alpha.\>\alpha$.
Define the normalization function $\NORMonly :
\HTmp{\Sigma_\ty}{\Sigma}{X_\ty}{X} \to \HTmp{\Sigma_\ty}{\Sigma'}{X_\ty}{X}$
recursively by
\begin{align*}
\NORM{y\> \bar t}
  & = y\> \NORM{\bar t}
\\
\NORM{\cst{f}(\bar u)\> \bar t}
  & =
\begin{cases}
  k\langle\tau\rangle\> \NORM{\bar t} & \text{if $\cal k(\cst{f}, i) = 1$ for every $i$, with $\cal w(\cst{f}) = k$ and $\cst{f}(\bar u) : \tau$}
\\
  \cst{f}(\bar u)\> \NORM{\bar t} & \text{otherwise}
\end{cases}
\\
\NORM{\DB{m}\langle\tau\rangle\> \bar t}
  & = \cal w_\db\langle\tau\rangle\> \NORM{\bar t}
\\
\NORM{\lambda\> t}
  & = \lambda\> \NORM{t}
\end{align*}
\end{defi}

\begin{defi}
\label{def:eta-weight-lkb}
Define the $\eta$-expansion polynomial $\cal H : \Ty{\Sigma_\ty}{X_\ty} \to
\Poly$ by
\begin{align*}
\cal H(\alpha) & = (\cal w_\lambda + \cal w_\db)\mathbf{h}_\alpha
&
\cal H(\kappa(\bar\tau)) & = 0
\end{align*}
\end{defi}

\begin{defi}
\label{def:weight-lkb}
Let $\cal w : \Sigma \to \Ordpos$, $\cal w_\lambda, \cal w_\db \in \Ordpos$,
and $\cal k : \Sigma \times \Natpos \to \Ordpos$.
For every $\cst{f} \in \Sigma$ and $i > \arity{\cst{f}}$, we require
$(\mathrm{K})$~$\cal k(\cst{f}, i) = 1$.
Given a list of preterms $\bar t$, let $\bar t\RIGHT$ denote the longest suffix
consisting of steady preterms, and let $\bar t\LEFT$ denote the complementary
prefix.
Define the polymorphic weight function $\cal W :
\HTmp{\Sigma_\ty}{\Sigma}{X_\ty}{X} \to \Poly$ recursively by
\begin{align*}
  \cal W(y\>\bar{t})
& = 1 + \mathbf{w}_{y\>\NORM{\bar{t}\LEFT}}
  + \sum\nolimits_{i=1}^{|\bar{t}\RIGHT|}
    \mathbf{k}_{y\>\NORM{\bar{t}\LEFT},i} (\cal W(\bar{t}\RIGHT_i) - \cal w_\db)
  + \cal H(\tau)
\\[-\jot]
&\phantom{=}\;\quad\text{if $y\>\bar{t} : \tau$}
\\
  \cal W(\cst{f}(\bar u)\> \bar t_n)
& = \cal w(\cst{f}) + \sum\nolimits_{i=1}^n \cal k(\cst{f}, i)
    \cal W(t_i) + \cal H(\tau)
\\[-\jot]
&\phantom{=}\;\quad\text{if $\cst{f}(\bar u)\> \bar t_n : \tau$}
\\
  \cal W(\DB{m}\> \bar t_n)
& = \cal w_\db + \sum\nolimits_{i=1}^n \cal W(t_i) + \cal H(\tau)
\\[-\jot]
&\phantom{=}\;\quad\text{if $\DB{m}\> \bar t_n : \tau$}
\\
  \cal W(\lambda\> t)
& = \cal w_\lambda + \cal W(t)
\end{align*}
\end{defi}

Notice, in the definition above, the presence of $\cal H(\tau)$ monomials to
account for $\eta$-expansion caused by type variable instantiation.

\begin{defi}
\label{def:compare-type-poly}
Let $\tau, \upsilon \in \Ty{\Sigma_\ty}{X_\ty}$. The polymorphism comparison
$\upsilon \unrhd \tau$ holds if $\tau$ is not a type variable or if
$\upsilon = \tau$.
Moreover, let $s, t \in \HTmp{\Sigma}{\Sigma_\ty}{X}{X_\ty}$ such that
$s : \tau$, $t : \upsilon$. We write $t \unrhd s$ if $\upsilon \unrhd \tau$.
\end{defi}

\begin{defi}
\label{def:lkb}
Let $\cal w_\ty, \cal w, \cal w_\lambda, \cal w_\db, \cal k,\allowbreak \cal W$
be as in Definition~\ref{def:weight-lkb}.
Let $\Succ^\ty$ be a precedence on $\Sigma_\ty$.
Let $\Greater_\ty$ be the strict first-order KBO on $\Tm{\Sigma_\ty}{X_\ty}$
induced by $\cal w_\ty$ and $\Succ^\ty$.
Let $\Succ$ be a precedence on $\Sigma$.

The \emph{strict polymorphic $\lambda$KBO} $\Greater_\lkb$ induced by $\cal w_\ty, \cal
w, \cal w_\lambda, \cal w_\db, \cal k,\allowbreak \Succ^\ty, \Succ$ on
$\HTmp{\Sigma_\ty}{\Sigma}{X_\ty}{X}$ is defined inductively so that $t \Greater_\lkb s$ if
\begin{enumerate}[(1)]
\item the rule
  \ref{itm:lkbm-wt},
  \ref{itm:lkbm-lam-typ},
  \ref{itm:lkbm-lam-body},
  \ref{itm:lkbm-db-args}, or
  \ref{itm:lkbm-sym-args}
  of the definition of $\Greater_\lkbm$ applies mutatis mutandis, or
\item the rule
  \ref{itm:lkbm-lam-other},
  \ref{itm:lkbm-db-diff},
  \ref{itm:lkbm-db-other},
  \ref{itm:lkbm-sym-diff}, or
  \ref{itm:lkbm-sym-typ}
  of the definition of $\Greater_\lkbm$ applies mutatis mutandis
  and $t \unrhd s$ holds.
\end{enumerate}
The \emph{nonstrict polymorphic $\lambda$KBO} $\Greatersim_\lkb$ induced by $\cal
w_\ty, \cal w, \cal w_\lambda, \cal w_\db, \cal k,\allowbreak \Succ^\ty, \Succ$
on $\HTmp{\Sigma_\ty}{\Sigma}{X_\ty}{X}$ is defined inductively so that $t \Greatersim_\lkb s$
if
\begin{enumerate}[(1)]
\item the rule
  \ref{itm:nslkbm-wt},
  \ref{itm:nslkbm-var},
  \ref{itm:nslkbm-lam-typ},
  \ref{itm:nslkbm-lam-body},
  \ref{itm:nslkbm-db-args}, or
  \ref{itm:nslkbm-sym-args}
  of the definition of $\Greatersim_\lkbm$ applies mutatis mutandis, or
\item the rule
  \ref{itm:nslkbm-lam-other},
  \ref{itm:nslkbm-db-diff},
  \ref{itm:nslkbm-db-other},
  \ref{itm:nslkbm-sym-diff}, or
  \ref{itm:nslkbm-sym-typ}
  of the definition of $\Greatersim_\lkbm$ applies mutatis mutandis
  and $t \unrhd s$ holds.
\end{enumerate}
\end{defi}

For some of the rules, the condition $t \unrhd s$ is necessary to compare
$t$~and~$s$. For the other rules, the condition can be derived from the types of
$t$~and~$s$, either because both are of nonvariable type or because they are of
the same type.

Below we will connect the polymorphic $\lambda$KBO with its monomorphic
counterpart to lift its properties, which in turn were lifted from the ground
$\lambda$KBO.

\begin{defi}
\label{def:assignment-of-monomorphizing-subst}
The polynomial substitution $\POLY{\theta}$ associated with a monomorphizing
type substitution $\theta$ maps indeterminate $\mathbf{w}_{y\>\bar t}$ to
$\mathbf{w}_{(y\>\bar t)\theta}$, indeterminate $\mathbf{k}_{y\>\bar t,i}$ to
$\mathbf{k}_{(y\>\bar t)\theta,i}$, and indeterminate $\mathbf{h}_\alpha$ to the
number of $\eta$-expansions incurred as a result of instantiating $\alpha$ for
one preterm of type $\alpha$ (excluding any subterms).
\end{defi}

\begin{lem}
\label{lem:weight-monomorphizing-subst}
\begin{sloppypar}
Given a monomorphizing type substitution~$\theta$, we have
$\cal W(t){\big|}_{\POLY{\theta}} = \cal W_\m(t\theta)$.
\end{sloppypar}
\end{lem}

\begin{proof}
Let $\sigma = \POLY{\theta}$. Let $t : \tau$, and let $k$ be the number of
$\eta$-expansions incurred as a result of applying $\theta$ on a term of type
$\tau$ (excluding any subterms). The proof is by induction on the definition
of~$\cal W$.

\medskip

\noindent
\textsc{Case $t = y\>\bar{t}$:}\enskip We have
\begin{align*}
      & \cal W(y\>\bar{t}){\big|}_{\sigma}
\\
{=}\; & 1 + {\mathbf{w}_{y\>\bar{t}\LEFT}{\big|}_{\sigma}
  + \sum\nolimits_{i=1}^{\smash{|\bar{t}\RIGHT|}}
    \mathbf{k}_{y\>\bar{t}\LEFT,i}{\big|}_{\sigma}
    (\cal W(\bar{t}\RIGHT_i){\big|}_{\sigma} - \cal w_\db)}
  + {\cal{H}(\tau){\big|}_{\sigma}}
\\[-\jot]
  & \quad\text{by definition of $\cal W$}
\\
{=}\; & 1 + {\mathbf{w}_{y\>\bar{t}\LEFT}{\big|}_{\sigma}
  + \sum\nolimits_{i=1}^{\smash{|\bar{t}\RIGHT|}}
    \mathbf{k}_{y\>\bar{t}\LEFT,i}{\big|}_{\sigma}
    (\cal W_\m(\bar{t}\RIGHT_i\theta) - \cal w_\db)}
  + {\cal{H}(\tau){\big|}_{\sigma}}
\\[-\jot]
  & \quad\text{by the induction hypothesis}
\\
{=}\; & 1 + {\mathbf{w}_{y\>\bar{t}\LEFT}{\big|}_{\sigma}
  + \sum\nolimits_{i=1}^{\smash{|\bar{t}\RIGHT|}}
    \mathbf{k}_{y\>\bar{t}\LEFT,i}{\big|}_{\sigma}
    (\cal W_\m(\bar{t}\RIGHT_i\theta) - \cal w_\db)}
  + (\cal w_\lambda + \cal w_\db)k
\\[-\jot]
  & \quad\text{by definition of $\cal H$ and the semantics of $\mathbf{h}$}
\\
{=}\; & 1 + \mathbf{w}_{(y\theta)\>(\bar{t}\theta)\LEFT} 
        + \sum\nolimits_{i=1}^{|(\bar{t}\theta)\RIGHT|}
        \mathbf{k}_{(y\theta)\>(\bar{t}\theta)\LEFT,i}
        (\cal W_\m(\bar{t}\RIGHT_i\theta) - \cal w_\db)
        + (\cal w_\lambda + \cal w_\db)k
\\[-\jot]
  & \quad\text{by definition of $\sigma$}
\\
{=}\; & \cal W_\m((y\>\bar{t})\theta)
\\[-\jot]
  & \quad\text{by definition of $\cal W_\m$ and substitution}
\end{align*}

\medskip

\noindent
\textsc{Case $t = \cst{f}(\bar u)\> \bar t_n$:}\enskip
We have
\begin{align*}
      & \cal W(\cst{f}(\bar u)\> \bar t_n){\big|}_{\sigma}
\\
{=}\; & \cal w(\cst{f}) + \sum\nolimits_{i=1}^n \cal k(\cst{f}, i)
    \cal W(t_i){\big|}_{\sigma} + \cal H(\tau){\big|}_{\sigma}
  && \text{by definition of $\cal W$}
\\
{=}\; & \cal w(\cst{f}) + \sum\nolimits_{i=1}^n \cal k(\cst{f}, i)
    \cal W_\m(t_i\theta) + \cal H(\tau){\big|}_{\sigma}
  && \text{by the induction hypothesis}
\\
{=}\; & \cal w(\cst{f}) + \sum\nolimits_{i=1}^n \cal k(\cst{f}, i)
    \cal W_\m(t_i\theta) + (\cal w_\lambda + \cal w_\db)k
  && \text{by definition of $\cal H$}
\\
{=}\; & \cal W_\m((\cst{f}(\bar u)\> \bar t_n)\theta)
  && \text{by definition of $\cal W_\m$, substitution, and}
\\[-\jot]
  &&& \text{hypothesis $(\mathrm{K})$}
\end{align*}

\medskip

\noindent
\textsc{Case $t = \DB{m}\> \bar t_n$:}\enskip
\noindent
Similar to the previous case.

\medskip

\noindent
\textsc{Case $t = \lambda\> t$:}\enskip
We have
\begin{align*}
      & \cal W(\lambda\> t){\big|}_{\sigma}
\\
{=}\; & \cal w_\lambda + \cal W(t){\big|}_{\sigma}
  && \text{by definition of $\cal W$}
\\
{=}\; & \cal w_\lambda + \cal W_\m(t\theta)
  && \text{by the induction hypothesis}
\\
{=}\; & \cal W_\m(\lambda\> (t\theta))
  && \text{by definition of $\cal W_\m$}
\\
{=}\; & \cal W_\m((\lambda\> t)\theta)
  && \text{by definition of substitution}
\qedhere
\end{align*}
\end{proof}

\begin{defi}
\label{def:truncating-substition}
Let $\sigma$ be a substitution. Given a preterm $t$, let $t\sigma?$ denote its
\emph{truncating substitution}, in which any outermost $\lambda$s introduced
due to $\eta$-expansion as a result of applying $\sigma$ to $t$ are omitted.
(In contrast, any introduced De Bruijn indices are kept.)
\end{defi}

For example, if $\cst{c} : \alpha$ and $\alpha\sigma = \kappa \to \kappa \to
\kappa$, then $\cst{c}\sigma? = \cst{c}\>\DB{1}\>\DB{0}$, whereas
$\cst{c}\sigma = \lambda\>\lambda\>\cst{c}\>\DB{1}\>\DB{0}$.

\begin{thm}
\label{thm:lkb-monomorphizing-subst-stable}
If $t \Greater_\lkb s$, then $t\theta \Greater_\lkbm s\theta$ for any monomorphizing
type substitution~$\theta$.
If $t \Greatersim_\lkb s$, then $t\theta \Greatersim_\lkbm s\theta$ for any
monomorphizing type substitution~$\theta$.

\end{thm}

\begin{proof}
The proof of the two claims is by induction on the shape of the derivation of $t \Greater_\lkb s$
and $t \Greatersim_\lkb s$.

\medskip
For the first claim, we proceed by case distinction on the rule deriving $t \Greater_\lkb s$:

\noindent
\textsc{Rule \ref{itm:lkbm-wt}:}\enskip
From $\cal W(t) > \cal W(s)$,
by Lemma~\ref{lem:weight-monomorphizing-subst}, we have
$\cal W_\m(t\theta_1) > \cal W_\m(s\theta_1)$.
Thus, rule \ref{itm:lkbm-wt} of $\Greater_\lkbm$ applies.

\medskip

\noindent
\textsc{Rules \ref{itm:lkbm-lam}, \ref{itm:lkbm-db}, \ref{itm:lkbm-sym}:}\enskip
We have $\cal W(t) \ge \cal W(s)$. We also have $t \unrhd s$, either because
the rule requires it or because it follows from the types of $t$~and~$s$. We
perform a case analysis on $t \unrhd s$.

\medskip

\noindent
\textsc{Subcase 1}, where $t$ and $s$ are of nonvariable types:\enskip
The corresponding rule \ref{itm:lkbm-lam}, \ref{itm:lkbm-db}, or
\ref{itm:lkbm-sym} for $\Greater_\lkbm$ applies. The only mismatch between the two
definitions is the use of $\Greater_\lkb$ versus $\Greater_\lkbm$, and it is repaired by the
induction hypothesis.

\medskip

\noindent
\textsc{Subcase 2}, where $t$ has some variable type $\alpha$ but not
$s$:\enskip If $\alpha\theta$ is a function type, then applying $\theta$ to $t$
results in some $\eta$-expansion, which leads to a heavier weight;
rule~\ref{itm:lkbm-wt} then applies. Otherwise, $\alpha\theta$ is not a
function type, and the reasoning is as for subcase~1.

\medskip

\noindent
\textsc{Subcase 3}, where $t$ and $s$ have some variable type $\alpha$:\enskip
Let $k$ be the number of curried arguments expected by values of type
$\alpha\theta$. This means that we have
\begin{align*}
  t\theta & = \underbrace{\lambda\ldots\lambda}_{k~\text{times}}\>
    (t\theta){\uparrow}^k\> \lnf{(\DB{k-1})}\>\ldots\>\lnf{\DB{0}}
& s\theta & = \underbrace{\lambda\ldots\lambda}_{k~\text{times}}\>
  (s\theta){\uparrow}^k\> \lnf{(\DB{k-1})}\>\ldots\>\lnf{\DB{0}}
\end{align*}
First, we apply rule~\ref{itm:lkbm-lam-body} $k$~times to remove the
leading $\lambda$s on both sides. It remains to show that $t\theta? \Greater_\lkbm
s\theta?$. By inspection of the rules
\ref{itm:lkbm-db-diff}~and~\ref{itm:lkbm-db-args} of $\Greater_\lkb$,
we find that $t{\uparrow}^k \Greater_\lkb s{\uparrow}^k$. For each of the $\Greater_\lkb$ rules that
could have been used to establish this, we can check that the corresponding
$\Greater_\lkbm$ rule is applicable. The additional De Bruijn indices
$\lnf{(\DB{k-1})}\>\ldots\>\lnf{\DB{0}}$ on both sides are harmless.

\medskip

The proof of the second claim is analogous.
\end{proof}

\begin{thm}
\label{thm:lkb-coincide-monomorphic}
The relation $\Greater_\lkb$ coincides with $\Greater_\lkbm$ on monomorphic preterms.
The relation $\Greatersim_\lkb$ coincides with $\Greatersim_\lkbm$ on monomorphic preterms.
\end{thm}

\begin{proof}
One direction of the equivalences follows by
Theorem~\ref{thm:lkb-monomorphizing-subst-stable}. It remains to show that $t
\Greater_\lkbm s$ implies $t \Greater_\lkb s$ and that $t \Greatersim_\lkbm s$
implies $t \Greatersim_\lkb s$ . The proof is by induction on the definition of
$\Greater_\lkbm$ and $\Greatersim_\lkbm$. It is easy to see that to every case
in the definition of $\Greater_\lkbm$ corresponds a case in the definition of
$\Greater_\lkb$ and every case in the definition of $\Greatersim_\lkbm$
corresponds a case in the definition of $\Greatersim_\lkb$. For the weights,
$\cal W_\m$ and $\cal W$ coincide. In particular, for a monomorphic preterm, the
polynomial returned by $\cal W$ contains no $\mathbf{h}_\alpha$ indeterminates.
\end{proof}

\begin{lem}
\label{lem:nlkb-implies-lkb}
If $t \Greater_\lkb s$, then $t \Greatersim_\lkb s$.
\end{lem}
\begin{proof}
Analogous to Lemma~\ref{lem:lkbm-implies-nslkbm}.
\end{proof}

\begin{thm}
\label{thm:lkb-transitive}
If $t \Greatersim_\lkb u$ and $u \Greatersim_\lkb s$, then $t \Greatersim_\lkb s$.
If moreover $t \Greater_\lkb u$ or $u \Greater_\lkb s$, then $t \Greater_\lkb s$.
\end{thm}
\begin{proof}
Analogous to Theorem~\ref{thm:lkbm-transitive},
using the fact that $\unrhd$ is transitive.
\end{proof}

\begin{thm}\label{thm:lkb-variable-guarantee}
Let $t \Greater_\lkb s$.
Let $\theta$ be a substitution such that all variables in
$t\theta$ and $s\theta$ are nonfunctional term variables.
Let $s\theta$ contain a nonfunctional variable $x$
outside of parameters.
Then $t\theta$ must also contain $x$ outside of parameters.
\end{thm}
\begin{proof}
By Theorems \ref{thm:lkbm-variable-guarantee} and \ref{thm:lkb-monomorphizing-subst-stable}.
\end{proof}

\subsection{\texorpdfstring{$\bm{\lambda}$LPO}{Lambda LPO}}
\label{ssec:lambda-lpo-polymorphic}

\begin{defi}
\label{def:llp}
Let $\Succ^\ty$ be a precedence on $\Sigma_\ty$.
Let $\Greater_\ty$ be the strict first-order LPO on $\Tm{\Sigma_\ty}{X_\ty}$
induced by $\Succ^\ty$.
Let $\Succ$ be a precedence on $\Sigma$.
Let $\cst{ws} \in \Sigma$ be the watershed.

The \emph{strict polymorphic $\lambda$LPO} $\Greater_\llp$ and the \emph{nonstrict
polymorphic $\lambda$LPO} $\Greatersim_\llp$ induced by $\Succ^\ty, \Succ$ on
$\HTmp{\Sigma_\ty}{\Sigma}{X_\ty}{X}$ are defined by mutual induction. The strict
relation is defined so that $t \Greater_\llp s$ if
\begin{enumerate}[(1)]
\item the rule
  \ref{itm:llpm-sym-sub},
  \ref{itm:llpm-sym-args},
  \ref{itm:llpm-sym-other},
  \ref{itm:llpm-db-sub},
  \ref{itm:llpm-db-args},
  \ref{itm:llpm-lam-sub},
  \ref{itm:llpm-lam-types}, or
  \ref{itm:llpm-lam-bodies}
  of the definition of $\Greater_\llp$ applies mutatis mutandi, or
\item the rule
  \ref{itm:llpm-sym-diff} or
  \ref{itm:llpm-sym-types}
  of the definition of $\Greater_\llp$ applies mutatis mutandi
  and either $\cst{g} \Succ \cst{ws}$ or $t \unrhd s$ holds, or
\item the rule
  \ref{itm:llpm-db-diff},
  \ref{itm:llpm-db-other}, or
  \ref{itm:llpm-lam-other}
  of the definition of $\Greater_\llp$ applies mutatis mutandi
  and $t \unrhd s$ holds.
\end{enumerate}
The nonstrict relation is defined so that $t \Greatersim_\llp s$ if
\begin{enumerate}[(1)]
\item the rule
  \ref{itm:nsllpm-var},
  \ref{itm:nsllpm-sym-sub},
  \ref{itm:nsllpm-sym-args},
  \ref{itm:nsllpm-sym-other},
  \ref{itm:nsllpm-db-sub},
  \ref{itm:nsllpm-db-args},
  \ref{itm:nsllpm-lam-sub},
  \ref{itm:nsllpm-lam-types}, or
  \ref{itm:nsllpm-lam-bodies}
  of the definition of $\Greatersim_\llp$ applies mutatis mutandi, or
\item the rule
  \ref{itm:nsllpm-sym-diff} or
  \ref{itm:nsllpm-sym-types}
  of the definition of $\Greater_\llp$ applies mutatis mutandi
  and $\cst{g} \Succ \cst{ws}$ or $t \unrhd s$ holds, or
\item the rule
  \ref{itm:nsllpm-db-diff},
  \ref{itm:nsllpm-db-other}, or
  \ref{itm:nsllpm-lam-other}
  of the definition of $\Greatersim_\llp$ applies mutatis
  mutandi and $t \unrhd s$ holds.
\end{enumerate}
\end{defi}

Like for $\lambda$KBO, the conditions $t \unrhd s$ are sometimes necessary to
guard against $\eta$-expansion on the right-hand side of $\Greatersim_\llp$.
However, they are not necessary in most cases. Consider the
precedence $\cst{h} \Succ \cst{f} \Succ \cst{a}$, with $\cst{h} \Succ \cst{ws}$,
and suppose $\cst{a} : \alpha$, $\cst{h} : \kappa$.
We allow the polymorphic comparison $\cst{h} \Greater_\llp \cst{a}$
even though instantiating $\alpha$ may lead to $\eta$-expansion of~$\cst{a}$:
\[\cst{h} \Greater_\llp \underbrace{\lambda\ldots\lambda}_{k~\text{times}}\> \cst{a}\> (\DB{k-1})\>\ldots\>\DB{0}\]
The key for this to work is that symbols above the watershed---here, $\cst{h}$---are considered larger
than both $\lambda$s and De Bruijn indices.

\begin{thm}
\label{thm:llp-nsllp-monomorphizing-subst-stable}
If $t \Greater_\llp s$, then $t\theta \Greatereq_\llpm t\theta? \Greater_\llpm s\theta$ for any
monomorphizing type substitution~$\theta$.
If $t \Greatersim_\llp s$, then $t\theta \Greatersim_\llpm s\theta$ for any
monomorphizing type substitution~$\theta$.
\end{thm}

\begin{proof}
As an induction hypothesis, the inequality $t\theta? \Greater_\llpm s\theta$ will be
useful to apply rules that have a $\cal{chkargs}$ condition.

First, we show $t\theta \Greatereq_\llpm t\theta?$. The only difference between the
two preterms is the presence of $k$ additional $\lambda$s on the left. If $k = 0$,
Lemma~\ref{lem:nsllpm-reflexive} can be used to establish $t\theta
\Greatereq_\llpm t\theta?$. Otherwise, the property can be established by
applying rule~\ref{itm:llpm-lam-sub} $k$~times.

The proof of the two remaining inequalities is by induction on $|t| + |s|$.

\medskip
\noindent
\textsc{Cases
\ref{itm:llpm-sym-sub}, \ref{itm:llpm-db-sub}, \ref{itm:llpm-lam-sub} of
$\Greater_\llp$:}
These cases all correspond to ``subterm'' rules.
If $t$ is of nonvariable type, we apply the corresponding rule of $\Greater_\llpm$,
relying on the induction hypothesis for the recursive comparison with
$\Greatersim_\llpm$.
Otherwise, suppose $t$ is of type $\alpha$. Let $k$ be the number of
curried arguments expected by values of type $\alpha\theta$. This means that
$t\theta?$ is of the form
\[t'{\uparrow}^k\> \lnf{(\DB{k-1})}\>\ldots\>\lnf{\DB{0}}\]
We apply the rule of $\Greater_\llpm$ corresponding to the rule that was used to
establish $t \Greater_\llp s$ in the first place.

For this to work, a recursive comparison must be possible. We show how it can be
done for the case of rule \ref{itm:llpm-sym-sub} of $\Greater_\llp$, where $t =
\cst{f}(\bar u)\> \bar t$; the other two cases are similar.
For rule \ref{itm:llpm-sym-sub} to have been applicable to establish
$\cst{f}(\bar u)\> \bar t \Greater_\llp s$, we must have $t_i \Greatersim_\llp s$ for some
$i$.
Now, to apply rule~\ref{itm:llpm-sym-sub} to derive
\[\cst{f}(\bar u\theta{\uparrow}^k)\> \bar t\theta{\uparrow}^k\>
  \lnf{(\DB{k-1})}\>\ldots\>\lnf{\DB{0}} \Greater_\llpm s\theta\]
we must show that $t_i\theta{\uparrow}^k
\Greatersim_\llpm s\theta$.
From $t_i \Greatersim_\llp s$, the induction hypothesis, and
by inspection of rules \ref{itm:nsllpm-db-diff}~and~\ref{itm:nsllpm-db-args} of
$\Greatersim_\llp$, we get $t_i\theta{\uparrow}^k
\Greatersim_\llpm t_i\theta \Greatersim_\llpm s\theta$, as desired.

\medskip
\noindent
\textsc{Rules \ref{itm:llpm-sym-diff}, \ref{itm:llpm-sym-types},
\ref{itm:llpm-sym-args}, \ref{itm:llpm-sym-other} of $\Greater_\llp$}:
These cases have a symbol as the head on the left-hand side.
We will focus on the case of rule \ref{itm:llpm-sym-diff}; the other three
cases are similar. For rule \ref{itm:llpm-sym-diff},
$t = \cst{g}(\bar w)\> \bar t$ and $s = \cst{f}(\bar u)\> \bar s$,
with $t \Greater_\llpm s_i$ for every $i$. We also have
either $\cst{g} \Succ \cst{ws}$ or $t \unrhd s$. We focus on the case
where $\cst{g} \Succ \cst{ws}$; the other case is similar to that of
rule \ref{itm:llpm-db-diff}, below.

Let $t : \upsilon$ and $s : \tau$. Let $l$ and $k$ be the number of curried
arguments expected by values of type $\upsilon\theta$ and $\tau\theta$,
respectively. This means that we have
\begin{align*}
  t\theta? & = \cst{g}(\bar w\theta {\uparrow}^l)\> \bar t\theta{\uparrow}^l\>
    \lnf{(\DB{l-1})}\>\ldots\>\lnf{\DB{0}}
& s\theta & = \underbrace{\lambda\ldots\lambda}_{k~\text{times}}\>
    \cst{f}(\bar u\theta {\uparrow}^k)\> \bar s\theta{\uparrow}^k\>
    \lnf{(\DB{k-1})}\>\ldots\>\lnf{\DB{0}}
\end{align*}
To show $t\theta? \Greater_\llp s\theta$, we apply rule~\ref{itm:llpm-sym-other} $k$
times to remove the $\lambda$s on the right. It then suffices to prove
$t\theta? \Greater_\llp s\theta?$. We apply rule~\ref{itm:llpm-sym-diff}. For the rule
to be applicable, due to the $\cal{chkargs}$ condition
we need $t\theta? \Greater_\llpm s_i\theta$ to hold for every $i$. This
follows from $t \Greater_\llp s_i$ and the induction hypothesis. In addition, we need
$t\theta? \Greater_\llpm \lnf{j}$ for $j \in \{0,\ldots,k-1\}$. This follows from
rule~\ref{itm:llpm-sym-other}.

\medskip
\noindent
\textsc{Rules \ref{itm:llpm-db-diff}, \ref{itm:llpm-db-args} of $\Greater_\llp$:}
These cases compare applied De Bruijn indices $t = \DB{n}\> \bar t$ and
$s = \DB{m}\> \bar s$, where $n \ge m$. We also know that $t \Greater_\llp s_i$ for
every $i$.
We perform a case analysis on $t \unrhd s$.

\medskip
\noindent
\textsc{Subcase 1}, where $t$ and $s$ are of nonvariable types:\enskip
Rule \ref{itm:llpm-db-diff} or \ref{itm:llpm-db-args} of
$\Greater_\llpm$ applies. For the rule to be applicable, due to the $\cal{chkargs}$
condition we need $t\theta = t\theta? \Greater_\llpm s_i\theta$ to hold for every $i$.
This follows from $t \Greater_\llp s_i$ and the induction hypothesis.

\medskip
\noindent
\textsc{Subcase 2}, where $t$ has some variable type $\alpha$ but not
$s$:\enskip
If $\alpha\theta$ is nonfunctional, the reasoning is as for subcase~1.
Otherwise, $\alpha\theta$ is a function type, and applying $\theta$ to $t$
results in $k > 0$ $\eta$-expansions. This means we have
\begin{align*}
  t\theta? & = (\DB{n + k})\> (\bar t\theta){\uparrow}^k\>
    \lnf{(\DB{k-1})}\>\ldots\>\lnf{\DB{0}}
& s\theta & = \DB{m}\> \bar s\theta
\end{align*}
Since $n + k > m$, we apply rule~\ref{itm:llpm-db-diff} to establish $t\theta?
\Greater_\llpm s\theta$. For the rule to be applicable,
due to the $\cal{chkargs}$ condition
$t\theta? \Greater_\llpm s_i\theta$ must hold for every $i$. This follows from
$t \Greater_\llp s_i$ and the induction hypothesis.

\medskip
\noindent
\textsc{Subcase 3}, where $t$ and $s$ has some variable type $\alpha$:\enskip
Let $k$ be the number of curried arguments expected by values of type
$\alpha\theta$. This means that we have
\begin{align*}
  t\theta? & = (\DB{n + k})\> (\bar t\theta){\uparrow}^k\>
    \lnf{(\DB{k-1})}\>\ldots\>\lnf{\DB{0}}
& s\theta & = \underbrace{\lambda\ldots\lambda}_{k~\text{times}}\>
  (\DB{m + k})\> (s\theta){\uparrow}^k\> \lnf{(\DB{k-1})}\>\ldots\>\lnf{\DB{0}}
\end{align*}
and must show $t\theta? \Greater_\llpm s\theta$. First,
we apply rule~\ref{itm:llpm-db-other} $k$~times to remove the
$\lambda$s on the right. For the rule to be applicable,
due to the $\cal{chkargs}$ condition $t\theta? \Greater_\llpm s\theta?$ must hold.
To prove it, we apply rule \ref{itm:llpm-db-diff} or \ref{itm:llpm-db-args},
depending on whether $n > m$ or $n = m$. For the tuple comparison in rule
\ref{itm:llpm-db-args}, it is easy to see that the additional De Bruijn
arguments are harmless. For either rule to be applicable,
due to the $\cal{chkargs}$ condition we also need
$t\theta? \Greater_\llpm s_i\theta$ for every $i$ and $t\theta? \Greater_\llpm \lnf{j}$
for every $j \in \{0,\ldots,k-1\}$.
The first inequality follows from $t \Greater_\llp s_i$ and the induction hypothesis.
The second inequality follows from rule~\ref{itm:llpm-db-sub}, since
one of the arguments in $t\theta?$ is $\lnf{j}$.

\medskip
\noindent
\textsc{Rule \ref{itm:llpm-db-other} of $\Greater_\llp$:}
This case compares an applied De Bruijn index $t = \DB{n}\>\bar t$ and either a
$\lambda$-abstraction $\lambda\> s'$ or an applied symbol
$\cst{f}(\bar u)\> \bar s$ below the watershed.
In the $\lambda$ subcase, we have $t \Greater_\llpm s'$. We apply
rule~\ref{itm:llpm-db-other} to derive $t\theta? \Greater_\llpm s\theta$. This
requires us to prove $t\theta? \Greater_\llpm s'\theta$, which follows from $t
\Greater_\llpm s'$ and the induction hypothesis.
In the other subcase, the proof is similar to as in cases
\ref{itm:llpm-db-diff}, \ref{itm:llpm-db-args} of $\Greater_\llpm$ above.

\medskip
\noindent
\textsc{Rule \ref{itm:llpm-lam-types} of $\Greater_\llp$:}
This case compares two $\lambda$-abstractions $t =
\lambda\langle\upsilon\rangle\> t'$ and $s = \lambda\langle\tau\rangle\> s'$.
We have $t \Greater_\llp s'$. To derive the desired inequality
$t\theta? = t\theta \Greater_\llpm s\theta$,
we apply rule~\ref{itm:llpm-lam-types}, which requires us to prove
$\upsilon\theta \Greater_\ty \tau\theta$ and $t\theta \Greater_\llpm s'\theta$.
The first inequality follows from $\upsilon \Greater_\ty \tau$ by stability under
substitution of the standard LPO.
The second inequality follows from $t \Greater_\llp s'$ and the induction
hypothesis.

\medskip
\noindent
\textsc{Rules \ref{itm:llpm-lam-bodies} of $\Greater_\llp$:}
These cases compare two $\lambda$-abstractions $t =
\lambda\langle\upsilon\rangle\> t'$ and $s = \lambda\langle\upsilon\rangle\>
s'$. We have $t' \Greater_\llp s'$. By the induction hypothesis,
$t'\theta \Greater_\llp s'\theta$. By rule \ref{itm:llpm-lam-bodies}, we get
$t\theta = \lambda\langle\upsilon\theta\rangle\> t'\theta
\Greater_\llpm \lambda\langle\upsilon\theta\rangle\> s'\theta = s\theta$,
as desired.

\medskip
\noindent
\textsc{Rule \ref{itm:llpm-lam-other} of $\Greater_\llp$:}
This case compares a $\lambda$-abstraction
$t = \lambda\langle\upsilon\rangle\> t'$ and an applied symbol
$\cst{f}(\bar u)\> \bar s$ below the watershed.
The proof is similar to as in cases
\ref{itm:llpm-db-diff}, \ref{itm:llpm-db-args} of $\Greater_\llpm$ above.

\medskip
\noindent
\textsc{Rule \ref{itm:nsllpm-var} of $\Greatersim_\llp$:}
This case compares two preterms $y\> \bar t$ and $y\> \bar s$ headed by the same
variable and of the same type $\tau$.
Let $k$ be the number of curried arguments expected by values of type
$\tau\theta$. This means that we have
\begin{align*}
  t\theta & = \underbrace{\vphantom{()}\lambda\ldots\lambda}_{k~\text{times}}\>
    \underbrace{\vphantom{()}\smash{(y\> \bar t)\theta{\uparrow}^k\> \lnf{(k-1)} \ldots \lnf{0}}\,}_{(y\> \bar t)\theta?}
& s\theta & = \underbrace{\vphantom{()}\lambda\ldots\lambda}_{k~\text{times}}\>
    \underbrace{\vphantom{()}\smash{(y\> \bar s)\theta{\uparrow}^k\> \lnf{(k-1)} \ldots \lnf{0}}\,}_{(y\> \bar s)\theta?}
\end{align*}
To show $t\theta \Greatersim_\llpm s\theta$, we apply rule~\ref{itm:nsllpm-lam-bodies}
$k$ times. The rule is applicable if $(y\> \bar t)\theta? \Greatersim_\llpm
(y\> \bar s)\theta?$. It is easy to see that this last inequality can be established
using rule~\ref{itm:nsllpm-var} given that $y\> \bar t \Greatersim_\llpm y\> \bar s$, using
the induction hypothesis to compare pairs $t_i, s_i$ and using
Lemma~\ref{lem:nsllpm-reflexive} for the pairs $\lnf{j}, \lnf{j}$ of (identical)
De Bruijn indices introduced by $\eta$-expansion.

\medskip
\noindent
\textsc{Cases
\ref{itm:nsllpm-sym-sub}, \ref{itm:nsllpm-sym-diff}, \ref{itm:nsllpm-sym-types},
\ref{itm:nsllpm-sym-other}, \ref{itm:nsllpm-db-sub}, \ref{itm:nsllpm-db-diff},
\ref{itm:nsllpm-db-other}, \ref{itm:nsllpm-lam-sub}, \ref{itm:nsllpm-lam-types},
\ref{itm:nsllpm-lam-other}
of $\Greatersim_\llp$}:
These cases are similar to cases
\ref{itm:llpm-sym-sub}, \ref{itm:llpm-sym-diff}, \ref{itm:llpm-sym-types},
\ref{itm:llpm-sym-other}, \ref{itm:llpm-db-sub}, \ref{itm:llpm-db-diff},
\ref{itm:llpm-db-other}, \ref{itm:llpm-lam-sub}, \ref{itm:llpm-lam-types},
\ref{itm:llpm-lam-other}
of $\Greater_\llp$.
These cases correspond to strict inequalities. We first establish $t\theta
\Greatersim_\llpm t\theta?$ by applying rule \ref{itm:nsllpm-lam-sub} repeatedly.
Then we show $t\theta? \Greatersim_\llpm s\theta$ in the same way as in the
corresponding case of $\Greater_\llp$.

\medskip
\noindent
\textsc{Rules \ref{itm:nsllpm-sym-args}, \ref{itm:nsllpm-db-args} of
$\Greatersim_\llp$}:
These cases may correspond to nonstrict comparisons---for example, if the
argument tuples are equal.
We apply rule~\ref{itm:nsllpm-lam-bodies} repeatedly to
eliminate any $\lambda$s on both sides.
If there are any $\lambda$s remaining on the left, proceed as in the previous
case (\ref{itm:nsllpm-sym-sub}, \ref{itm:nsllpm-sym-diff}, etc.). Otherwise,
the rest of the proof is similar to case
\ref{itm:llpm-sym-args} or \ref{itm:llpm-db-args} of
$\Greater_\llp$.

\medskip
\noindent
\textsc{Rules \ref{itm:nsllpm-lam-bodies} of $\Greatersim_\llp$:}
Analogous to case \ref{itm:llpm-lam-bodies} of $\Greater_\llp$.
\end{proof}

\begin{thm}
\label{thm:llp-coincide-monomorphic}
The relation $\Greater_\llp$ coincides with $\Greater_\llpm$ on monomorphic preterms.
The relation $\Greatersim_\llp$ coincides with $\Greatersim_\llpm$ on monomorphic preterms.
\end{thm}

\begin{proof}
One direction of the equivalences follows by
Theorem~\ref{thm:llp-nsllp-monomorphizing-subst-stable}. It remains to show that
$t \Greater_\llpm s$ implies $t \Greater_\llp s$ and that $t \Greatersim_\llpm
s$ implies $t \Greatersim_\llp s$. The proof is by induction on the definition
of $\Greater_\llpm$ and $\Greatersim_\llpm$. It is easy to see that every case
in the definition of $\Greater_\llpm$ corresponds a case in the definition of
$\Greater_\llp$ and every case in the definition of $\Greatersim_\llpm$ corresponds
a case in the definition of $\Greatersim_\llp$.
\end{proof}

\begin{lem}
  \label{lem:nllp-implies-llp}
  If $t \Greater_\llp s$, then $t \Greatersim_\llp s$.
  \end{lem}
  \begin{proof}
  Analogous to Lemma~\ref{lem:llpm-implies-nsllpm}.
    \end{proof}
  
  \begin{thm}
  \label{thm:llp-transitive}
  If $t \Greatersim_\llp u$ and $u \Greatersim_\llp s$, then $t \Greatersim_\llp s$.
  If moreover $t \Greater_\llp u$ or $u \Greater_\llp s$, then $t \Greater_\llp s$.
  \end{thm}
  \begin{proof}
  Analogous to Theorem~\ref{thm:llpm-transitive},
  using the fact that $\unrhd$ is transitive.
    \end{proof}

  \begin{thm}\label{thm:llp-variable-guarantee}
  Let $t \Greater_\llp s$.
  Let $\theta$ be a substitution such that all variables in
  $t\theta$ and $s\theta$ are nonfunctional term variables.
  Let $s\theta$ contain a nonfunctional variable $x$
  outside of parameters.
  Then $t\theta$ must also contain $x$ outside of parameters.
  \end{thm}
  \begin{proof}
  By Theorems \ref{thm:llpm-variable-guarantee} and \ref{thm:llp-nsllp-monomorphizing-subst-stable}.
    \end{proof}

\section{Examples}
\label{sec:examples}

Let us see how $\lambda$KBO and $\lambda$LPO work on some realistic examples.
All the examples below are beyond the reach of the derived higher-order KBO and
LPO presented by Bentkamp et al.\ \cite{bentkamp-et-al-2021-hosup} and
implemented in Zipperposition.

\begin{exa}
Consider the following clause:
$\cst{p}\>(\lambda\>\cst{f}\>(y\>\DB{0}))
  \mathrel\lor \cst{p}\>(\lambda\>y\>\DB{0})$.
We would like to orient the two literals. The orientation
$\cst{p}\>(\lambda\>\cst{f}\>(y\>\DB{0})) \Greater \cst{p}\>(\lambda\>y\>\DB{0})$
appears more promising.
With $\lambda$KBO, assuming a weight of $1$ for $\cst{f}$, $\cst{p}$,
$\lambda$, and $\DB{0}$ and argument coefficients of $1$, via a mechanical
application of Definition~\ref{def:weight-lkb} we get the polynomial inequality
$1 + 1 + 1 + 1 + \mathbf{w}_y + \mathbf{k}_{y,1}(1 - 1)
  \vthinspace>\vthinspace
  1 + 1 + 1 + \mathbf{w}_y + \mathbf{k}_{y,1}(1 - 1)$.
In other words, $1 > 0$.
With $\lambda$LPO, the desired orientation is easy to derive since
$y\>\DB{0}$ is a subterm of $\cst{f}\>(y\>\DB{0})$.
\end{exa}

\begin{exa}
The following equation defines the transitivity of a relation $r$,
encoded as a binary predicate:
$\cst{trans}\> (\lambda\>\lambda\>r\>\DB{1}\>\DB{0})
  \vthinspace\approx\vthinspace
  \heavy{\forall}\> (\lambda\>
    \heavy{\forall}\> (\lambda\>
      \heavy{\forall}\> (\lambda\>
        r\> \DB{2}\> \DB{1} \mathbin{\heavy{\land}} r\> \DB{1}\> \DB{0}
        \mathbin{\heavy{\rightarrow}}
        r\> \DB{2}\> \DB{0)}))$.
We start with $\lambda$KBO.
Assume a weight of 1 for $\cst{trans}$, $\heavy{\forall}$, $\heavy{\land}$,
$\heavy{\to}$, $\lambda$, $\DB{0}$, \dots{} and argument coefficients of $1$.
After simplification, the polynomial inequality for a right-to-left orientation
becomes
$\mathbf{w}_r + 4 \vthinspace<\vthinspace 3 \mathbf{w}_r + 14$,
which clearly holds. Is there a way to orient the equation from left to right
instead? There is if we make $\cst{trans}$ heavier and set a
higher weight coefficient on its argument. Take $\cal w(\cst{trans}) = 5$ and
$\cal k(\cst{trans}, 1) = 3$. Then we get
$3 \mathbf{w}_r + 15 \vthinspace>\vthinspace 3 \mathbf{w}_r + 14$.
In contrast, $\lambda$LPO cannot orient the equation from left to right; among
the proof obligations that emerge are $r\> \DB{1}\> \DB{0} \Greatersim_\llp^? r\>
\DB{2}\> \DB{1}$ and $r\> \DB{1}\> \DB{0} \Greatersim_\llp^? r\> \DB{2}\> \DB{0}$,
and these cannot be discharged.
\end{exa}

\begin{exa}
In the clause
$y\>(\lambda\>\cst{a}\>\DB{0}) \mathrel\lor
\lnot\> y\> (\lambda\> \cst{f}\> (\cst{sk}(y)\> \DB{0})) \mathrel\lor
\lnot\> y\> (\lambda\> \DB{0})$,
we would like to make $y\>(\lambda\>\cst{a}\>\DB{0})$ the maximal literal.
The apparent difficulty is the presence of the variable $y$ deep inside the
second literal. Fortunately, since it occurs in a parameter, it has no impact on
the $\lambda$KBO weight. We are then free to make the symbol~$\cst{a}$ as heavy
as we want to ensure that $\lambda\>\cst{a}\>\DB{0}$ is heavier than $\lambda\>
\cst{f}\> (\cst{sk}(y)\> \DB{0})$ and $\lambda\> \DB{0}$, both of which have
constant weights. A similar approach can be taken with $\lambda$LPO, using the
precedence instead of weights.
\end{exa}

\begin{exa}
In functional programming, the map function on lists is defined recursively by
\begin{align*}
\cst{map}\> (\lambda\>f\>\DB{0})\> \cst{nil}
  & \approx \cst{nil}
&
\cst{map}\> (\lambda\>f\>\DB{0})\> (\cst{cons}\> x\> \mathit{xs})
  & \approx \cst{cons}\> (f\> x)\> (\cst{map}\> (\lambda\>f\>\DB{0})\> \mathit{xs})
\end{align*}
The first equation is easy to orient from left to right. Not so for the second
equation. With $\lambda$KBO, to compensate for the two occurrences of $f$ on the
right-hand side, we would need to set a coefficient of at least $2$ on
$\cst{map}$'s first argument; this would make the left-hand side heavier but
would also make the right-hand side even heavier. In general, KBO is rather
ineffective at orienting recursive equations from left to right.
With $\lambda$LPO, the issue is the undischargeable proof obligation $f\> \DB{0}
\Greatersim_\llp^? f\> x$.

With both orders, a right-to-left orientation is also problematic, because $f$
might ignore its argument, resulting in an $x$ on the left-hand side with no
matching $x$ on the right-hand side. On the positive side, superposition provers
rarely need to orient recursive equations in their full generality. Instead, the
calculus considers \relax{instances} of the equations where higher-order
variables are replaced by concrete functions. These equation instances are often
orientable from right to left.
\end{exa}

\section{Naive Algorithms}
\label{sec:naive-algorithms}

The definitions given in Sect.~\ref{sec:the-polymorphic-level} provide a sound
theoretical basis for an implementation, but they should not be followed
naively.
Our algorithms below perform the comparisons $t \Greater s$, $t \Greatersim s$,
$t = s$, $t \Lesssim s$, and $t \Less s$ simultaneously, reusing
subcomputations. Typically, given terms $s, t$, a superposition prover might
need to check both $t \Greatersim s$ and $t \Lesssim s$.

As our programming language, we use a functional programming notation inspired
by Standard ML, OCaml, and Haskell.
First, we need a type to represent the result of a comparison:

\begin{quote}
\textbf{datatype} \,$\mathsf{cmp}$ \,$=$\,
  $\mathsf{G}$ $\mid$
  $\mathsf{GE}$ $\mid$
  $\mathsf{E}$ $\mid$
  $\mathsf{LE}$ $\mid$
  $\mathsf{L}$ $\mid$
  $\mathsf{U}$
\end{quote}
The six values represent $\Greater$, $\Greatersim$, $=$, $\Lesssim$, $\Less$,
and ``unknown'' or ``incomparable,'' respectively.

\subsection{\texorpdfstring{$\bm{\lambda}$KBO}{Lambda KBO}}
\label{ssec:lambda-kbo-naive-algorithm}

Our first algorithm will perform $\lambda$KBO comparisons in both directions
simultaneously.
The following auxiliary functions are used to combine an imprecise comparison
result $\mathsf{GE}$ or $\mathsf{LE}$ with another result, using $\mathsf{U}$
on mismatch:

\begin{quote}
\textbf{function} \,$\mathsf{mergeWithGE}\; \mathit{cmp}$ \,$:=$ \\
\q \textbf{match} $\mathit{cmp}$ \textbf{with} \\
\q\q $\mathsf{L}$ $\mid$ $\mathsf{LE}$ $\Rightarrow$ $\mathsf{U}$ \\
\q\Q $\mathsf{E}$ $\Rightarrow$ $\mathsf{GE}$ \\
\q\Q $\_$ $\Rightarrow$ $\mathit{cmp}$ \\
\q \textbf{end}
\end{quote}

\begin{quote}
\textbf{function} \,$\mathsf{mergeWithLE}\; \mathit{cmp}$ \,$:=$ \\
\q \textbf{match} $\mathit{cmp}$ \textbf{with} \\
\q\q $\mathsf{G}$ $\mid$ $\mathsf{GE}$ $\Rightarrow$ $\mathsf{U}$ \\
\q\Q $\mathsf{E}$ $\Rightarrow$ $\mathsf{LE}$ \\
\q\Q $\_$ $\Rightarrow$ $\mathit{cmp}$ \\
\q \textbf{end}
\end{quote}

The lexicographic extension of a comparison operator working with our
comparison type is defined as follows:

\begin{quote}
\textbf{function} \,$\mathsf{lexExt}\; \mathit{op}\; \bar b\; \bar a$ \,$:=$ \\
\q \textbf{match} $\bar b$, $\bar a$ \textbf{with} \\
\q\q $[]$, $[]$ $\Rightarrow$ $\mathsf{E}$ \\
\q\Q $b :: \bar b'$, $a :: \bar a'$ $\Rightarrow$ \\
\q\q \textbf{match} $\mathit{op}\; b\; a$ \textbf{with} \\
\q\q\q $\mathsf{G}$ $\Rightarrow$ $\mathsf{G}$ \\
\q\q\Q $\mathsf{GE}$ $\Rightarrow$ $\mathsf{mergeWithGE}\; (\mathsf{lexExt}\; \mathit{op}\; \bar b'\; \bar a')$ \\
\q\q\Q $\mathsf{E}$ $\Rightarrow$ $\mathsf{lexExt}\; \mathit{op}\; \bar b'\; \bar a'$ \\
\q\q\Q $\mathsf{LE}$ $\Rightarrow$ $\mathsf{mergeWithLE}\; (\mathsf{lexExt}\; \mathit{op}\; \bar b'\; \bar a')$ \\
\q\q\Q $\mathsf{L}$ $\Rightarrow$ $\mathsf{L}$ \\
\q\q\Q $\mathsf{U}$ $\Rightarrow$ $\mathsf{U}$ \\
\q\q \textbf{end} \\
\q \textbf{end}
\end{quote}
We need to support only the case in which both lists have the same length.

Next to the lexicographic extension, we also define a form of componentwise
extension of a comparison operator:

\begin{quote}
\textbf{function} \,$\mathsf{smooth}\; \mathit{cmp}$ \,$:=$ \\
\q \textbf{match} $\mathit{cmp}$ \textbf{with} \\
\q\q $\mathsf{G}$ $\Rightarrow$ $\mathsf{GE}$ \\
\q\Q $\mathsf{L}$ $\Rightarrow$ $\mathsf{LE}$ \\
\q\Q $\_$ $\Rightarrow$ $\mathit{cmp}$ \\
\q \textbf{end}

\medskip

\textbf{function} \,$\mathsf{cwExt}\; \mathit{op}$ \,$:=$ \\
\q $\mathsf{lexExt}\; (\textbf{fun}\; b\; a \Rightarrow \mathsf{smooth}\; (\mathit{op}\; b\; a))$
\end{quote}
We need to support only the case in which both lists have the same length.

Next, we need a function that checks whether $\unrhd$ or its inverse $\unlhd$ holds
and that adjusts the comparison result accordingly:

\begin{quote}
\textbf{function} \,$\mathsf{considerPoly}\; t\; s\; \mathit{cmp}$ \,$:=$ \\
\q \textbf{match} $\mathit{cmp}$ \textbf{with} \\
\q\q $\mathsf{G}$ $\mid$ $\mathsf{GE}$ $\Rightarrow$
  \textbf{if} $t \unrhd s$ \textbf{then} $\mathit{cmp}$ \textbf{else} $\mathsf{U}$ \\
\q\q $\mathsf{L}$ $\mid$ $\mathsf{LE}$ $\Rightarrow$
  \textbf{if} $t \unlhd s$ \textbf{then} $\mathit{cmp}$ \textbf{else} $\mathsf{U}$ \\
\q\Q $\_$ $\Rightarrow$ $\mathit{cmp}$ \\
\q \textbf{end}
\end{quote}

The function for checking inequalities is very simple:

\begin{quote}
\textbf{function} \,$\mathsf{surelyNonneg}\; w$ \,$:=$ \\
\q all coefficients in the standard form of $w$ are $\ge 0$
\end{quote}
It returns $\mathsf{true}$ if all the polynomials in the list are certainly
nonnegative for any values of the indeterminates and $\mathsf{false}$ if this is
not known to be the case, either because there exists a counterexample or
because the approach is too imprecise to tell.

Polynomials in standard form have at most one \emph{constant monomial}:\ a
monomial consisting of only a coefficient with no indeterminates. If absent, it
is taken to be 0. The weight comparison can be refined by considering the sign
of the constant monomial in the difference $\cal W(t) - \cal W(s)$. If the sign
is positive, $\cal W(t) \ge \cal W(s)$ actually means $\cal W(t) > \cal W(s)$.
If the sign is negative, $\cal W(t) \le \cal W(s)$ actually means $W(t) < \cal
W(s)$.

\begin{quote}
\textbf{function} \,$\mathsf{analyzeWeightDiff}\; w$ \,$:=$ \\
\q \textbf{match} $\mathsf{surelyNonneg}\; w$, $\mathsf{surelyNonneg}\; (-w)$ \textbf{with} \\
\q\q $\mathsf{false}$, $\mathsf{false}$ $\Rightarrow$ $\mathsf{U}$ \\
\q\Q $\mathsf{true}$, $\mathsf{false}$ $\Rightarrow$ \textbf{if\/} the constant monomial of $w$ is $> 0$ \textbf{then} $\mathsf{G}$ \textbf{else} $\mathsf{GE}$ \\
\q\Q $\mathsf{false}$, $\mathsf{true}$ $\Rightarrow$ \textbf{if\/} the constant monomial of $w$ is $< 0$ \textbf{then} $\mathsf{L}$ \textbf{else} $\mathsf{LE}$ \\
\q\Q $\mathsf{true}$, $\mathsf{true}$ $\Rightarrow$ $\mathsf{E}$ \\
\q \textbf{end}
\end{quote}

For preterms with possibly equal weights, a lexicographic comparison
implemented by the $\mathsf{compareShapes}$ function below breaks the
tie. We assume the existence of a function $\mathsf{compareSyms}\; \cst{g}\;
\cst{f}$ based on $\Succ$ that returns $\mathsf{G}$, $\mathsf{E}$, or
$\mathsf{L}$ and of a function $\mathsf{compareTypes}\; \upsilon\; \tau$ based
on $\Greater_\ty$ that returns $\mathsf{G}$, $\mathsf{E}$, $\mathsf{L}$, or
$\mathsf{U}$. The $\mathsf{compareShapes}$ function is mutually recursive with
the main comparison function, $\mathsf{compareTerms}$.

\begin{quote}
\textbf{function} \,$\mathsf{compareShapes}\; t\; s$ \,$:=$ \\
\q \textbf{match} $t$, $s$ \textbf{with} \\
\q\q $y\> \bar t$, $y\> \bar s$ $\Rightarrow$ \textbf{if} $\bar t$ are steady \textbf{then} $\mathsf{cwExt}\; \mathsf{compareTerms}\; \bar t\; \bar s$ \textbf{else} $\mathsf{U}$ \\
\q\Q $y\> \_$, $\_$ $\mid$ $\_$, $x\> \_$ $\Rightarrow$ $\mathsf{U}$ \\
\q\Q $\lambda\langle\upsilon\rangle\> t'$, $\lambda\langle\tau\rangle\> s'$ $\Rightarrow$ \\
\q\q \textbf{match} $\mathsf{compareTypes}\; \upsilon\; \tau$ \textbf{with} \\
\q\q\q $\mathsf{E}$ $\Rightarrow$ $\mathsf{compareShapes}\; t'\; s'$ \\
\q\q\Q $\mathit{cmp}$ $\Rightarrow$ $\mathit{cmp}$ \\
\q\q \textbf{end} \\
\q\Q $\lambda\> \_$, $\_$ $\Rightarrow$ $\mathsf{considerPoly}\; t\; s\; \mathsf{G}$ \\
\q\Q $\DB{n}\> \_$, $\lambda\> \_$ $\Rightarrow$ $\mathsf{considerPoly}\; t\; s\; \mathsf{L}$ \\
\q\Q $\DB{n}\> \bar t$, $\DB{m}\>\bar s$ $\Rightarrow$ \\
\q\q \textbf{if\/} $n > m$ \textbf{then} $\mathsf{considerPoly}\; t\; s\; \mathsf{G}$ \\
\q\q \textbf{else if\/} $n < m$ \textbf{then} $\mathsf{considerPoly}\; t\; s\; \mathsf{L}$ \\
\q\q \textbf{else} $\mathsf{lexExt}\; \mathsf{compareTerms}\; \bar t\; \bar s$ \\
\q\Q $\DB{n}\> \_$, $\cst{f}\langle\_\rangle(\_)\>\_$ $\Rightarrow$ $\mathsf{considerPoly}\; t\; s\; \mathsf{G}$ \\
\q\Q $\cst{g}\langle\bar\upsilon\rangle(\bar w)\>\bar t$, $\cst{f}\langle\bar\tau\rangle(\bar u)\>\bar s$ $\Rightarrow$ \\
\q\q \textbf{match} $\mathsf{compareSyms}\; \cst{g}\; \cst{f}$ \textbf{with} \\
\q\q\q $\mathsf{E}$ $\Rightarrow$ \\
\q\q\q \textbf{match} $\mathsf{lexExt}\; \mathsf{compareTypes}\; \bar\upsilon\; \bar\tau$ \textbf{with} \\
\q\q\q\q $\mathsf{E}$ $\Rightarrow$ $\mathsf{lexExt}\; \mathsf{compareTerms}\; (\bar w \cdot \bar t)\; (\bar u \cdot \bar s)$ \\
\q\q\q\Q $\mathit{cmp}$ $\Rightarrow$ $\mathsf{considerPoly}\; t\; s\; \mathit{cmp}$ \\
\q\q\q \textbf{end} \\
\q\q\Q $\mathit{cmp}$ $\Rightarrow$ $\mathsf{considerPoly}\; t\; s\; \mathit{cmp}$ \\
\q\q \textbf{end} \\
\q\Q $\cst{g}\langle\_\rangle(\_)\>\_$, $\_$ $\Rightarrow$ $\mathsf{considerPoly}\; t\; s\; \mathsf{L}$ \\
\q \textbf{end}
\end{quote}
In the above, the operator $\cdot$ denotes list concatenation.

The main function implementing $\lambda$KBO invokes $\mathsf{analyzeWeightDiff}$,
falling back on $\mathsf{compareShapes}$ to break ties:

\begin{quote}
\textbf{function} \,$\mathsf{compareTerms}\; t\; s$ \,$:=$ \\
\q \textbf{match} $\mathsf{analyzeWeightDiff}\; (\cal W(t) - \cal W(s))$ \textbf{with} \\
\q\q $\mathsf{G}$ $\Rightarrow$ $\mathsf{G}$ \\
\q\Q $\mathsf{GE}$ $\Rightarrow$ $\mathsf{mergeWithGE}\; (\mathsf{compareShapes}\; t\; s)$ \\
\q\Q $\mathsf{E}$ $\Rightarrow$ $\mathsf{compareShapes}\; t\; s$ \\
\q\Q $\mathsf{LE}$ $\Rightarrow$ $\mathsf{mergeWithLE}\; (\mathsf{compareShapes}\; t\; s)$ \\
\q\Q $\mathsf{L}$ $\Rightarrow$ $\mathsf{L}$ \\
\q\Q $\mathsf{U}$ $\Rightarrow$ $\mathsf{U}$ \\
\q \textbf{end}
\end{quote}

\subsection{\texorpdfstring{$\bm{\lambda}$LPO}{Lambda LPO}}
\label{ssec:lambda-lpo-naive-algorithm}

The following algorithm performs $\lambda$LPO comparisons in both directions
simultaneously. The main comparison function, $\mathsf{compareTerms}$, is
accompanied by four mutually recursive auxiliary functions.
\begin{quote}
\textbf{function} \,$\mathsf{considerPolyBelowWS}\; \cst{g}\; t\; s\; \mathit{cmp}$ \,$:=$ \\
\q \textbf{if} $\cst{g} \Succ \cst{ws}$ \textbf{then} $\mathit{cmp}$ \textbf{else}
  $\mathsf{considerPoly}\; t\; s\; \mathit{cmp}$

\medskip

\textbf{function} $\mathsf{checkSubs}\; \bar t\; s$ $:=$ \\
\q $\exists i.\; \mathsf{compareTerms}\; t_i\; s \in \{\mathsf{G}, \mathsf{GE}, \mathsf{E}\}$

\medskip

\textbf{function} $\mathsf{checkArgs}\; t\; \bar s$ $:=$ \\
\q $\forall i.\; \mathsf{compareTerms}\; t\; s_i = \mathsf{G}$

\medskip

\textbf{function} $\mathsf{compareArgs}\; t\; \bar v\; \bar t\; s\; \bar u\; \bar s$ $:=$ \\
\q \textbf{match} $\mathsf{lexExt}\; \mathsf{compareTerms}\; (\bar v \cdot \bar t)\; (\bar u \cdot \bar s)$ \textbf{with} \\
\q\q $\mathsf{G}$ $\Rightarrow$ \textbf{if} $\mathsf{checkArgs}\; t\; \bar s$ \textbf{then} $\mathsf{G}$ \textbf{else} $\mathsf{U}$ \\
\q\Q $\mathsf{GE}$ $\Rightarrow$ \textbf{if} $\mathsf{checkArgs}\; t\; \bar s$ \textbf{then} $\mathsf{GE}$ \textbf{else} $\mathsf{U}$ \\
\q\Q $\mathsf{E}$ $\Rightarrow$ $\mathsf{E}$ \\
\q\Q $\mathsf{LE}$ $\Rightarrow$ \textbf{if} $\mathsf{checkArgs}\; s\; \bar t$ \textbf{then} $\mathsf{LE}$ \textbf{else} $\mathsf{U}$ \\
\q\Q $\mathsf{L}$ $\Rightarrow$ \textbf{if} $\mathsf{checkArgs}\; s\; \bar t$ \textbf{then} $\mathsf{L}$ \textbf{else} $\mathsf{U}$ \\
\q\Q $\mathsf{U}$ $\Rightarrow$ $\mathsf{U}$ \\
\q \textbf{end}

\medskip

\textbf{function} $\mathsf{compareTerms}\; t\; s$ $:=$ \\
\q \textbf{match} $t$ \textbf{with} \\
\q\q $y\> \bar t$ $\Rightarrow$ \\
\q\q \textbf{match} $s$ \textbf{with} \\
\q\q\q $x\> \bar s$ $\Rightarrow$ \textbf{if} $y = x \mathrel\land \bar t$ are steady \textbf{then} $\mathsf{cwExt}\; \mathsf{compareTerms}\; \bar t\; \bar s$ \textbf{else} $\mathsf{U}$ \\
\q\q\Q $\cst{f}\langle\_\rangle(\_)\> \bar s$ $\mid$ $\DB{m}\> \bar s$ $\Rightarrow$ \textbf{if} $\mathsf{checkSubs}\; \bar s\; t$ \textbf{then} $\mathsf{L}$ \textbf{else} $\mathsf{U}$ \\
\q\q\Q $\lambda\> s'$ $\Rightarrow$ \textbf{if} $\mathsf{checkSubs}\; [s']\; t$ \textbf{then} $\mathsf{L}$ \textbf{else} $\mathsf{U}$ \\
\q\q \textbf{end} \\
[\jot]
\q\Q $\cst{g}\langle\bar\upsilon\rangle(\bar w)\> \bar t$ $\Rightarrow$ \\
\q\q \textbf{if} $\mathsf{checkSubs}\; \bar t\; s$ \textbf{then} \\
\q\q\q $\mathsf{G}$ \\
\q\q \textbf{else match} $s$ \textbf{with} \\
\q\q\q $x\> \_$ $\Rightarrow$ $\mathsf{U}$ \\
\q\q\Q $\cst{f}\langle\bar\tau\rangle(\bar u)\> \bar s$ $\Rightarrow$ \\
\q\q\q \textbf{if} $\mathsf{checkSubs}\; \bar s\; t$ \textbf{then} \\
\q\q\q\q $\mathsf{L}$ \\
\q\q\q \textbf{else match} $\mathsf{compareSyms}\; \cst{g}\; \cst{f}$ \textbf{with} \\
\q\q\q\q $\mathsf{G}$ $\Rightarrow$ \textbf{if} $\mathsf{checkArgs}\; t\; \bar s$ \textbf{then} $\mathsf{considerPolyBelowWS}\; \cst{g}\; \mathsf{G}$ \textbf{else} $\mathsf{U}$ \\
\q\q\q\Q $\mathsf{E}$ $\Rightarrow$ \\
\q\q\q\q \textbf{match} $\mathsf{lexExt}\; \mathsf{compareTypes}\; \bar\upsilon\; \bar\tau$ \textbf{with} \\
\q\q\q\q\q $\mathsf{G}$ $\Rightarrow$ \textbf{if} $\mathsf{checkArgs}\; t\; \bar s$ \textbf{then} $\mathsf{considerPolyBelowWS}\; \cst{g}\; \mathsf{G}$ \textbf{else} $\mathsf{U}$ \\
\q\q\q\q\Q $\mathsf{E}$ $\Rightarrow$ $\mathsf{compareArgs}\; t\; \bar w\; \bar t\; s\; \bar u\; \bar s$ \\
\q\q\q\q\Q $\mathsf{L}$ $\Rightarrow$ \textbf{if} $\mathsf{checkArgs}\; s\; \bar t$ \textbf{then} $\mathsf{considerPolyBelowWS}\; \cst{g}\; \mathsf{L}$ \textbf{else} $\mathsf{U}$ \\
\q\q\q\q\Q $\mathsf{U}$ $\Rightarrow$ $\mathsf{U}$ \\
\q\q\q\q \textbf{end} \\
\q\q\q\Q $\mathsf{L}$ $\Rightarrow$ \textbf{if} $\mathsf{checkArgs}\; s\; \bar t$ \textbf{then} $\mathsf{considerPolyBelowWS}\; \cst{g}\; \mathsf{L}$ \textbf{else} $\mathsf{U}$ \\
\q\q\q\Q $\mathsf{U}$ $\Rightarrow$ $\mathsf{U}$ \\
\q\q\q \textbf{end} \\
\q\q\Q $\DB{m}\> \bar s$ $\Rightarrow$ \\
\q\q\q \textbf{if} $\mathsf{checkSubs}\; \bar s\; t$ \textbf{then} $\mathsf{L}$ \\
\q\q\q \textbf{else} \textbf{if} $\cst{g} \Succ \cst{ws} \mathrel\land \mathsf{checkArgs}\; t\; \bar s$ \textbf{then} $\mathsf{G}$ \\
\q\q\q \textbf{else} \textbf{if} $\cst{g} \Preceq \cst{ws} \mathrel\land \mathsf{checkArgs}\; s\; \bar t$ \textbf{then} $\mathsf{considerPoly}\; t\; s\; \mathsf{L}$ \\
\q\q\q \textbf{else} $\mathsf{U}$ \\
\q\q\Q $\lambda\> s'$ $\Rightarrow$ \\
\q\q\q \textbf{if} $\mathsf{checkSubs}\; [s']\; t$ \textbf{then} $\mathsf{L}$ \\
\q\q\q \textbf{else} \textbf{if} $\cst{g} \Succ \cst{ws} \mathrel\land \mathsf{checkArgs}\; t\; [s']$ \textbf{then} $\mathsf{G}$ \\
\q\q\q \textbf{else} \textbf{if} $\cst{g} \Preceq \cst{ws} \mathrel\land \mathsf{checkArgs}\; s\; \bar t$ \textbf{then} $\mathsf{considerPoly}\; t\; s\; \mathsf{L}$ \\
\q\q\q \textbf{else} $\mathsf{U}$ \\
\q\q \textbf{end} \\
[\jot]
\q\Q $\DB{n}\> \bar t$ $\Rightarrow$ \\
\q\q \textbf{if} $\mathsf{checkSubs}\; \bar t\; s$ \textbf{then} \\
\q\q\q $\mathsf{G}$ \\
\q\q \textbf{else match} $s$ \textbf{with} \\
\q\q\q $x\> \_$ $\Rightarrow$ $\mathsf{U}$ \\
\q\q\Q $\cst{f}\langle\_\rangle(\_)\> \bar s$ $\Rightarrow$ \\
\q\q\q \textbf{if} $\mathsf{checkSubs}\; \bar s\; t$ \textbf{then} $\mathsf{L}$ \\
\q\q\q \textbf{else if} $\cst{f} \Succ \cst{ws} \mathrel\land\mathsf{checkArgs}\; s\; \bar t$ \textbf{then} $\mathsf{L}$ \\
\q\q\q \textbf{else if} $\cst{f} \Preceq \cst{ws} \mathrel\land\mathsf{checkArgs}\; t\; \bar s$ \textbf{then} $\mathsf{considerPoly}\; t\; s\; \mathsf{G}$ \\
\q\q\q \textbf{else} $\mathsf{U}$ \\
\q\q\Q $\DB{m}\> \bar s$ $\Rightarrow$ \\
\q\q\q \textbf{if} $\mathsf{checkSubs}\; \bar s\; t$ \textbf{then} \\
\q\q\q\q $\mathsf{L}$ \\
\q\q\q \textbf{else if} $n > m$ \textbf{then} \\
\q\q\q\q \textbf{if} $\mathsf{checkArgs}\; t\; \bar s$ \textbf{then}
  $\mathsf{considerPoly}\;t \;s\; \mathsf{G}$
  \textbf{else} $\mathsf{U}$ \\
\q\q\q \textbf{else if} $n = m$ \textbf{then} \\
\q\q\q\q  $\mathsf{compareArgs}\; t\; []\; \bar t\; s\; []\; \bar s$ \\
\q\q\q \textbf{else} \\
\q\q\q\q \textbf{if} $\mathsf{checkArgs}\; s\; \bar t$ \textbf{then}
  $\mathsf{considerPoly}\;t \;s\; \mathsf{L}$
  \textbf{else} $\mathsf{U}$ \\
\q\q\Q $\lambda\> s'$ $\Rightarrow$ \\
\q\q\q \textbf{if} $\mathsf{checkSubs}\; [s']\; t$ \textbf{then} $\mathsf{L}$ \\
\q\q\q \textbf{else if} $\mathsf{checkArgs}\; t\; [s']$ \textbf{then} $\mathsf{G}$ \\
\q\q\q \textbf{else} $\mathsf{U}$ \\
\q\q \textbf{end} \\
[\jot]
\q\Q $\lambda\langle\upsilon\rangle\> t'$ $\Rightarrow$ \\
\q\q \textbf{if} $\mathsf{checkSubs}\; [t']\; s$ \textbf{then} \\
\q\q\q $\mathsf{G}$ \\
\q\q \textbf{else match} $s$ \textbf{with} \\
\q\q\q $x\> \_$ $\Rightarrow$ $\mathsf{U}$ \\
\q\q\Q $\cst{f}\langle\_\rangle(\_)\> \bar s$ $\Rightarrow$ \\
\q\q\q \textbf{if} $\mathsf{checkSubs}\; \bar s\; t$ \textbf{then} $\mathsf{L}$ \\
\q\q\q \textbf{else if} $\cst{f} \Succ \cst{ws} \mathrel\land\mathsf{checkArgs}\; s\; [t']$ \textbf{then} $\mathsf{L}$ \\
\q\q\q \textbf{else if} $\cst{f} \Preceq \cst{ws} \mathrel\land\mathsf{checkArgs}\; t\; \bar s$ \textbf{then} $\mathsf{considerPoly}\; t\; s\; \mathsf{G}$ \\
\q\q\q \textbf{else} $\mathsf{U}$ \\
\q\q\Q $\DB{m}\> \bar s$ $\Rightarrow$ \\
\q\q\q \textbf{if} $\mathsf{checkSubs}\; \bar s\; t$ \textbf{then} $\mathsf{L}$ \\
\q\q\q \textbf{else if} $\mathsf{checkArgs}\; s\; [t']$ \textbf{then} $\mathsf{L}$ \\
\q\q\q \textbf{else} $\mathsf{U}$ \\
\q\q\Q $\lambda\langle\tau\rangle\> s'$ $\Rightarrow$ \\
\q\q\q \textbf{if} $\mathsf{checkSubs}\; [s']\; t$ \textbf{then} \\
\q\q\q\q $\mathsf{L}$ \\
\q\q\q \textbf{else match} $\mathsf{compareTypes}\; \upsilon\; \tau$ \textbf{with} \\
\q\q\q\q $\mathsf{G}$ $\Rightarrow$ \textbf{if} $\mathsf{checkArgs}\; t\; [s']$ \textbf{then} $\mathsf{G}$ \textbf{else} $\mathsf{U}$ \\
\q\q\q\Q $\mathsf{E}$ $\Rightarrow$ $\mathsf{compareTerms}\; t'\; s'$ \\
\q\q\q\Q $\mathsf{L}$ $\Rightarrow$ \textbf{if} $\mathsf{checkArgs}\; s\; [t']$ \textbf{then} $\mathsf{L}$ \textbf{else} $\mathsf{U}$ \\
\q\q\q\Q $\mathsf{U}$ $\Rightarrow$ $\mathsf{U}$ \\
\q\q\q \textbf{end} \\
\q\q \textbf{end} \\
\q \textbf{end}
\end{quote}

\section{Optimized Algorithms}
\label{sec:optimized-algorithms}

Another improvement, embodied by a separate pair of algorithms, consists of
following L\"ochner's refinement approach
\cite{loechner-2006-kbo,loechner-2006-lpo}. For the standard KBO and LPO, his
comparison algorithms are respectively linear and quadratic in the size of the
input terms. The use of polynomials instead of integers in the $\lambda$KBO
makes the computation slightly more expensive, but we can nonetheless benefit
from tupling.

\subsection{\texorpdfstring{$\bm{\lambda}$KBO}{Lambda KBO}}
\label{ssec:lambda-kbo-optimized-algorithm}

The naive bidirectional algorithm for $\lambda$KBO is wasteful because it
recursively recomputes preterm weights. If $t = \cst{f}(\bar t)$, the subterm
$t_i$'s weight is computed first in the main function by the call to
$\mathsf{analyzeWeightDiff}$ and then possibly again in
$\mathsf{compareShapes}$, when $\mathsf{compareTerms}$ is called to break ties.
Although a factor of 2 might not sound particularly expensive, the factor is
higher for the subterms' subterms, their subsubterms, and so on. Thus, the
native algorithm is quadratic in the size of the input
preterms~\cite{loechner-2006-kbo}.

Our solution, inspired by L\"ochner \cite{loechner-2006-kbo}, consists of
interleaving the two passes:\ computing the weights and comparing the shapes.
The information for the passes is stored in a tuple. In this way, the subterms'
weights can be shared between the passes. At the end of the combined pass, we
can look at the tuple and determine what result to return.

First, we need to extend the lexicographic and componentwise extension functions
to thread through additional information---in our case, weights---returned by
the operator $\mathit{op}$ as the first component of a pair, the second
component being the comparison result.

\begin{quote}
\textbf{function} \,$\mathsf{lexExtData}\; \mathit{op}\; \bar b\; \bar a$ \,$:=$ \\
\q \textbf{match} $\bar b$, $\bar a$ \textbf{with} \\
\q\q $[]$, $[]$ $\Rightarrow$ $([], \mathsf{E})$ \\
\q\Q $b :: \bar b'$, $a :: \bar a'$ $\Rightarrow$ \\
\q\q \textbf{match} $\mathit{op}\; b\; a$ \textbf{with} \\
\q\q\q $(w{,}\; \mathsf{G})$ $\Rightarrow$ $([w], \mathsf{G})$ \\
\q\q\Q $(w{,}\; \mathsf{GE})$ $\Rightarrow$ \\
\q\q\q \textbf{let} $(\bar w{,}\; \mathit{cmp}) := \mathsf{lexExtData}\; \mathit{op}\; \bar b'\; \bar a'$ \textbf{in} \\
\q\q\q $(w :: \bar w{,}\; \mathsf{mergeWithGE}\; \mathit{cmp})$ \\
\q\q\Q $(w{,}\; \mathsf{E})$ $\Rightarrow$ \\
\q\q\q \textbf{let} $(\bar w{,}\; \mathit{cmp}) := \mathsf{lexExtData}\; \mathit{op}\; \bar b'\; \bar a'$ \textbf{in} \\
\q\q\q $(w :: \bar w{,}\; \mathit{cmp})$ \\
\q\q\Q $(w{,}\; \mathsf{LE})$ $\Rightarrow$ \\
\q\q\q \textbf{let} $(\bar w{,}\; \mathit{cmp}) := \mathsf{lexExtData}\; \mathit{op}\; \bar b'\; \bar a'$ \textbf{in} \\
\q\q\q $(w :: \bar w{,}\; \mathsf{mergeWithLE}\; \mathit{cmp})$ \\
\q\q\Q $(w{,}\; \mathsf{L})$ $\Rightarrow$ $([w]{,}\; \mathsf{L})$ \\
\q\q\Q $(w{,}\; \mathsf{U})$ $\Rightarrow$ $([w]{,}\; \mathsf{U})$ \\
\q\q \textbf{end} \\
\q \textbf{end}

\medskip

\textbf{function} \,$\mathsf{cwExtData}\; \mathit{op}$ \,$:=$ \\
\q $\mathsf{lexExtData}\; (\textbf{fun}\; b\; a \Rightarrow {}$ \\
\q\q \textbf{let} $(w{,}\; \mathit{cmp}) = \mathit{op}\; b\; a$ \textbf{in} \\
\q\q $(w{,}\; \mathsf{smooth}\; \mathit{cmp}))$
\end{quote}
In the above, the operator $::$ (``cons'') prepends an element to a list.

The auxiliary function $\mathsf{considerWeight}$ resembles the unoptimized
$\mathsf{compareTerms}$, but it uses its arguments $w$ and $\mathit{cmp}$
instead of recomputing them, where $\mathit{cmp}$ is the result of a shape
comparison.

\begin{quote}
\textbf{function} \,$\mathsf{considerWeight}\; w\; \mathit{cmp}$ \,$:=$ \\
\q $(w{,}$ \textbf{match} $\mathsf{analyzeWeightDiff}\; w$ \textbf{with} \\
\q\phantom{$(w{,}$ }\q $\mathsf{G}$ $\Rightarrow$ $\mathsf{G}$ \\
\q\phantom{$(w{,}$ }\Q $\mathsf{GE}$ $\Rightarrow$ $\mathsf{mergeWithGE}\; \mathit{cmp}$ \\
\q\phantom{$(w{,}$ }\Q $\mathsf{E}$ $\Rightarrow$ $\mathit{cmp}$ \\
\q\phantom{$(w{,}$ }\Q $\mathsf{LE}$ $\Rightarrow$ $\mathsf{mergeWithLE}\; \mathit{cmp}$ \\
\q\phantom{$(w{,}$ }\Q $\mathsf{L}$ $\Rightarrow$ $\mathsf{L}$ \\
\q\phantom{$(w{,}$ }\Q $\mathsf{U}$ $\Rightarrow$ $\mathsf{U}$ \\
\q\phantom{$(w{,}$ }\textbf{end}$)$
\end{quote}

The core of the code consists of two mutually recursive functions:\
$\mathsf{processArgs}$ and $\mathsf{processTerms}$. They compute weights and
compare shapes, returning pairs of the form $(w{,}\; \mathit{cmp})$, where
$\mathit{cmp}$ takes both the preterms' weights and their shapes into account. The
code for $\mathsf{processTerms}$ follows the structure of the unoptimized
$\mathsf{compareShape}$ but is instrumented to also compute weights. It calls
$\mathsf{processArgs}$ to compare argument lists. In $\mathsf{processArgs}$, the
weights computed as part of the lexicographic comparison are reused and extended
with any missing weights if the comparison ended before the end of the lists
(i.e., if $m < n$).

\begin{quote}
\textbf{function} \,$\mathsf{processArgs}\; \bar t_n\; \bar s_n$ \,$:=$ \\
\q \textbf{let} $(\bar w_m{,}\; \mathit{cmp}) := \mathsf{lexExtData}\; \mathsf{processTerms}\; \bar t_n\; \bar s_n$ \textbf{in} \\
\q $\mathsf{considerWeight}\; (\sum_{i=1}^m w_i + \sum_{i=m+1}^n (\cal W(t_i) - \cal W(s_i)))\; \mathit{cmp}$

\medskip

\textbf{function} \,$\mathsf{processTerms}\; t\; s$ \,$:=$ \\
\q \textbf{match} $t$, $s$ \textbf{with} \\
\q\q $y\> \bar t$, $x\> \bar s$ $\Rightarrow$ \\
\q\q \textbf{if} $y = x$ \textbf{then} \\
\q\q\q \textbf{if} $\bar t$ are steady \textbf{then} \\
\q\q\q\q \textbf{let} $(\bar w{,}\; \mathit{cmp}) := \mathsf{cwExtData}\; \mathsf{processTerms}\; \bar t\; \bar s$ \textbf{in} \\
\q\q\q\q $\mathsf{considerWeight}\; \bigl(\sum\nolimits_{i=1}^{|\bar t|} \mathbf{k}_{y,i}\, w_{i}\bigr)\; \mathit{cmp}$ \\
\q\q\q \textbf{else} \\
\q\q\q\q $\mathsf{considerWeight}\; (\cal W(t) - \cal W(s))\;\mathsf{U}$ \\
\q\q \textbf{else} \\
\q\q\q $\mathsf{considerWeight}\; (\cal W(t) - \cal W(s))\;\mathsf{U}$ \\
\q\Q $y\> \_$, $\_$ $\mid$ $\_$, $x\> \_$ $\Rightarrow$ $\mathsf{considerWeight}\; (\cal W(t) - \cal W(s))\;\mathsf{U}$ \\
\q\Q $\lambda\langle\upsilon\rangle\> t'$, $\lambda\langle\tau\rangle\> s'$ $\Rightarrow$ \\
\q\q \textbf{match} $\mathsf{compareTypes}\; \upsilon\; \tau$ \textbf{with} \\
\q\q\q $\mathsf{E}$ $\Rightarrow$ $\mathsf{processTerms}\; t'\; s'$ \\
\q\q\Q $\mathit{cmp}$ $\Rightarrow$ $\mathsf{considerWeight}\; (\cal W(t') - \cal W(s'))\; \mathit{cmp}$ \\
\q\q \textbf{end} \\
\q\Q $\lambda\> \_$, $\_$ $\Rightarrow$ $\mathsf{considerWeight}\; (\cal W(t) - \cal W(s))\; (\mathsf{considerPoly}\; t\; s\; \mathsf{G})$ \\
\q\Q $\DB{n}\> \_$, $\lambda\> \_$ $\Rightarrow$ $\mathsf{considerWeight}\; (\cal W(t) - \cal W(s))\; (\mathsf{considerPoly}\; t\; s\; \mathsf{L})$ \\
\q\Q $\DB{n}\> \bar t$, $\DB{m}\>\bar s$ $\Rightarrow$ \\
\q\q \textbf{if\/} $n > m$ \textbf{then} $\mathsf{considerWeight}\; (\cal W(t) - \cal W(s))\; (\mathsf{considerPoly}\; t\; s\; \mathsf{G})$ \\
\q\q \textbf{else if\/} $n < m$ \textbf{then} $\mathsf{considerWeight}\; (\cal W(t) - \cal W(s))\; (\mathsf{considerPoly}\; t\; s\; \mathsf{L})$ \\
\q\q \textbf{else} $\mathsf{processArgs}\; \bar t\; \bar s$ \\
\q\Q $\DB{n}\> \_$, $\cst{f}\langle\_\rangle(\_)\>\_$ $\Rightarrow$ $\mathsf{considerWeight}\; (\cal W(t) - \cal W(s))\; (\mathsf{considerPoly}\; t\; s\; \mathsf{G})$ \\
\q\Q $\cst{g}\langle\bar\upsilon\rangle(\bar w)\>\bar t$, $\cst{f}\langle\bar\tau\rangle(\bar u)\>\bar s$ $\Rightarrow$ \\
\q\q \textbf{match} $\mathsf{compareSyms}\; \cst{g}\; \cst{f}$ \textbf{with} \\
\q\q\q $\mathsf{E}$ $\Rightarrow$ \\
\q\q\q \textbf{match} $\mathsf{lexExt}\; \mathsf{compareTypes}\; \bar\upsilon\; \bar\tau$ \textbf{with} \\
\q\q\q\q $\mathsf{E}$ $\Rightarrow$ $\mathsf{processArgs}\; (\bar w \cdot \bar t)\; (\bar u \cdot \bar s)$ \\
\q\q\q\Q $\mathit{cmp}$ $\Rightarrow$ $\mathsf{considerWeight}\; (\cal W(t) - \cal W(s))\; (\mathsf{considerPoly}\; t\; s\; \mathit{cmp})$ \\
\q\q\q \textbf{end} \\
\q\q\Q $\mathit{cmp}$ $\Rightarrow$ $\mathsf{considerWeight}\; (\cal W(t) - \cal W(s))\; (\mathsf{considerPoly}\; t\; s\; \mathit{cmp})$ \\
\q\q \textbf{end} \\
\q\Q $\cst{g}\langle\_\rangle(\_)\>\_$, $\_$ $\Rightarrow$ $\mathsf{considerWeight}\; (\cal W(t) - \cal W(s))\; (\mathsf{considerPoly}\; t\; s\; \mathsf{L})$ \\
\q \textbf{end}
\end{quote}
When calling $\mathsf{processTerms}$ to compare two preterms, we would normally
ignore the $w$ component of the result and only consider $\mathit{cmp}$, which
should be equal to what the untupled $\mathsf{compareTerms}$ would return.

One last point to discuss is the representation of polynomials. In the standard
KBO, multisets of variables must be compared. These can be seen as polynomials
of degree~1. L\"ochner's approach for the KBO variable check is to use an array
indexed by a finite variable set $X$. Clearly, this technique does not scale to
polynomials of arbitrarily high degrees. Instead of arrays, we can use maps or
hash tables indexed by sorted lists of indeterminates. With a reasonable map
implementation, this would replace an $\mathrm{O}(n)$ complexity with
$\mathrm{O}(n \log n)$, where $n = \left|s\right| + \left|t\right|$, the size of
the input preterms.

One of L\"ochner's ideas that also applies in our setting is to maintain two
counters indicating how many monomials are nonnegative or nonpositive in the
current polynomial expressed in standard form. These counters must be updated
whenever the map or hash table is modified. The two calls to
$\mathsf{surelyNonneg}$ in $\mathsf{analyzeWeightDiff}$ can then be replaced by
two conditions that each check whether a counter is 0.

\subsection{\texorpdfstring{$\bm{\lambda}$LPO}{Lambda LPO}}
\label{ssec:lambda-lpo-optimized-algorithm}

The naive bidirectional algorithm has exponential complexity because of the
overlapping computations of $\mathsf{checkArgs}$ and $\mathsf{checkSubs}$. Our
solution, again inspired by L\"ochner \cite{loechner-2006-lpo}, consists of
postponing the checks and avoiding redundant comparisons. Specifically, our
algorithm below draws inspiration from L\"ochner's $\mathsf{clpo}_{6}$.

We start with a simple auxiliary function:

\begin{quote}
\textbf{function} $\mathsf{flip}\; \mathit{cmp}$ $:=$ \\
\q \textbf{match} $\mathit{cmp}$ \textbf{with} \\
\q\q $\textsf{G}$ $\Rightarrow$ $\textsf{L}$ \\
\q\Q $\textsf{GE}$ $\Rightarrow$ $\textsf{LE}$ \\
\q\Q $\textsf{E}$ $\Rightarrow$ $\textsf{E}$ \\
\q\Q $\textsf{LE}$ $\Rightarrow$ $\textsf{GE}$ \\
\q\Q $\textsf{L}$ $\Rightarrow$ $\textsf{G}$ \\
\q\Q $\textsf{U}$ $\Rightarrow$ $\textsf{U}$ \\
\q \textbf{end}
\end{quote}

\begin{sloppypar}
The following six functions are mutually recursive. The main function is
called $\mathsf{compareTerms}$, as in Section~\ref{sec:naive-algorithms}.
\end{sloppypar}

\begin{quote}

\textbf{function} $\mathsf{checkSubs}\; \bar t\; s$ $:=$ \\
\q $\exists i.\; \mathsf{compareTerms}\; t_i\; s \in \{\mathsf{G}, \mathsf{GE}, \mathsf{E}\}$

\medskip

\textbf{function} $\mathsf{compareSubsBothWays}\; t\; \bar t\; s\; \bar s$ $:=$ \\
\q \textbf{if} $\mathsf{checkSubs}\; \bar t\; s$ \textbf{then} $\mathsf{G}$ \\
\q \textbf{else if} $\mathsf{checkSubs}\; \bar s\; t$ \textbf{then} $\mathsf{L}$ \\
\q \textbf{else} $\mathsf{U}$

\medskip

\textbf{function} $\mathsf{compareRest}\; t\; \bar s$ $:=$ \\
\q \textbf{match} $\bar s$ \textbf{with} \\
\q\q $[]$ $\Rightarrow$ $\mathsf{G}$ \\
\q\Q $s :: \bar s'$ $\Rightarrow$ \\
\q\q \textbf{match} $\mathsf{compareTerms}\; t\; s$ \textbf{with} \\
\q\q\q $\mathsf{G}$ $\Rightarrow$ $\mathsf{compareRest}\; t\; \bar s'$ \\
\q\q\Q $\mathsf{E}$ $\mid$ $\mathsf{LE}$ $\mid$ $\mathsf{L}$ $\Rightarrow$ $\mathsf{L}$ \\
\q\q\Q $\mathsf{GE}$ $\mid$ $\mathsf{U}$ $\Rightarrow$ \textbf{if} $\mathsf{checkSubs}\; \bar s'\; t$ \textbf{then} $\mathsf{L}$ \textbf{else} $\mathsf{U}$ \\
\q\q \textbf{end} \\
\q \textbf{end}

\medskip

\textbf{function} $\mathsf{compareRegularArgs}\; t\; \bar t\; s\; \bar s$ $:=$ \\
\q \textbf{match} $\bar t$, $\bar s$ \textbf{with} \\
\q\q $[]$, $[]$ $\Rightarrow$ $\mathsf{E}$ \\
\q\Q $t_1 :: \bar t'$, $s_1 :: \bar s'$ $\Rightarrow$ \\
\q\q \textbf{match} $\mathsf{compareTerms}\; t_1\; s_1$ \textbf{with} \\
\q\q\q $\mathsf{G}$ $\Rightarrow$ $\mathsf{compareRest}\; t\; \bar s'$ \\
\q\q\Q $\mathsf{GE}$ $\Rightarrow$ $\mathsf{mergeWithGE}\; (\mathsf{compareRegularArgs}\; t\; \bar t'\; s\; \bar s')$ \\
\q\q\Q $\mathsf{E}$ $\Rightarrow$ $\mathsf{compareRegularArgs}\; t\; \bar t'\; s\; \bar s'$ \\
\q\q\Q $\mathsf{LE}$ $\Rightarrow$ $\mathsf{mergeWithLE}\; (\mathsf{compareRegularArgs}\; t\; \bar t'\; s\; \bar s')$ \\
\q\q\Q $\mathsf{L}$ $\Rightarrow$ $\mathsf{flip}\; (\mathsf{compareRest}\; s\; \bar t')$ \\
\q\q\Q $\mathsf{U}$ $\Rightarrow$ $\mathsf{compareSubsBothWays}\; t\; \bar t'\; s\; \bar s'$ \\
\q\q \textbf{end} \\
\q \textbf{end}

\medskip

\textbf{function} $\mathsf{compareArgs}\; t\; \bar v\; \bar t\; s\; \bar u\; \bar s$ $:=$ \\
\q \textbf{match} $\bar v$, $\bar u$ \textbf{with} \\
\q\q $[]$, $[]$ $\Rightarrow$ $\mathsf{compareRegularArgs}\; t\; \bar t\; s\; \bar s$ \\
\q\Q $v_1 :: \bar v'$, $u_1 :: \bar u'$ $\Rightarrow$ \\
\q\q \textbf{match} $\mathsf{compareTerms}\; v_1\; u_1$ \textbf{with} \\
\q\q\q $\mathsf{G}$ $\Rightarrow$ $\mathsf{compareRest}\; t\; \bar s$ \\
\q\q\Q $\mathsf{GE}$ $\Rightarrow$ $\mathsf{mergeWithGE}\; (\mathsf{compareArgs}\; t\; \bar v'\; \bar t\; s\; \bar u'\; \bar s)$ \\
\q\q\Q $\mathsf{E}$ $\Rightarrow$ $\mathsf{compareArgs}\; t\; \bar v'\; \bar t\; s\; \bar u'\; \bar s$ \\
\q\q\Q $\mathsf{LE}$ $\Rightarrow$ $\mathsf{mergeWithLE}\; (\mathsf{compareArgs}\; t\; \bar v'\; \bar t\; s\; \bar u'\; \bar s)$ \\
\q\q\Q $\mathsf{L}$ $\Rightarrow$ $\mathsf{flip}\; (\mathsf{compareRest}\; s\; \bar t)$ \\
\q\q\Q $\mathsf{U}$ $\Rightarrow$ $\mathsf{compareSubsBothWays}\; t\; \bar t\; s\; \bar s$ \\
\q\q \textbf{end} \\
\q \textbf{end}

\medskip

\textbf{function} $\mathsf{compareTerms}\; t\; s$ $:=$ \\
\q \textbf{match} $t$, $s$ \textbf{with} \\
\q\q $y\> \bar t$, $x\> \bar s$ $\Rightarrow$ \textbf{if} $y = x \mathrel\land \bar t$ are steady \textbf{then} $\mathsf{cwExt}\; \mathsf{compareTerms}\; \bar t\; \bar s$ \textbf{else} $\mathsf{U}$ \\
\q\Q $y\> \_$, $\cst{f}\langle\_\rangle(\_)\> \bar s$ $\mid$ $y\> \_$, $\DB{m}\> \bar s$ $\Rightarrow$ \textbf{if} $\mathsf{checkSubs}\; \bar s\; t$ \textbf{then} $\mathsf{L}$ \textbf{else} $\mathsf{U}$ \\
\q\Q $y\> \_$, $\lambda\> s'$ $\Rightarrow$ \textbf{if} $\mathsf{checkSubs}\; [s']\; t$ \textbf{then} $\mathsf{L}$ \textbf{else} $\mathsf{U}$ \\
\q\Q $\cst{g}\langle\_\rangle(\_)\> \bar t$, $x\> \_$ $\Rightarrow$ \textbf{if} $\mathsf{checkSubs}\; \bar t\; s$ \textbf{then} $\mathsf{G}$ \textbf{else} $\mathsf{U}$ \\
\q\Q $\cst{g}\langle\bar\upsilon\rangle(\bar w)\> \bar t$, $\cst{f}\langle\bar\tau\rangle(\bar u)\> \bar s$ $\Rightarrow$ \\
\q\q \textbf{match} $\mathsf{compareSyms}\; \cst{g}\; \cst{f}$ \textbf{with} \\
\q\q\q $\mathsf{G}$ $\Rightarrow$ $\mathsf{considerPolyBelowWS}\; \cst{g}\; t\; s\; (\mathsf{compareRest}\; t\; \bar s)$ \\
\q\q\Q $\mathsf{E}$ $\Rightarrow$ \\
\q\q\q \textbf{match} $\mathsf{lexExt}\; \mathsf{compareTypes}\; \bar\upsilon\; \bar\tau$ \textbf{with} \\
\q\q\q\q $\mathsf{G}$ $\Rightarrow$ $\mathsf{considerPolyBelowWS}\; \cst{g}\; t\; s\; (\mathsf{compareRest}\; t\; \bar s)$ \\
\q\q\q\Q $\mathsf{E}$ $\Rightarrow$ $\mathsf{compareArgs}\; t\; \bar w\; \bar t\; s\; \bar u\; \bar s$ \\
\q\q\q\Q $\mathsf{L}$ $\Rightarrow$ $\mathsf{considerPolyBelowWS}\; \cst{g}\; t\; s\; (\mathsf{flip}\; (\mathsf{compareRest}\; s\; \bar t))$ \\
\q\q\q\Q $\mathsf{U}$ $\Rightarrow$ $\mathsf{compareSubsBothWays}\; t\; \bar t\; s\; \bar s$ \\
\q\q\q \textbf{end} \\
\q\q\Q $\mathsf{L}$ $\Rightarrow$ $\mathsf{considerPolyBelowWS}\; \cst{g}\; t\; s\; (\mathsf{flip}\; (\mathsf{compareRest}\; t\; \bar s))$ \\
\q\q\Q $\mathsf{U}$ $\Rightarrow$ $\mathsf{compareSubsBothWays}\; t\; \bar t\; s\; \bar s$ \\
\q\q \textbf{end} \\
\q\Q $\cst{g}\langle\_\rangle(\_)\> \bar t$, $\DB{m}\> \bar s$ $\Rightarrow$ \\
\q\q \textbf{if} $\cst{f} \Succ \cst{ws}$ \textbf{then} $\mathsf{compareRest}\; t\; \bar s$ \\
\q\q \textbf{else} $\mathsf{considerPoly}\; t\; s\; (\mathsf{flip}\; (\mathsf{compareRest}\; s\; \bar t))$ \\
\q\Q $\cst{g}\langle\_\rangle(\_)\> \bar t$, $\lambda\> s'$ $\Rightarrow$ \\
\q\q \textbf{if} $\cst{f} \Succ \cst{ws}$ \textbf{then} $\mathsf{compareRest}\; t\; [s']$ \\
\q\q \textbf{else} $\mathsf{considerPoly}\; t\; s\; (\mathsf{flip}\; (\mathsf{compareRest}\; s\; \bar t))$ \\
\q\Q $\DB{n}\> \bar t$, $x\> \_$ $\Rightarrow$ \textbf{if} $\mathsf{checkSubs}\; \bar t\; s$ \textbf{then} $\mathsf{G}$ \textbf{else} $\mathsf{U}$ \\
\q\Q $\DB{n}\> \bar t$, $\cst{f}\langle\_\rangle(\_)\> \bar s$ $\Rightarrow$ \\
\q\q \textbf{if} $\cst{f} \Succ \cst{ws}$ \textbf{then} $\mathsf{flip}\; (\mathsf{compareRest}\; s\; \bar t)$ \\
\q\q \textbf{else} $\mathsf{considerPoly}\; t\; s\; (\mathsf{compareRest}\; t\; \bar s)$ \\
\q\Q $\DB{n}\> \bar t$, $\DB{m}\> \bar s$ $\Rightarrow$ \\
\q\q \textbf{if} $n > m$ \textbf{then} $\mathsf{considerPoly}\; t\; s\; (\mathsf{compareRest}\; t\; \bar s)$ \\
\q\q \textbf{else if} $n = m$ \textbf{then} $\mathsf{compareRegularArgs}\; t\; \bar t\; s\; \bar s$ \\
\q\q \textbf{else} $\mathsf{considerPoly}\; t\; s\; (\mathsf{flip}\; (\mathsf{compareRest}\; s\; \bar t))$ \\
\q\Q $\DB{n}\> \_$, $\lambda\> s'$ $\Rightarrow$ $\mathsf{compareRest}\; t\; [s']$ \\
\q\Q $\lambda\> t'$, $x\> \_$ $\Rightarrow$ \textbf{if} $\mathsf{checkSubs}\; [t']\; s$ \textbf{then} $\mathsf{G}$ \textbf{else} $\mathsf{U}$ \\
\q\Q $\lambda\> t'$, $\cst{f}\langle\_\rangle(\_)\> \bar s$ $\Rightarrow$ \\
\q\q \textbf{if} $\cst{f} \Succ \cst{ws}$ \textbf{then} $\mathsf{flip}\; (\mathsf{compareRest}\; s\; [t'])$ \\
\q\q \textbf{else} $\mathsf{considerPoly}\; t\; s\; (\mathsf{compareRest}\; t\; \bar s)$ \\
\q\Q $\lambda\> t'$, $\DB{m}\> \_$ $\Rightarrow$ $\mathsf{flip}\; (\mathsf{compareRest}\; s\; t')$ \\
\q\Q $\lambda\langle\upsilon\rangle\> t'$, $\lambda\langle\tau\rangle\> s'$ $\Rightarrow$ \\
\q\q \textbf{match} $\mathsf{compareTypes}\; \upsilon\; \tau$ \textbf{with} \\
\q\q\q $\mathsf{G}$ $\Rightarrow$ $\mathsf{compareRest}\; t\; [s']$ \\
\q\q\Q $\mathsf{E}$ $\Rightarrow$ $\mathsf{compareTerms}\; t'\; s'$ \\
\q\q\Q $\mathsf{L}$ $\Rightarrow$ $\mathsf{flip}\; (\mathsf{compareRest}\; s\; [t'])$ \\
\q\q\Q $\mathsf{U}$ $\Rightarrow$ $\mathsf{compareSubsBothWays}\; t\; [t']\; s\; [s']$ \\
\q\q \textbf{end} \\
\q \textbf{end}
\end{quote}

\section{Conclusion}
\label{sec:conclusion}

We defined two new term orders, $\lambda$KBO and $\lambda$LPO, for use with
$\lambda$-superposition. We expect these new order to improve Zipperposition's
performance, measured as both proving time and success rate.
Some of the ideas might also apply to $\lambda$-free superposition
\cite{bentkamp-et-al-2021-lfhosup} and combinatory superposition
\cite{bhayat-reger-2020-combsup}: Despite working on logics devoid of
$\lambda$-abstractions, these proof calculi contain axiom (\textsc{Ext}) and
could benefit from the implicit $\eta$-expansion that makes its positive literal
maximal.

\paragraph{Acknowledgment.}

Petar Vukmirovi\'c discussed ideas with us
and provided some of the examples in Sect.~\ref{sec:examples}.
Ahmed Bhayat suggested textual improvements.
We thank them for their help.

Bentkamp and Blanchette's research has received funding from the 
European Research Council
(ERC, Matryoshka, 713999 and Nekoka, 101083038)
Blanchette's research has received
funding from the Netherlands Organization for Scientific Research (NWO) under the Vidi
program (project No. 016.Vidi.189.037, Lean Forward).
Hetzenberger's research has received funding from the 
European Research Council (ERC, ARTIST, 101002685).

Views and opinions expressed are however those of the authors only and do not necessarily reflect those of the European Union or the European Research Council. Neither the European Union nor the granting authority can be held responsible for them.

We have used artificial intelligence tools for textual editing.

\bibliographystyle{splncs04}
\bibliography{ms}

\end{document}